\documentclass[11pt]{amsart}
\pdfoutput=1

\usepackage{amssymb}
\usepackage{amsthm}
\usepackage{amsfonts}
\usepackage{amsmath}
\usepackage{pmboxdraw}
\usepackage{verbatim} 
\usepackage{graphicx}
\usepackage{color}
\usepackage[colorlinks=true, citecolor=blue, filecolor=black, linkcolor=black, urlcolor=black]{hyperref}
\usepackage{cite}
\usepackage[normalem]{ulem}
\usepackage{subcaption}
\usepackage{bbm}
\usepackage{bm}
\usepackage{mathtools}
\usepackage{todonotes}
\usepackage{enumitem}
\usepackage{multirow,makecell}  
\usepackage{tablefootnote}

\usepackage{xparse}

\ExplSyntaxOn
\NewDocumentCommand{\MeijerG}{smmmm}
 {
  \IfBooleanTF{#1}
   {
    \vic_meijerg:nnnnnn { #2 } { #3 } { #4 } { #5 } { small } { }
   }
   {
    \vic_meijerg:nnnnnn { #2 } { #3 } { #4 } { #5 } { } { \; }
   }
 }

\seq_new:N \l__vic_meijerg_args_in_seq
\seq_new:N \l__vic_meijerg_args_out_seq

\cs_new_protected:Nn \vic_meijerg:nnnnnn
 {
  \seq_set_split:Nnn \l__vic_meijerg_args_in_seq { | } { #3 }
  \seq_clear:N \l__vic_meijerg_args_out_seq  
  \seq_map_inline:Nn \l__vic_meijerg_args_in_seq
   {
    \seq_put_right:Nn \l__vic_meijerg_args_out_seq
     {
      \begin{#5matrix} ##1 \end{#5matrix}
     }
   }
  G\sp{#1}\sb{#2}
  \left(
  \seq_use:Nn \l__vic_meijerg_args_out_seq { #6\middle|#6 }
  #6\middle|#6
  #4
  \right)
 }
\ExplSyntaxOff

\newcommand{\RN}[1]{%
	\textup{\uppercase\expandafter{\romannumeral#1}}%
}

\def\bfs{\boldsymbol}

\def\pa{\partial}

\def\Re{ \mathrm{Re}}
\def\Im{ \mathrm{Im}}

\def\NN{\mathcal{N}}

\def\SS{\mathcal{S}}

\def\C{\mathbb{C}}

\def\N{\mathbb{N}}

\def\R{\mathbb{R}}

\newcommand{\Pf}{{\textup{Pf}}}
\newcommand{\erfc}{\operatorname{erfc}}
\newcommand{\erf}{\operatorname{erf}}

\newcommand{\Li}{\operatorname{Li}}
\newcommand{\sgn}{\operatorname{sgn}}
\newcommand{\Var}{\operatorname{Var}}
\newcommand{\Cov}{\operatorname{Cov}}

\newcommand{\im}{\operatorname{Im}}

\theoremstyle{plain}
\newtheorem*{thm*}{Theorem}
\newtheorem{thm}{Theorem}[section]
\newtheorem{lem}[thm]{Lemma}
\newtheorem{cor}[thm]{Corollary}
\newtheorem{prop}[thm]{Proposition}
\newtheorem*{prop*}{Proposition}
\newtheorem*{lem*}{Lemma}
\newtheorem{rem}[thm]{Remark}
\newtheorem{ex}[thm]{Example}
\newtheorem{conj}[thm]{Conjecture}
\newtheorem{defi}[thm]{Definition}

\theoremstyle{definition}
\newtheorem*{eg*}{Example}
\newtheorem*{egs*}{Examples}
\newtheorem*{Q*}{Question}

\theoremstyle{remark}
\newtheorem*{rmk*}{Remark}
\newtheorem*{rmks*}{Remarks}

\newcommand{\abs}[1]{\lvert#1\rvert}
\numberwithin{equation}{section}

\begin{document}
\title[Counting statistics of the real and symplectic Ginibre ensemble]{Universality in the number variance and counting statistics of the real and symplectic Ginibre ensemble
}
\author{Gernot Akemann}
\address{Faculty of Physics, Bielefeld University, P.O. Box 100131, 33501 Bielefeld, Germany}
\email{akemann@physik.uni-bielefeld.de}

\author{Sung-Soo Byun}
\address{Department of Mathematical Sciences and Research Institute of Mathematics, Seoul National University, Seoul 151-747, Republic of Korea}
\email{sungsoobyun@snu.ac.kr}

\author{Markus Ebke}
\address{Department of Mathematics, Friedrich-Alexander-Universit\"at Erlangen-N\"urnberg, Cauerstrasse 11, 91058 Erlangen, Germany}
\email{markus.ebke@fau.de}

\author{Grégory Schehr}
\address{Sorbonne Université, Laboratoire de Physique Théorique et Hautes Energies, CNRS UMR 7589, 4 Place Jussieu, 75252 Paris Cedex 05, France}
\email{schehr@lpthe.jussieu.fr}



\date{\today}

\begin{abstract}
In this article, we compute and compare the statistics of the number of eigenvalues in a centred disc of radius $R$ in all three Ginibre ensembles.
We determine the mean and variance as functions of $R$ in the vicinity of the origin, where the real and symplectic ensembles exhibit respectively an additional attraction to or repulsion from the real axis, leading to different results.
In the large radius limit, all three ensembles coincide and display a universal bulk behaviour of $O(R^2)$ for the mean, and $O(R)$ for the variance.  
We present detailed conjectures for the bulk and edge scaling behaviours of the real Ginibre ensemble, having real and complex eigenvalues. 
For the symplectic ensemble we can go beyond the Gaussian case (corresponding to the Ginibre ensemble) and prove the universality of the full counting statistics both in the bulk and at the edge of the spectrum for rotationally invariant potentials, extending a recent work which considered the mean and the variance. 
This statistical behaviour coincides with the universality class of the complex Ginibre ensemble, which has been shown to be associated with the ground state of non-interacting fermions in a two-dimensional rotating harmonic trap. 
All our analytical results and conjectures are corroborated by numerical simulations. 
\end{abstract}

\maketitle

\section{Introduction}
There is currently a wide interest in the statistics of point processes, with applications ranging from ecology \cite{law2009ecological}, finance \cite{bauwens2009modelling}, signal processing \cite{paiva2009reproducing}, statistical and condensed-matter physics \cite{torquato2003local}, quantum optics \cite{bardenet2022point} all the way to machine learning \cite{kulesza2012determinantal}. To characterize such point processes, a useful and common tool is the so called full counting statistics (FCS). If we denote by $\mathcal{S}$ the full system and by $\Omega \subset \mathcal{S}$ a sub-system of $\mathcal{S}$, the FCS addresses the fluctuations of the total number of points $\NN_{\Omega}$ inside the region $\Omega$. The simplest example is probably the Poisson point processes where, in the homogeneous case, the points are just uniformly and independently distributed over $\mathcal{S}$, with a finite density $\rho$~\cite{daley2003introduction}. In this case, $\NN_{\Omega}$ is simply a Poisson random variable of parameter $\rho |\Omega|$, where $|\Omega|$ denotes the size of the sub-domain $\Omega$. A remarkable property of the Poisson point processes  is that the mean and the variance of $\NN_{\Omega}$ do coincide, i.e., $\mathbb{E}[\NN_{\Omega}] = \Var(\NN_\Omega)$, indicating that the fluctuations are very large here. Hence the mean does not carry so much information about the random variable~$\NN_{\Omega}$. 

In fact, in many relevant systems in mathematics and physics the corresponding point processes are much more rigid, in the sense that $ \Var(\NN_\Omega)/\mathbb{E}[\NN_{\Omega}] \to 0 $ as $|\Omega| \to \infty$. Such systems are generically called \emph{hyperuniform} \cite{torquato2003local,scardicchio2009statistical} and they have drawn a lot of attention during the last twenty years \cite{torquato2018hyperuniform}.
A prominent example of such hyperuniform systems is provided by determinantal point processes (DPP), which include for instance complex random matrix ensembles, non-interacting trapped fermions \cite{dean2016noninteracting} and many others, see e.g. \cite{johansson2005random,HKPV2006} and references therein. 
In the context of fermions, FCS turns out to be particularly interesting since, in many cases, the variance $\Var(\NN_\Omega)$ is proportional to the entanglement entropy between $\Omega$ and its complement $\overline{\Omega}$ in $\mathcal{S}$~\cite{calabrese2012exact,calabrese2011entanglement,calabrese2015random,song2011entanglement,klich2009quantum}. This property is all the more interesting since entanglement entropy is very difficult to measure experimentally, whereas FCS is much more accessible for fermionic systems. In particular the recently developed quantum Fermi microscopes \cite{parsons2015site,haller2015single,cheuk2015quantum} allow to take ``pictures'' of the positions of the fermions, in particular in two-dimensional systems where most of these experiments are carried out. In fact, it was recently realised that several non-interacting trapped fermionic systems, in one and two dimensions, are related to Gaussian random matrix ensembles with complex entries~\cite{dean2019noninteracting}. In particular, in two-dimensions, it was shown that the positions of fermions in a two-dimensional \emph{rotating} harmonic trap are in one-to-one correspondence with the eigenvalues of Gaussian random matrices belonging to the complex Ginibre ensemble \cite{lacroix2019rotating}. This physical situation is akin to the well known \emph{lowest Landau level} problem of electrons in a plane and in the presence of a perpendicular magnetic field \cite{garcia2002critical,cooper2012quantum,forrester1999exact}. This connection with the physics of the lowest Landau levels
has naturally motivated the study of the FCS in the complex Ginibre ensemble, denoted here as GinUE \cite{charles2020entanglement,lacroix2019rotating,akemann2022universality,charlier2021large,CL22,byun2022characteristic,shirai2006large,ameur2022disk,ameur2022exponential,fenzl2022precise} (for a recent review see \cite{byun2022progress}), as well as some natural extensions of it, including the higher Landau levels \cite{kulkarni2021multilayered,kulkarni2022density,smith2022counting,MR3340195}, related to the so-called poly-analytic Ginibre ensemble \cite{haimi2013polyanalytic}. We also refer to \cite{marino2014phase,marino2016number} and references therein for earlier work on the counting statistics of Hermitian random matrix ensembles and its applications to one-dimensional systems of trapped fermions.
For such models, analytical progress is possible thanks to the fact that the underlying point processes are DPPs, for which very powerful analytical tools are available, already for a finite number of points $N$, see \cite{johansson2005random,HKPV2006}.

From a mathematical perspective it is natural to ask how robust the predictions of the complex Ginibre ensemble are, i.e. the question of their universality. There are at least two directions how this question of universality can be addressed.
First, one can replace the Gaussian weight in the space of matrix elements with a more general potential, which leads to normal random matrix ensembles, see e.g. \cite[Section 5]{byun2022progress}. 
A second direction is to consider the ``cousins'' of the GinUE, with real (GinOE) or quaternion matrix elements (GinSE) with independent Gaussian distributions, without further symmetry constraints, cf. \cite{byun2023progress}. 
Both ensembles lead to Pfaffian point-processes instead, which are technically much harder to study than DPP. 

In this paper, we will consider both of these universality questions. Currently, it is not clear how to directly relate the ground state wave function of a Schr\"odinger equation with an anharmonic potential to a complex ensemble of random matrices with a non-Gaussian distribution in two dimensions. Nonetheless, there have been several examples in one dimension, see e.g. \cite{MR3800894, lacroix2017statistics,smith2021full} as well as the review \cite{dean2019noninteracting}, where such a connection could be established. Note also that anharmonic potentials have recently been studied in the context of the quantum Hall effect~\cite{oblak2023anisotropic}. Besides, it is not known if the Pfaffian point processes obtained for the GinOE and GinSE enjoy a direct quantum mechanical realisation. However, establishing a mathematically rigorous universality statement within all three ensembles, including normal matrices, gives rise to the expectation that the FCS might also be robust in more general quantum Hamiltonians. 

It has been shown previously that the local statistics agrees for all complex eigenvalue correlation functions in all three Ginibre ensembles, away from the real line, both in the bulk and at the edge of the support of the spectrum (see e.g. \cite{MR2530159,akemann2019universal}). Here, the support is determined by the circular law on a disc of radius $\sqrt{N}$. (In the latter case, we sometimes rescale the model so that the radius of the circular law is $1$.)
This does not imply though, that the FCS statistics also agrees, when considering a domain $\Omega$ of macroscopic extent. Due to the (partial) rotational invariance in these three ensembles, and due to the map to a rotating trap in two-dimensions, it is most natural to choose a centred disc $D_R$ of radius $R$ for the domain $\Omega$. 
Here, one has to distinguish three different regimes (see Figure~\ref{Fig_Scaling Regimes} for illustration): the \emph{origin} regime, when the radius $R= {O}(1)$\footnote{This regime was called deep bulk in \cite{lacroix2019rotating}.}, the \emph{bulk} regime when $1\ll R<\sqrt{N}$, and the \emph{edge} regime when $R\approx \sqrt{N}$ is in the vicinity of the edge of support. 
When taking the limit $R \gg 1$ contact can be made between the origin and bulk limit. 
\begin{figure}[t]
	\begin{subfigure}{0.32\textwidth}
		\begin{center}
			\includegraphics[width=\textwidth]{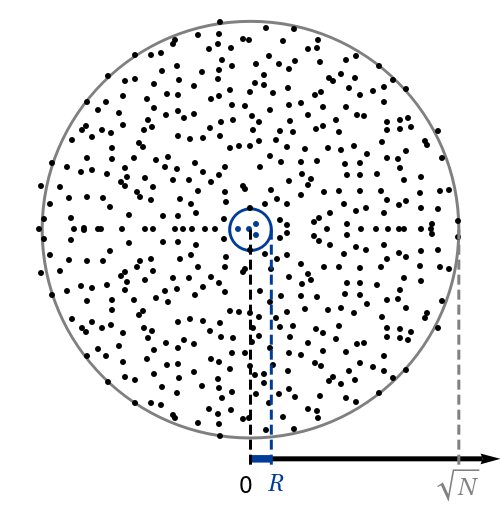}
		\end{center}
		\subcaption{Origin regime: $R = O(1)$}
	\end{subfigure}
	\begin{subfigure}{0.32\textwidth}
		\begin{center}
			\includegraphics[width=\textwidth]{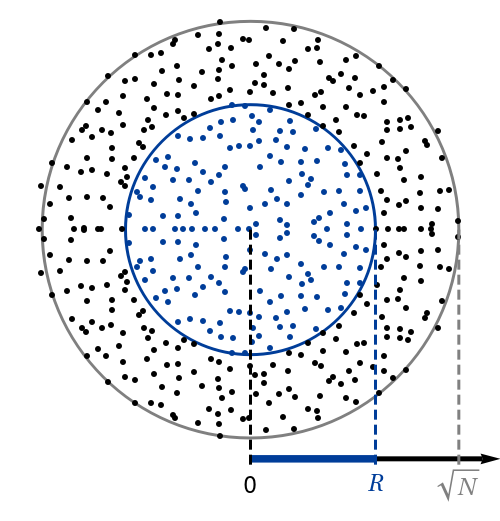}
		\end{center}
		\subcaption{Bulk regime: $R = O(\sqrt{N})$}
	\end{subfigure}
	\begin{subfigure}{0.32\textwidth}
		\begin{center}
			\includegraphics[width=\textwidth]{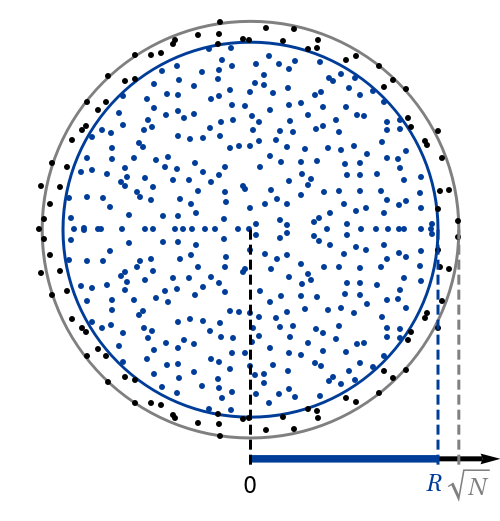}
		\end{center}
		\subcaption{Edge regime: $R = \sqrt{N} - O(1)$}
	\end{subfigure}

	\caption{Scaling regimes with points from the real Ginibre ensemble ($N=500$).} \label{Fig_Scaling Regimes}
\end{figure}

It is useful to recall the universality results that have been obtained so far for the Ginibre ensembles. In \cite{lacroix2019rotating}, the mean, the variance and all higher order cumulants were determined for the GinUE and its extension to rotational invariant potentials $W(|z|)$, the random normal matrix ensembles in the GinUE class. 
(See also \cite{ameur2023eigenvalues} for recent work on the case of the potential $W$ with a hard wall.)
In this case, the cumulants are given in terms of poly-logarithms for finite and infinite $N$. 
These results were extended to normal symplectic ensembles in the GinSE class for the mean and variance in \cite{akemann2022universality}. It was found that when measuring the variance in units of the mean density given by the Laplacian of the potential $\Delta W(R)$ at radius $R$, the linear growth of the variance in the bulk is universal, i.e., the GinUE and the normal GinSE coincide. The same was found for the universal scaling function of the variance at the edge of the spectrum. (Here, we use the convention $\Delta=\partial \bar{\partial}$, a quarter of the usual Laplacian.)
On the other hand, in the origin regime, the variance depends on the singularities (or zeros) of the potential $W$ and on the repulsion from the real axis present in the GinSE class. 

In this work we will address the third class, the GinOE, and complete the FCS for the GinSE. 
First, we analytically determine the mean and the variance in the GinOE in the origin regime. 
For the bulk and edge regime we formulate precise conjectures for the GinOE to be universal, and provide numerical evidence for these conjectures. 
We make contact between the origin and bulk regime by taking the large radius limit analytically. 
Because, in the GinOE, correlations between real, complex and both kind of eigenvalues have to be considered separately, the computation of the mean and the variance will occupy a substantial part of the paper. 
In the last part we explicitly determine all cumulants for finite and large-$N$ in the planar symplectic ensembles with rotationally invariant potentials $W$ in the GinSE class, and prove their universality in the bulk and edge scaling limit. Indeed, we show that they coincide with those found for the random normal matrix ensembles.

The remainder of the paper is organised as follows. 
In the next Section \ref{sec:main}, we briefly recall some basic properties of the three Ginibre ensembles, including their joint densities and correlation functions of complex eigenvalues (see Appendix \ref{App:A} for details for the GinOE). 
The following main results are presented there. In Subsection \ref{subsec:mean} the mean number of eigenvalues in a centred disc is given for finite~$N$, outside the limiting support, and in the large $N$ limit at the origin. 
This includes new results for the GinOE. Subsection \ref{subsec:var-inf} is devoted to the limiting variance. In the origin limit new expressions are given for the GinSE and GinUE. For the GinOE the main Theorem \ref{Thm_GinOE inf} states the limiting variance at the origin for asymptotically large radius. We also discuss the issue of universality for all three ensembles there and formulate conjectures for the behaviour of the  variance in the GinOE in the bulk, at the edge and outside the droplet. These conjectures are supported by numerics. 
Finally Subsection~\ref{subsec:FCS-SE} summarises our findings for the FCS in planar symplectic ensembles.
The proofs for the results from Subsection~\ref{subsec:mean} on the mean are given in Section \ref{sec:proofmean}. The results on the variance from Subsection~\ref{subsec:var-inf} are proved in Section \ref{sec:proofvar}, and the universal FCS for planar symplectic ensembles is given in Section~\ref{sec:proofFCS}. We conclude in Section~\ref{sec:conc} and list a number of open problems.

\section{Main results}\label{sec:main}

In order to present our main results we briefly recall the Ginibre ensembles together with the quantities we need. For the GinOE more details are presented in Appendix \ref{App:A}. In terms of matrices $G$, the Ginibre ensembles consists of random matrices that have $N^2$ independent, centred normal real, complex or quaternion matrix elements, $g_{jk}\in\mathcal{N}_F(0,(N\beta)^{-\frac12})$, without further symmetry constraint. Here the Dyson index $\beta=1,2,4$ labels the field $F=\mathbb{R,C,H}$, respectively. 

The joint probability distribution function (jpdf) of complex eigenvalues $z_1,\dots,z_N$ of the GinUE respectively GinSE, normalised by the partition function $\mathcal{Z}_N^{(2,4)}$, reads
\begin{align}
\mathcal{P}_N^{(2)}(z_1,\dots,z_N)&= \frac{1}{\mathcal{Z}_{N}^{(2)}}  \prod_{j>k=1}^{N} |z_j-z_k|^2 \prod_{l=1}^N e^{-N|z_l|^2},  
\label{jpdf2}
\\
\mathcal{P}_N^{(4)}(z_1,\dots,z_N) &= \frac{1}{\mathcal{Z}_{N}^{(4)}}  \prod_{j>k=1}^{N} |z_j-z_k|^2 |z_j-\overline{z}_k|^2  \prod_{l=1}^N |z_l-\overline{z}_l|^2 e^{-2N|z_l|^2}. \label{jpdf4}
\end{align}
For $\beta=4$ note that an $N \times N$ quaternion valued matrix $G$ has an infinite number of eigenvalues, however they decompose into $N$ equivalence classes of the form $[z_j] = \{q^{-1} z_j q \mid q \in \mathbb{H} \}$ for some $z_j \in \C$ and $j = 1, \dots, N$, see \cite{Mehta} for a discussion.
Furthermore it holds that $[z_j] \cap \C = \{z_j, \overline{z}_j\}$ and the $2N$ complex numbers $\{z_1, \overline{z}_1, \dots, z_N, \overline{z}_N\}$ are precisely the eigenvalues of the complex $2N \times 2N$ representation of $G$.
In \eqref{jpdf4} we give the jpdf for $N$ complex eigenvalues only, as the remaining $N$ are fixed by complex conjugation. In principle, we could restrict the jpdf to the upper-half plane $\C_+:=\{z \mid z\in\C, \Im(z)>0\}$, because the jpdf is invariant under complex conjugation, cf.~\cite{reviewKhoruzhenkoSommers}. 
However, in much of the literature \cite{Mehta,MR1928853,akemann2021scaling} the jpdf is given on the full complex plane $\C$, the reason being that it is not difficult to find the complex eigenvalue correlation functions in the lower half plane as well.

For the GinOE with $N\times N$ real matrices $G$, we have to distinguish $k$ real and $2l$ (non-real) complex eigenvalues, coming in conjugated pairs. Because of $N=k+2l$, $k$ is always of the same parity as $N$, otherwise the probability of finding $k$ such real eigenvalues is zero.  
The corresponding jpdfs for fixed $k$ and $l$ are not immediately needed here and are  presented in Appendix \ref{App:A} for completeness. They depend on $k$ real and on $l$ complex eigenvalues, with the remaining $l$ given by complex conjugation. In contrast to $\beta=4$ it is not straightforward to write the correlation functions in the entire complex plane, see however \cite{sommers2007symplectic}, due to sign-functions present in the jpdf. 
Therefore, as most authors, we restrict ourselves to the upper half plane (except for the density, see below).
Notice that all three Ginibre ensembles correspond to a two-dimensional Coulomb gas at inverse temperature $(k_B T)^{-1}=2$ \cite{forrester2010log,forrester2016analogies,Serfaty,Lewin22}. In the following, we will use the notation $\beta = 1, 2, 4$ to describe respectively the GinOE, the GinUE and the GinSE, but it does not have the interpretation of an inverse temperature. 
This is in contrast to the Hermitian random matrices, where the Gaussian Orthogonal/Unitary/Symplectic ensembles correspond to the Coulomb gas with respective $\beta$-values of 1, 2, and 4. (See also a recent work \cite{mezzadri2023matrix} for a non-Hermitian $\beta$-ensemble.)

\medskip 

The $k$-point correlation functions at finite-$N$, denoted by bold face, are defined as
\begin{equation}
\textbf{R}_{N,k}^{(\beta)}(z_1,\dots,z_k):=\frac{N!}{(N-k)!}\prod_{j=k+1}^N\int_\C \mathcal{P}_N^{(\beta)}(z_1,\dots,z_N) \,dA(z_j),
\label{RkDef}
\end{equation} 
for $\beta=2,4$, 
where $dA(z)=d^2z/\pi$ is the area measure. 
Notice that in part of the literature the flat measure without a factor of $\pi$ is used instead. 
For $\beta=1$ we refer to the discussion in Appendix \ref{App:A}.

As a consequence we obtain the following determinantal, (respectively Pfaffian), point processes:
\begin{itemize}
    \item \textbf{GinUE:} The $N$ complex eigenvalues in \eqref{jpdf2} form a 
     DPP with kernel $K_N$ \cite{ginibre1965statistical},
\begin{equation} \label{Rk beta2N}
\textbf{R}_{N,k}^{(2)}(z_1,\dots,z_k)= \det \Big[ K_N(z_j,z_l) \Big]_{j,l=1}^k, 
\end{equation}
where 
\begin{equation}
K_N(z,w)= e^{ -\frac{N}{2}(|z|^2+|w|^2)+Nz\Bar{w}}NQ(N,Nz\Bar{w})=  e^{ -\frac{N}{2}(|z|^2+|w|^2)+Nz\Bar{w}}N \sum_{j=0}^{N-1} \frac{( N z \bar{w} )^j}{ j! } .
\end{equation}
Here, 
\begin{equation}
Q(n+1,z)=\frac{\Gamma(n+1,z)}{\Gamma(n+1)}=e^{-z}e_n(z) 
\end{equation}
is the normalised incomplete Gamma-function, that relates to the incomplete exponential $e_n(z)=\sum_{j=0}^{n}\frac{z^j}{j!}$. 
\smallskip 
     \item \textbf{GinSE:} The $N$ complex eigenvalues in \eqref{jpdf4} form a Pfaffian point process \cite{Mehta,MR1928853}
\begin{equation} \label{Rk beta4N}
\textbf{R}_{N,k}^{(4)}(z_1,\cdots, z_k) =N^{k/2}\prod_{j=1}^{k} (\bar{z}_j-z_j) \, \Pf \Big[ e^{-N|z_j|^2-N|z_l|^2} 
\begin{pmatrix} 
\kappa_N(z_j,z_l) & \kappa_N(z_j,\bar{z}_l) 
\\
\kappa_N(\bar{z}_j,z_l) & \kappa_N(\bar{z}_j,\bar{z}_l) 
\end{pmatrix}  \Big]_{ j,l=1,\cdots k },
\end{equation}
where the $2\times2$ matrix valued kernel above contains the skew-kernel 
    \begin{equation} \label{kappaN}
    \kappa_N(z,w) := N\sum_{j=0}^{N-1}\sum_{l=0}^{j}N^{k+l+1}\frac{z^{2j+1}w^{2l}-w^{2j+1}z^{2l}}{(2k+1)!!\,(2l)!!}.
    \end{equation}
    \end{itemize}
    
For the GinOE we instead give the densities $\textbf{R}_{N,1,\C}^{(1)}(z)$ and $\textbf{R}_{N,1,\R}^{(1)}(x)$ for the non-real complex and real eigenvalues, respectively, as examples. They were first derived by \cite{edelman1997probability}, respectively \cite{MR1231689}.
\begin{itemize}
\item  \textbf{GinOE:} The 1-point densities of complex ($z=x+iy\in\C$) (respectively, real ($x\in\R$)) eigenvalues are
\begin{align}
\textbf{R}_{N,1,\C}^{(1)}(z) &= \sqrt{ 2N\pi} N|y| \erfc(\sqrt{2N}|y|) e^{2Ny^2} Q(N-1,N|z|^2), 
\label{RN1complex1}
\\
\textbf{R}_{N,1,\R}^{(1)}(x) &= \sqrt{\frac{N}{2\pi}} Q(N-1,Nx^2) +  \frac{ N^{N/2}   }{ 2^{N/2 } \Gamma(\frac{N}{2})  } e^{ -\frac{N}{2} x^2 } |x|^{N-1}  P\Big( \frac{N-1}{2}, \frac{Nx^2}{2} \Big),  
\label{RN1real1}
\end{align}
where 
$$
\erfc(z) = 1- \erf(z) = \frac{2}{ \sqrt{\pi}  } \int_z^\infty e^{-t^2}\,dt
$$
is the complementary error function and
\begin{equation}
P(n+1,z)=1-Q(n+1,z)=\frac{\gamma(n+1,z)}{\Gamma(n+1)} , \qquad \gamma(a,z)= \int_0^z t^{a-1} e^{-t}\,dt 
\end{equation}
is the second regularised incomplete Gamma-function. Notice that in \eqref{RN1complex1} we defined the density of complex eigenvalues in the full complex plane $\C$. 
In general, we will also encounter mixed correlation functions, $\textbf{R}_{N,k,\R,l,\C}^{(1)}$, of $k$ real and $l$ complex eigenvalues.
For these higher order correlation functions, which are needed for the computation of the variance, it is not so obvious where to put absolute values around the imaginary parts, see however \cite{sommers2007symplectic}.
For  the general structure of the Pfaffian point process on the upper half plane $\C_+$, and the definition of general correlation functions among real, complex, or both kind of eigenvalues,  we refer to Appendix~\ref{App:A}. 
\end{itemize}

Notice that due to the definition \eqref{RkDef} for $k=1$, together with \eqref{RN1complex1} and \eqref{RN1real1}, the densities are normalised to $N$ on the full complex plane, rather than to unity,
\begin{equation}
\int_{\C} \textbf{R}_{N,1}^{(\beta=2,4)}(z) \,dA(z)=N=\int_{\C} \textbf{R}_{N,1,\C}^{(1)}(z) \,dA(z)+\int_{\R} \textbf{R}_{N,1,\R}^{(1)}(x) \,dx. 
\end{equation}
Furthermore, our normalisation is such that for $N\gg1$ the support of the (complex) densities (also called droplet) converges to the circular law with radius unity.

\subsection{Expected number in a centred disc in the Ginibre ensembles}\label{subsec:mean}

\subsubsection{Finite-\texorpdfstring{$N$}{N} Ginibre ensembles}\label{subsubsec:meanfinite}
 
We denote by $E_N^{(\beta)}(a)$ the expected number of eigenvalues in the centred disc $D_a$ of radius $a$. 
For the GinOE, we split this into the contribution of complex eigenvalues $E_{N,\C}^{(1)}(a)$, and  write $E_{N,\R}^{(1)}(a)$ for the expected number of real eigenvalues in the interval $(-a,a)$. 
Note that the droplet has radius $a = 1$. In other words, we define 
\begin{equation} \label{EN a Def}
E_N^{(\beta=2,4)}(a) := \int_{D_a} \textbf{R}_{N,1}^{(\beta=2,4)}(z) \,dA(z), \qquad E_N^{(1)}(a) :=  E^{(1)}_{N, \C}(a)+E^{(1)}_{N,\R}(a),
\end{equation}
where
\begin{equation} \label{EN RC 1 a Def}
E^{(1)}_{N,\C}(a):= \int_{D_a} \textbf{R}_{N,1,\C}^{(1)}(z) \,dA(z), \qquad  E^{(1)}_{N,\R}(a) := \int_{-a}^a \textbf{R}_{N,1,\R}^{(1)}(x) \,dx . 
\end{equation}
We first obtain analytic expressions for the expected number of eigenvalues.
For this, recall that the (generalised) hypergeometric function ${}_pF_q$ is given by \cite[Chapter 16]{olver2010nist}
\begin{equation}\label{def hypergeometric}
{}_pF_q(  a_1,\dots,a_p; b_1,\dots,b_q ; z ) = \sum_{k=0}^\infty \frac{ (a_1)_k \dots (a_p)_k }{ (b_1)_k \dots (b_q)_k } \frac{z^k}{k!},
\end{equation}
where $(x)_k = x (x+1) \dots (x+k-1)$ is the Pochhammer symbol.

\begin{prop}[\textbf{Expected number of eigenvalues for finite-$N$}]  \label{Prop_EN GinOE}
For each $a \geq 0$ and $N\in\N$, we have the following: 
\begin{itemize}
    \item[\textup{(i)}] \textbf{\textup{GinUE:}} We have 
\begin{equation} \label{E_N cplx Ginibre}
\begin{split}
E_N^{(2)}(a) &=\sum_{k=0}^{N-1} P(k+1,Na^2)= Na^2+N(1-a^2) \,P(N,Na^2)-\frac{ (Na^2)^N\, e^{-Na^2} }{ (N-1)! }. 
\end{split}
\end{equation}
\item[\textup{(ii)}] \textbf{\textup{GinSE:}} We have 
\begin{equation}
\begin{split}  \label{EN symplectic Ginibre}
\quad E_N^{(4)}(a)  &= \sum_{k=0}^{N-1}  P(2k+2,2Na^2)  =\frac12 E_{2N}^{(2)}(a)- \frac{e^{-2Na^2} }{2} \sum_{k=0}^{N-1} \frac{(2Na^2)^{2k+1}   }{(2k+1)!} 
\\
&=\frac12  E_{2N}^{(2)}(a)- \frac{e^{-2Na^2} }{2} \Big( \sinh(2Na^2)- \frac{ (2Na^2)^{2N+1} }{(2N+1)!} \, {}_1F_2\Big( 1; N+1, N+\frac32 ;  N^2a^4 \Big) \Big).  
\end{split}
\end{equation}
\item[\textup{(iii)}] \textbf{\textup{GinOE:}} We have 
\begin{equation}
\begin{split} \label{EN C 1 a GinOE}
E_{N,\C}^{(1)}(a) 
& =  \sum_{k=0}^{N-2} P(k+1,Na^2)  - 2 \int_0^{ \sqrt{N}a } r\, e^{2r^2} \erfc(\sqrt{2}r)   Q(N-1,r^2)\,dr
\end{split}
\end{equation}
and for even $N$,
\begin{align}
\begin{split}
E_{N,\R}^{(1)}(a) 
&=\sqrt{ \frac{2}{\pi} } \frac{ a\sqrt{N} \,  \Gamma(N-1,Na^2)+\gamma(N-\tfrac12, Na^2) }{ (N-2)! } 
\\
&\quad +\frac{ 1   }{ (N-2)!!  }     \int_{0}^{ \sqrt{N}a }  e^{ -x^2/2 } x^{N-1} \erf\Big(  \frac{x}{ \sqrt{2} }\Big)\,dx
\\
&\quad - \frac{1}{ 2^{N/2} \, (N-2)!! }        \sum_{k=1}^{ N/2-1 } \frac{  (N+2k-3)!! }{ 2^{k-1/2} \, (2k-1)!! }   P\Big( \frac{N-1}{2}+k, Na^2\Big).  
\end{split}
\end{align}
\end{itemize}
\end{prop}

We emphasize that Proposition~\ref{Prop_EN GinOE} (i) and (ii) are given in \cite[Proposition 3.1]{akemann2022universality} in a slightly different form.
For $N$ odd the corresponding expressions in Proposition~\ref{Prop_EN GinOE} (iii) are given in Lemma~\ref{Prop_GinOE expected real} in the main text.

For $\beta=1$, the $N \to \infty$ leading order asymptotic behaviours are given by 
\begin{equation} \label{EN GinOE leading order}
E_{N,\C}^{(1)}(a) \sim N \min(a^2,1) \ ,\qquad 
E_{N,\R}^{(1)}(a) \sim \sqrt{ \frac{2N}{\pi} } \min(a,1). 
\end{equation}
This corresponds to the results in \cite{MR1231689}.
We also mention that the function 
\begin{equation}
\Big( 1-e^{2r^2} \erfc(\sqrt{2}r)  \Big) Q(N-1,r^2)\ ,\quad \mbox{with}\ r=|z|, 
\end{equation}
in the integrand of \eqref{EN C 1 a GinOE} corresponds to the radial density of the complex eigenvalues of the GinOE, where the angular integral has been performed.

As an application of Proposition~\ref{Prop_EN GinOE}, we obtain the following corollary. 

\begin{cor}[\textbf{Expected number of eigenvalues outside the limiting support}] \label{Cor_outside the disc number}
As $N\to \infty$, we have 
\begin{align}
N-E_N^{(\beta=1,2)}(1) &=  \frac{ \sqrt{N} }{ \sqrt{2\pi} }+o(\sqrt{N}),  \label{EN disc GinUE}
\\
 N- E_N^{(4)}(1)  &=  \frac{ \sqrt{N} }{ 2\sqrt{\pi} }+o(\sqrt{N}), \label{EN disc GinSE} 
\end{align}
For $\beta=1$ we also have the more detailed asymptotics 
\begin{equation} \label{EN detailed}
E_{N,\C}^{(1)}(1)= N - \frac32 \sqrt{ \frac{2}{\pi} } \sqrt{N}+o(\sqrt{N}), \qquad  E_{N,\R}^{(1)}(1)=  \sqrt{ \frac{2}{\pi} } \sqrt{N}+o(\sqrt{N}).
\end{equation}
\end{cor}

Notice the difference by a factor $1/\sqrt{2}$ for $\beta=4$ comparing \eqref{EN disc GinUE} and \eqref{EN disc GinSE}. 
Let us stress that in \cite{MR3450566}, the behaviour \eqref{EN disc GinUE} was shown in the more general context of the elliptic GinUE. 
Furthermore, for the GinUE, more precise asymptotic behaviour is available in \cite{charlier2022asymptotics}, see also Remark~\ref{Rem_ML FCS}. 
We mention that the order $O(\sqrt{N})$ is closely related to the order of the semi-large gap probabilities, cf. \cite{charlier2021large,byun2023almost}.
Let us also note that for the GinUE case, the optimal convergence rate of $O(N^{-1/2})$ to the circular law was established in \cite{MR4302281}. This result was further generalised to the products of GinUE matrices as well \cite{MR4254801}. 

In Figure~\ref{Fig_ENGinOE} (A), the different contributions of real and complex eigenvalues to the expected number are plotted for $\beta=1$.
In Figure~\ref{Fig_ENGinOE} (B), the deviation from the total number of eigenvalues $N$ is compared in all three ensembles.

\begin{figure}[t]
    \begin{subfigure}{0.48\textwidth}
        \begin{center}
            \includegraphics[width=\textwidth]{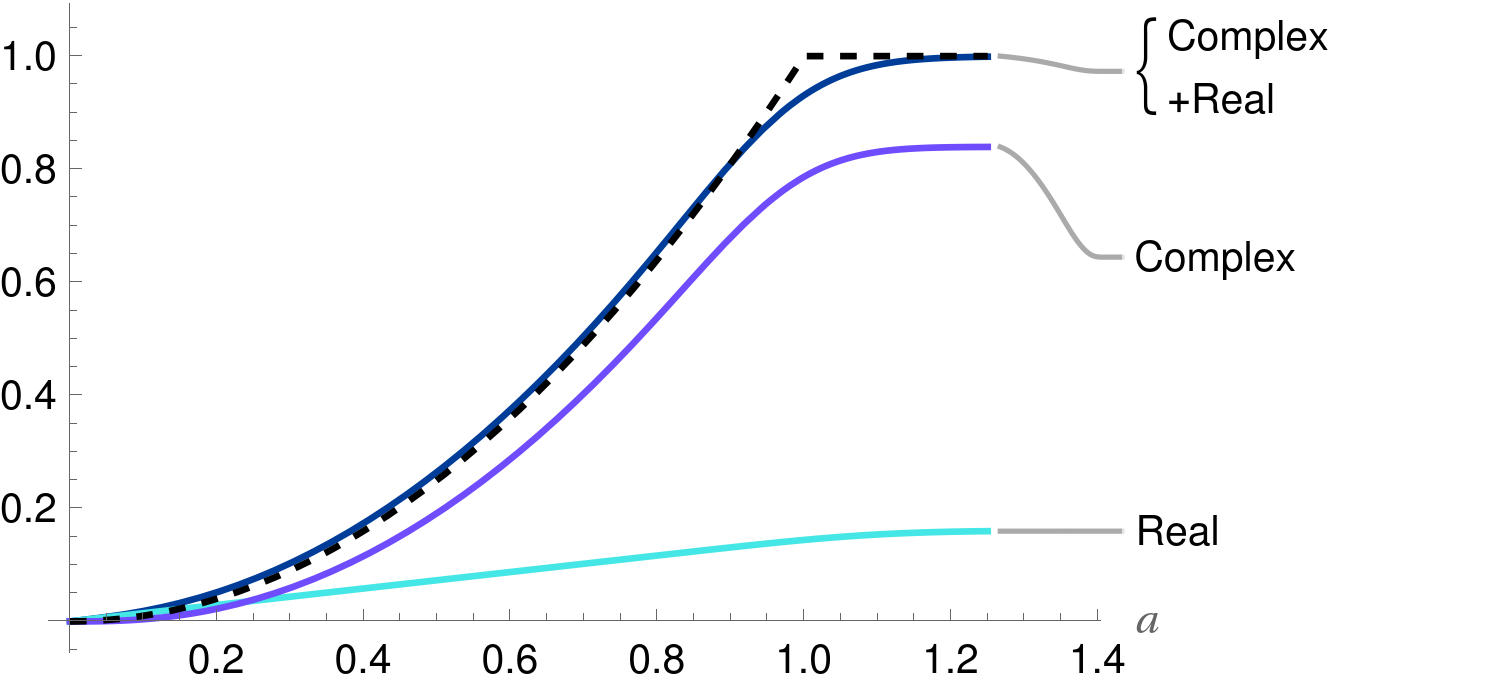}
        \end{center}
        \caption{$E_N^{(1)}(a), E_{N,\C}^{(1)}(a), E_{N,\R}^{(1)}(a) $}
    \end{subfigure} 
    \begin{subfigure}{0.48\textwidth}
        \begin{center}
            \includegraphics[width=\textwidth]{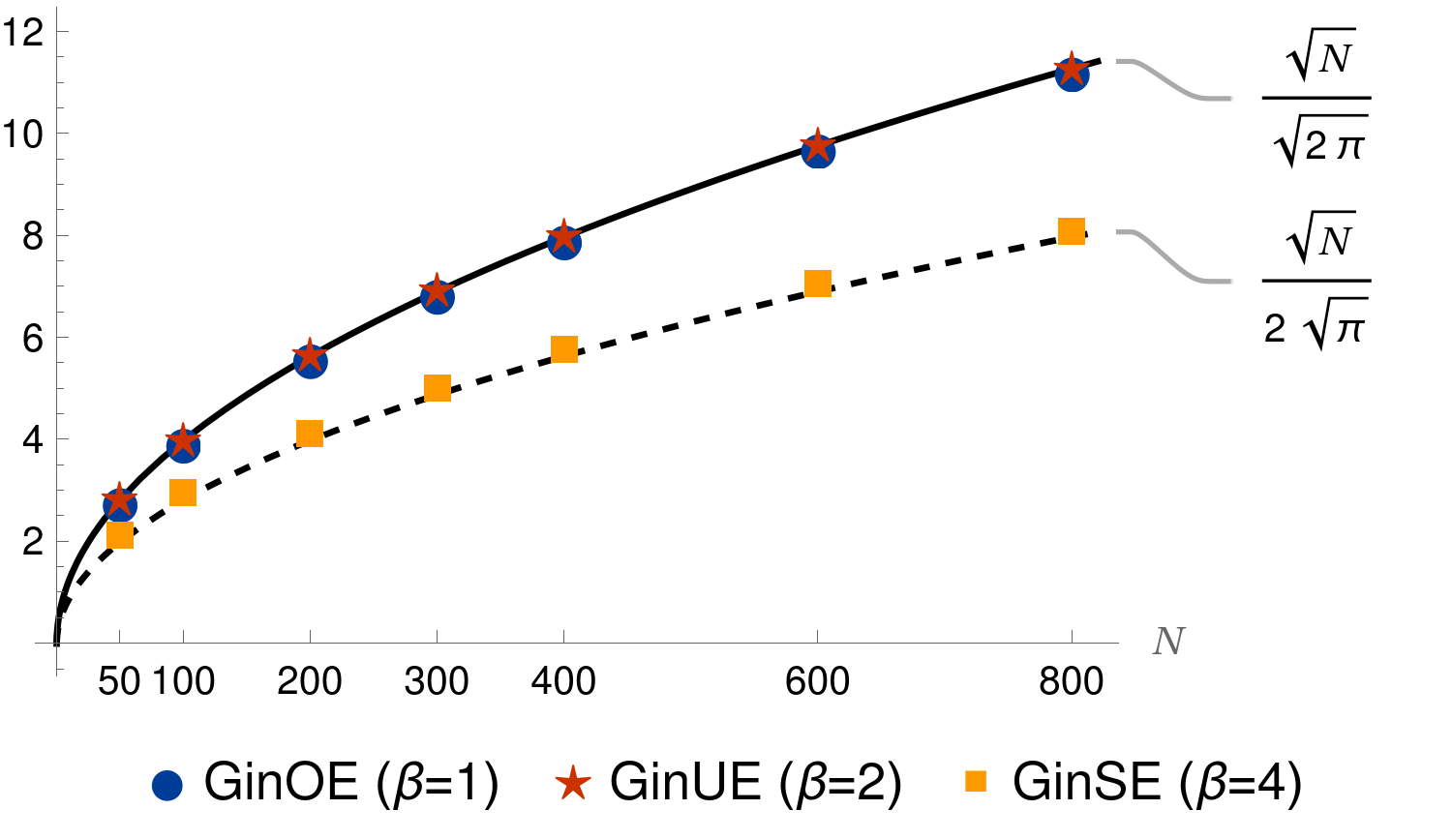}
            \caption{$N - E_N^{(\beta)}(1)$} 
        \end{center}
    \end{subfigure} 
    \caption{(A): Graphs of $a \mapsto E_N^{(1)}(a)/N $ (Complex+Real), $ E_{N,\C}^{(1)}(a)/N $ (Complex) and $ E_{N,\R}^{(1)}(a)/N $ (Real) in the GinOE, where $N=30$.  \\   
    (B): Datapoints of $N \mapsto N - E_N^{(\beta)}(1)$ computed from simulations and comparison to Corollary~\ref{Cor_outside the disc number}. The GinUE and GinOE nicely agree in the top curve, and the GinSE is smaller by a factor of $1/\sqrt{2}$ as predicted.}
    \label{Fig_ENGinOE}
\end{figure}

\subsubsection{Infinite Ginibre ensembles in the origin scaling limit}\label{subsubsec:mean-inf}

The origin scaling limit is defined by rescaling all variables in the $k$-point correlation functions as $z_j^\prime=\sqrt{N}z_j$, $j=1,\dots,k$, and in particular the radius of the disc for the FCS as $R = \sqrt{N} a$. In other words we define
\begin{equation}\label{lim Rkb24}
{R}_{k}^{(\beta=2,4)}(z_1^\prime,\dots,z_k^\prime):=\lim_{N\to\infty}N^{-k}\ \textbf{R}_{N,k}^{(\beta=2,4)}(z_1=z_1^\prime/\sqrt{N},\dots,z_k=z_k^\prime/\sqrt{N}).
\end{equation}
For the GinOE we have to distinguish the different dimensions of the eigenvalues:
\begin{align}
{R}_{k,\C}^{(1)}(z_1^\prime,\dots,z_k^\prime)&:=\lim_{N\to\infty}N^{-k}\ \textbf{R}_{N,k,\C}^{(1)}(z_1=z_1^\prime/\sqrt{N},\dots,z_k=z_k^\prime/\sqrt{N}),
\label{originRkC1}
\\
{R}_{k,\R}^{(1)}(x_1^\prime,\dots,x_k^\prime)&:=\lim_{N\to\infty}N^{-k/2}\ \textbf{R}_{N,k,\R}^{(1)}(x_1=x_1^\prime/\sqrt{N},\dots,x_k=x_k^\prime/\sqrt{N}).
\label{originRkR1}
\end{align}
Likewise, for the mixed correlation functions $\textbf{R}_{N,k,\R,l,\C}^{(1)}$ defined in Appendix~\ref{App:A}, we have for the origin scaling limit for $k=1$ real and $l=1$ complex eigenvalue
\begin{equation}
{R}_{1,\R,1,\C}^{(1)}(x_1^\prime,z_1^\prime):= \lim_{N\to\infty}N^{-\frac{1}{2}-1}\ 
  \textbf{R}_{N,1,\R,1,\C}^{(1)}(x_1=x_1^\prime/\sqrt{N},z_1=z_1^\prime/\sqrt{N}).
  \label{originR11}
\end{equation}
Consequently, we define the limiting mean at the origin as 
\begin{equation}
E^{(\beta=2,4)}(R) := \int_{D_R} R_1^{(\beta=2,4)}(z) \,dA(z), \qquad E^{(1)}(R) :=  E^{(1)}_\C(R)+E^{(1)}_\R(R),
\label{limEDef}
\end{equation}
where
\begin{equation}
E^{(1)}_\C(R):= \int_{D_R} R_{1,\C}^{(1)}(z) \,dA(z), \qquad  E^{(1)}_\R(R) := \int_{-R}^R R_{1,\R}^{(1)}(x) \,dx . 
\label{limEDef1}
\end{equation}
These have the interpretation as the expected number of eigenvalues in a disc of radius $R$.  Furthermore, for the radial densities of complex eigenvalues on $\C$ we write 
\begin{equation}
\widehat{R}_1^{(\beta=2,4)}(r):= \frac{1}{2\pi} \int_0^{2\pi} R_1^{(\beta=2,4)} ( r e^{i\theta} )\,d \theta, \qquad 
\widehat{R}_{1,\C}^{(1)}(r):= \frac{1}{2\pi} \int_0^{2\pi} R_{1,\C}^{(1)} ( r e^{i\theta} )\,d \theta.
\end{equation}

\begin{prop}[\textbf{Limiting mean in the origin scaling limit}] \label{Prop_EN inf}
We have 
\begin{align}
 \label{E DR C}
&E^{(2)}(R) = R^2,
\\
&\label{E DR S}
E^{(4)}(R) = R^2-\frac14+\frac14 e^{-4R^2},
\\
 \label{E DR R}
&E^{(1)}(R) = R^2 +\frac12- \frac12 e^{ 2R^2 } \erfc(\sqrt{2}R).
\end{align}
In particular, we have
\begin{equation}
E^{(1)}(R) = E^{(1)}_\C(R)+ E^{(1)}_\R(R),  
\end{equation}
where 
\begin{equation} \label{E DC DR R}
E_\C^{(1)}(R) =  R^2-\sqrt{ \frac{2}{\pi} } R +\frac12- \frac12 e^{ 2R^2 } \erfc(\sqrt{2}R), \qquad  E_\R^{(1)}(R) =  \sqrt{ \frac{2}{\pi} } R .
\end{equation}
Furthermore, the respective radial densities read
\begin{equation} 
\widehat{R}_1^{(2)}(r)=1, \qquad  \widehat{R}_1^{(4)}(r)=1-e^{-4r^2}, \qquad \widehat{R}_{1,\C}^{(1)}(r)=1-e^{2r^2} \erfc(\sqrt{2}r).
\end{equation}
\end{prop}

\begin{table}[h!]
    \centering
   \begin{center}
{\def\arraystretch{3}
\centering
\begin{tabular}{ |p{2cm}|p{4.5cm}|p{2.5cm}|p{4.5cm}|  }
 \hline
& \centering \textup{GinOE} & \centering \textup{GinUE} &  \hspace{3.5em} \textup{GinSE} 
\\ 
\hline 
\centering $E^{(\beta)}(R)$ & \centering $R^2 +\dfrac12- \dfrac12 e^{ 2R^2 } \erfc(\sqrt{2}R)$ & \centering $R^2$  &  $ \hspace{2em} R^2-\dfrac14+\dfrac14 e^{-4R^2}$ 
 \\ 
 \hline
\centering $R \to 0$ & \centering $ \sqrt{ \dfrac{2}{\pi} } R+O(R^3) $ & \centering $R^2$ & \hspace{3em} $2R^4+O(R^6)$
 \\ 
   \hline 
\centering $R \to \infty$  &  \centering $R^2 +\dfrac12+O\Big(\dfrac{1}{R}\Big)$  & \centering $R^2$ & \hspace{2em} $R^2-\dfrac14 +O(e^{-4R^2})$
 \\ 
 \hline
\end{tabular}
}
\end{center}
    \caption{Comparison of the exact results for the expected mean number and its expansion for the small and large $R$ limit for all three Ginibre ensembles.  }
    \label{tab:E limits}
\end{table}

The expansion of the exact results from Proposition~\ref{Prop_EN inf} in the limits $R\to0$ and $R\to\infty$ is collected in Table~\ref{tab:E limits}. Recalling the rescaling of the origin limit with $R^2=Na^2$, the large-$R$ limit makes contact with the bulk regime where $0<a<1$. 
Note that all three ensembles show a $R^2$-behaviour in the large $R$ limit, however $E(R) = R^2$ is exact only for the GinUE. 
Furthermore, close to the origin, we observe an $R^4$-behaviour for the GinSE, which arises from the repulsive effect of the real axis.
The linear behaviour for the GinOE stems from the real eigenvalues, while the number of complex eigenvalues in \eqref{E DC DR R} reveals a cubic behaviour 
\begin{equation} \label{E DC DR 0}
E^{(1)}_\C(R) =  \sqrt{ \frac{2}{\pi} } \frac{4}{3} R^3 + O(R^5), \qquad R \to 0.
\end{equation}
Figure~\ref{Fig_InfGin_Mean} illustrates the expected number through plots for all three ensembles.

\begin{figure}[t]
    \begin{subfigure}{0.48\textwidth}
		\begin{center}
            \includegraphics[width=\textwidth]{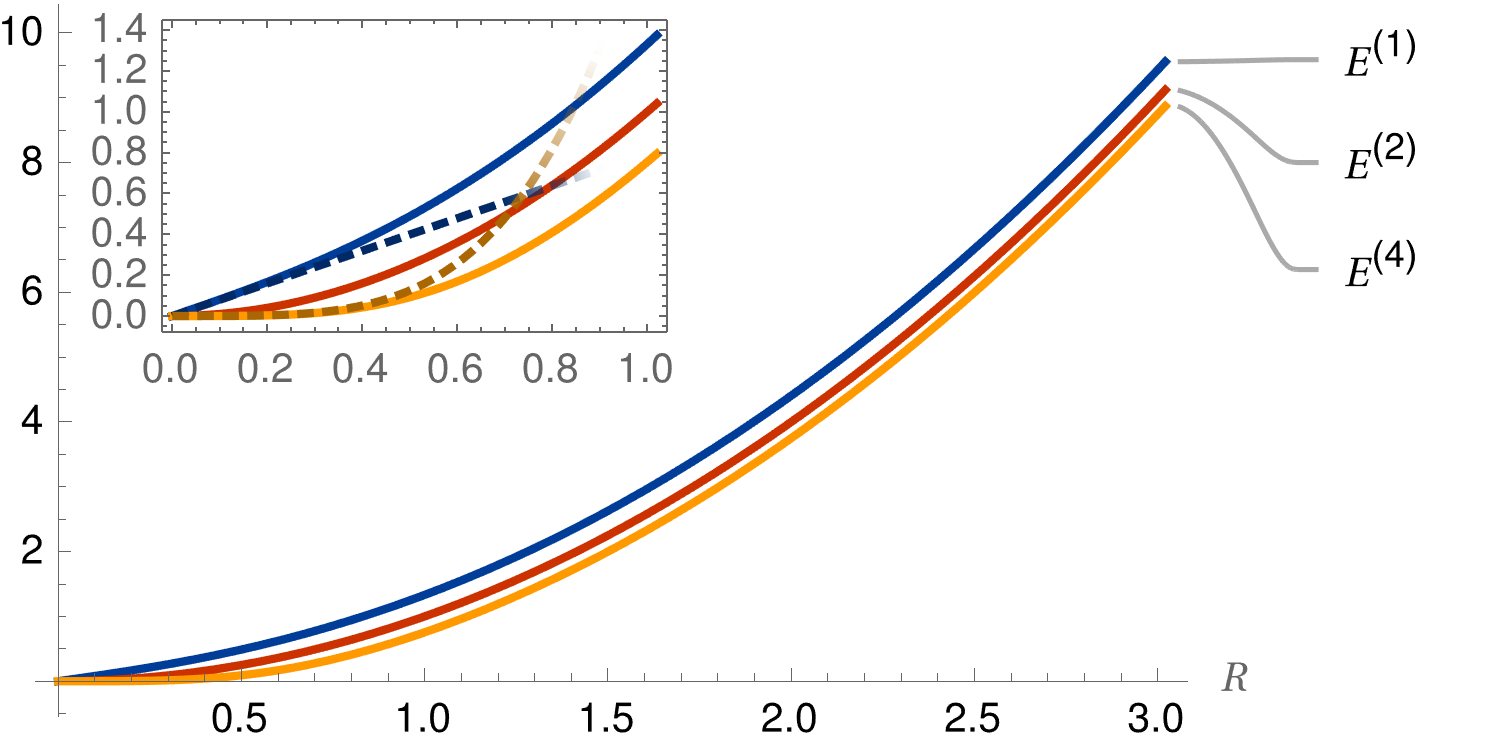}
		\end{center}
		\subcaption{GinOE, GinUE and GinSE}
	\end{subfigure}
	\begin{subfigure}{0.48\textwidth}
		\begin{center}
			\includegraphics[width=\textwidth]{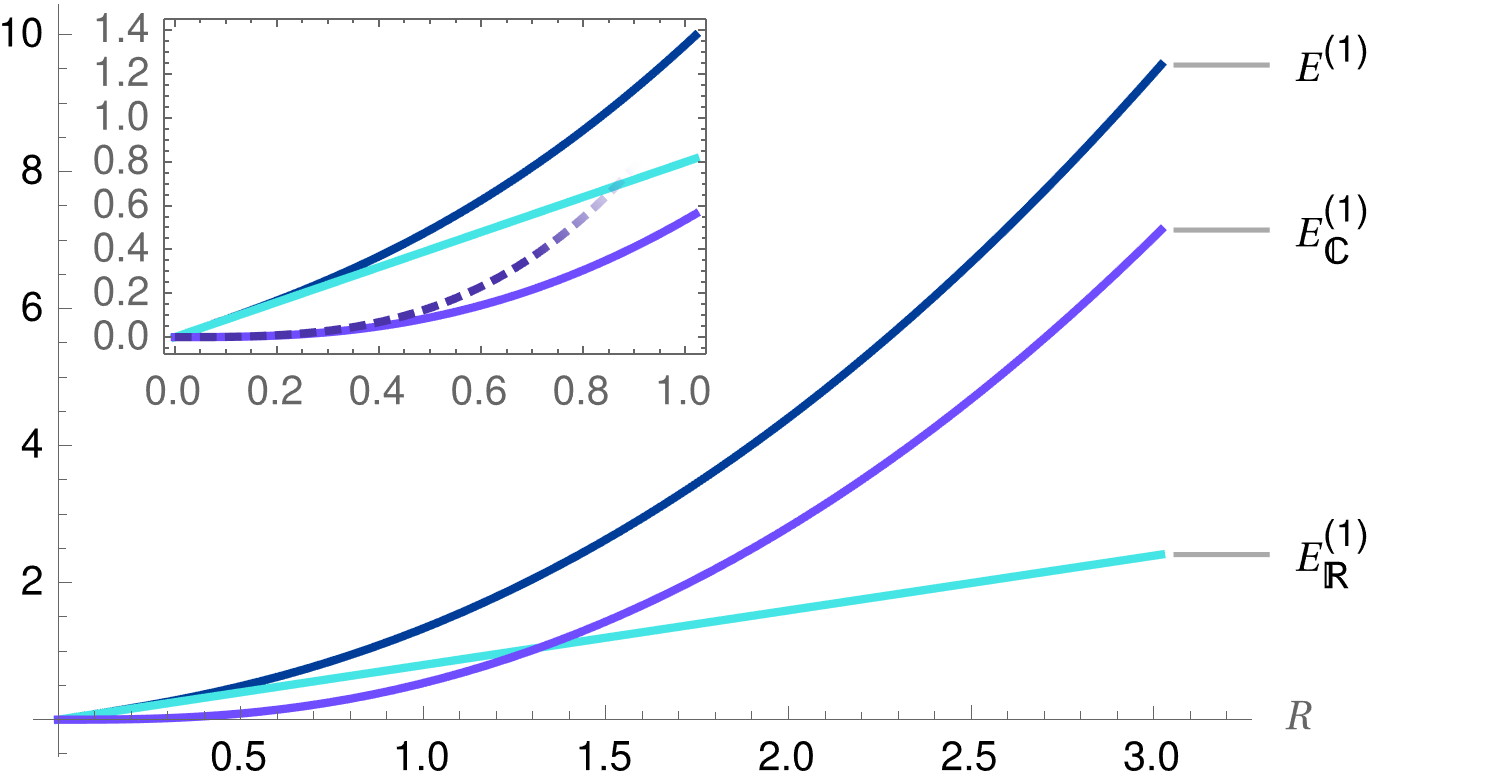}
		\end{center}
		\subcaption{Real and complex counting numbers for GinOE}
    \end{subfigure}
    \caption{Graphs of the expected number of the infinite Ginibre ensembles, cf.\ \eqref{E DR C}, \eqref{E DR S}, \eqref{E DR R} and \eqref{E DC DR R}.
    The inset shows the behaviour close to the origin, where dashed lines indicate the $R \to 0$ limits from Table~\ref{tab:E limits} and \eqref{E DC DR 0}.
    }
   \label{Fig_InfGin_Mean}
\end{figure}

\newpage 

\subsection{Number variance of the infinite Ginibre ensembles}\label{subsec:var-inf}
\hfill\\
In this subsection we will only consider the large-$N$ limit for the number variance. 

\subsubsection{Infinite Ginibre ensembles in the origin scaling limit}

Similarly to \eqref{limEDef} and \eqref{limEDef1} 
we define the limiting variance in the origin limit as 
\begin{equation}
\begin{split}
\Var \, \NN^{(\beta=2,4)}( D_R ) &:=  \int_{ D_R } R_{1}^{(\beta=2,4)}(z)\,dA(z) 
\\
&\quad +\int_{D_R^2 } \Big( R_{2}^{(\beta=2,4)} (z,w)-R_{1}^{(\beta=2,4)}(z) R_{1}^{(\beta=2,4)}(w) \Big) \,dA(z)\,dA(w).
\end{split}
\end{equation}
Likewise, for $\beta=1$, we define 
\begin{equation}
\Var \NN^{(1)}(D_R) := \Var \NN_\R^{(1)} (D_R) + 2\,\Cov \Big( \NN_\C^{(1)}(D_R),   \NN_\R^{(1)}(D_R) \Big) + \Var \NN_\C^{(1)} (D_R), 
\label{VarDef1}
\end{equation}
where 
\begin{align}
&\Var \NN_\R^{(1)}(D_R) := \int_{-R}^R R_{1,\R}^{(1)}(x)\,dx 
+\int_{(-R,R)^2} \Big(R_{2,\R}^{(1)}(x,y)-R_{1,\R}^{(1)}(x) R_{1,\R}^{(1)}(y)\Big)\, dx\,dy, 
\\ 
&\Var \NN_\C^{(1)} (D_R)  := \int_{ D_R } R_{1,\C}^{(1)}(z)\,dA(z)  +\int_{D_R^2 }\Big( R_{2,\C}^{(1)} (z,w)-R_{1,\C}^{(1)}(z) R_{1,\C}^{(1)}(w) \Big)\,dA(z)\,dA(w),
\\
&\Cov \Big( \NN_\C^{(1)}(D_R),   \NN_\R^{(1)}(D_R) \Big)  := \int_{D_R}  \int_{-R}^R    \Big( R_{1,\R,1,\C}^{(1)}(x,z)- R_{1,\R}^{(1)}(x) R_{1,\C}^{(1)}(z) \Big) \,dx \,dA(z).
\end{align}
In the proof of Proposition \ref{Prop_Expression var inf} we will explain how these expressions are determined through integrals over the upper half plane only, where the limiting correlations functions given in the Appendix \ref{App:A} are used.

Let us first recall some special functions that will be used to express the number variance. 
The modified Bessel function $I_\nu$ of the first kind is given by \cite[Chapter 10]{olver2010nist}
\begin{equation} \label{I nu definition}
I_\nu(z):=\sum_{k=0}^{\infty} \frac{(z/2)^{2k+\nu}}{k! \,\Gamma(k+\nu+1)}.
\end{equation}

\bigskip 

\begin{prop}[\textbf{Number variance of the infinite GinUE and GinSE in the origin limit}] \label{Prop_Var inf csGinibre}
We have the following equivalent expressions for the limiting variance at the origin for $\beta=2,4$.
\begin{itemize}
    \item[\textup{(i)}] \textbf{\textup{GinUE:}} We have
\begin{align}  \label{Var inf cGinibre v1}
\Var \, \NN^{(2)}( D_R) &= \sum_{j=1}^\infty P(j,R^2) Q(j,R^2) \\
&= R^2 e^{-2 R^2} \Big ( I_0(2 R^2) + I_1(2 R^2) \Big ) ,
\label{Var inf cGinibre v2}
\end{align}
with the following limiting expressions:
\begin{align}
 \label{linear behaviour Cplx}
&\Var \, \NN^{(2)}( D_R) \sim   \dfrac{R}{\sqrt{\pi}}, \qquad R \to \infty, 
\\
 \label{linear behaviour Cplx 0}
& \Var \, \NN^{(2)}( D_R) \sim R^2, \qquad R \to 0. 
\end{align}
\item[\textup{(ii)}] \textbf{\textup{GinSE:}} We have 
\begin{align}  \label{Var inf sGinibre v1}
\Var \NN^{(4)}(D_R) &=  \sum_{k=1}^\infty P(2k,2R^2) Q(2k,2R^2) 
\\
&=    R^2 e^{ -4R^2 } \Big ( I_0(4 R^2) + I_1(4 R^2) - {}_1 F_2  ( 1/2; 1,3/2; -4 R^4 ) \Big) \label{Var inf sGinibre v2}
\end{align}
with the limiting expressions:
\begin{align}
\label{linear behaviour Sym}
& \Var \NN^{(4)}(D_R)  \sim  \dfrac{ R }{ \sqrt{2\pi} }, \qquad R \to \infty, 
\\
& \label{linear behaviour Sym 0}
\Var \NN^{(4)}(D_R)  \sim  2 R^4, \qquad R \to 0.  
\end{align}
\end{itemize}
\end{prop}

The first expression \eqref{Var inf cGinibre v1} 
was derived in \cite{lacroix2019rotating}, the second  expression \eqref{Var inf cGinibre v2}  in \cite[Theorem 1.3]{shirai2006large}, where it is written as 
\begin{equation} \label{Var Shirai}
\Var \, \NN^{(2)}( D_R)= \frac{R}{\pi} \int_0^{4R^2} \Big(1-\frac{x}{4R^2}\Big)^{1/2} x^{-1/2} e^{-x} \,dx. 
\end{equation}
We will present an alternative derivation of this result compared to \cite{shirai2006large}.
The expression \eqref{Var inf cGinibre v2} is equivalent to \eqref{Var Shirai} due to \cite[3.364.1, 3.366.1]{gradshteyn2014table}.

The expression \eqref{Var inf sGinibre v1} for $\beta=4$ was shown in \cite{akemann2022universality} as a large-$N$ limit of the number variance \eqref{variance symplectic} of the finite symplectic Ginibre ensemble in the origin 
regime.
The second expression \eqref{Var inf sGinibre v2} is new and can be rewritten as 
\begin{equation} 
\Var \NN^{(4)}(D_R)= R^2 e^{ -4R^2 } \Big ( I_0(4 R^2) + I_1(4 R^2) - J_0(4R^2) \frac{\textbf{H}_{-1}(4R^2)}{2/\pi}-J_{1}(4R^2) \frac{\textbf{H}_0(4R^2)}{2/\pi} \Big) , \label{Var inf sGinibre v3}
\end{equation}
where 
\begin{equation}
J_\nu(z)= \sum_{k=0}^\infty (-1)^k \frac{(z/2)^{2k+\nu}}{ k! \Gamma(k+\nu+1) } 
\end{equation}
is the Bessel function of the first kind and 
\begin{equation}
\textbf{H}_\nu(z) =\sum_{k=0}^\infty (-1)^k \frac{ (z/2)^{2k+\nu+1} }{ \Gamma(k+\frac32) \Gamma(\nu+k+\frac32) }
\end{equation}
is the Struve function \cite[Chapters 10, 11]{olver2010nist}.

\begin{rem}
\label{rem:Iprime}
Let $c(R):= R^2 e^{-2 R^2} ( I_0(2 R^2) + I_1(2 R^2) )$.
Then we have 
\begin{equation} \label{c'(R) Bessel}
c'(R)= 2 R e^{-2R^2} I_0(2R^2), \qquad 
 c(R)= \frac{1}{8} c''(R)+\Big( R-\frac{1}{8R} \Big) c'(R), 
\end{equation}
which follows from $I_0^\prime(x)=I_1(x)$ and $xI_1^\prime(x)=xI_0(x)-I_1(x)$.
We mention that the function $c(R)$ also appears in several different contexts on the studies of the Ginibre ensembles, see e.g. \cite[Proposition 1]{can2019random} and \cite[Theorem 2.1]{byun2021real}.
\end{rem}

For the limiting variance and covariance of the GinOE at the origin we give explicit expressions as functions of $R$ in Proposition~\ref{Prop_Expression var inf}, containing up to two-fold integrals.
Here, let us state just their asymptotic behaviour.

\begin{thm}[\textbf{Asymptotic number variance of the infinite GinOE in the origin limit}]\label{Thm_GinOE inf}
As $R \to \infty$, we have 
\begin{align}
 \label{Var NR DR inf}
&\Var \NN_\R^{(1)}(D_R) \sim \frac{ 2\sqrt{2}-2 }{\sqrt{\pi}}  R, 
\\
\label{Var NC DR inf}
& \Var \NN_\C^{(1)} (D_R)  \sim \frac{2\sqrt{2} }{ \sqrt{\pi} }  R,
\\
\label{GinOE inf Cov asy}
&\Cov \Big( \NN_\C^{(1)}(D_R),   \NN_\R^{(1)}(D_R) \Big)  \sim -\frac{ 2\sqrt{2}-2 }{\sqrt{\pi}}  R .
\end{align}
In particular, put together we have 
\begin{equation}\label{linear behaviour real}
\Var \NN^{(1)}(D_R) \sim \frac{2}{\sqrt{\pi}} R.
\end{equation}
\end{thm}

Figure~\ref{Fig_InfGin_Var} illustrates the variance for the three ensembles.

\begin{rem}
The negative covariance between real and complex eigenvalues can be
intuitively understood as follows. Looking at the joint density \eqref{jpdf-kl-1}, it vanishes both when a complex eigenvalue (pair) becomes real, from the factor in the second line, and when it comes close to a real eigenvalue,
from the two factors in the first line.
The same negative correlation can also be expected on a larger scale
away from the origin, e.g.\ from the expected number outside the limiting
support in Corollary \ref{Cor_outside the disc number}, eq.~\eqref{EN detailed}: Whenever the number of real eigenvalues increases (by an $O(\sqrt{N})$), the number of complex eigenvalues has to decrease by the same amount, in order to keep the total balance in \eqref{EN disc GinUE}. This is in agreement with the total variance vanishing in Conjecture~\ref{Conj_GinOE}.
\end{rem}

\begin{rem}\label{remark_Poisson}
By Proposition~\ref{Prop_EN inf} we have that as $R \to 0$,  
\begin{equation}
E^{(4)}(R) \sim 2R^4, \qquad E_\R^{(1)}(R) \sim  \sqrt{ \frac{2}{\pi} } R.  
\end{equation}
On the one hand, it follows from \eqref{linear behaviour Sym 0} and \eqref{Var NR DR inf closed} that as $R \to 0$,   
\begin{equation}
\Var \NN^{(4)}(D_R)  \sim  2 R^4, \qquad \Var \NN_\R^{(1)}(D_R) \sim \sqrt{ \frac{2}{\pi} } R. 
\end{equation}
Thus, one can observe the behaviours:
\begin{equation}
\lim_{R \to 0} \frac{\Var \NN^{(4)}(D_R)}{E^{(4)}(R)} = \lim_{R \to 0} \frac{\Var \NN_\R^{(1)}(D_R)}{E_\R^{(1)}(R)} = 1,
\end{equation}
which is a characteristic feature of the Poisson point processes. Indeed, for the GinUE and GinSE, one can show that not only the variance but also all the higher cumulants are asymptotically the same as their mean value in this regime. This follows from the asymptotic results on the FCS (see \cite{lacroix2019rotating} for the GinUE case and Theorem~\ref{Thm_higher cumulants} for the GinSE case) -- see the Remark~\ref{rem_Poisson} below. Furthermore, the above asymptotic behaviour again indicates the possibility of Poissonian statistics also for the GinOE in the origin limit.
\end{rem}

\begin{figure}[t]
    \begin{subfigure}{0.48\textwidth}
		\begin{center}
			\includegraphics[width=\textwidth]{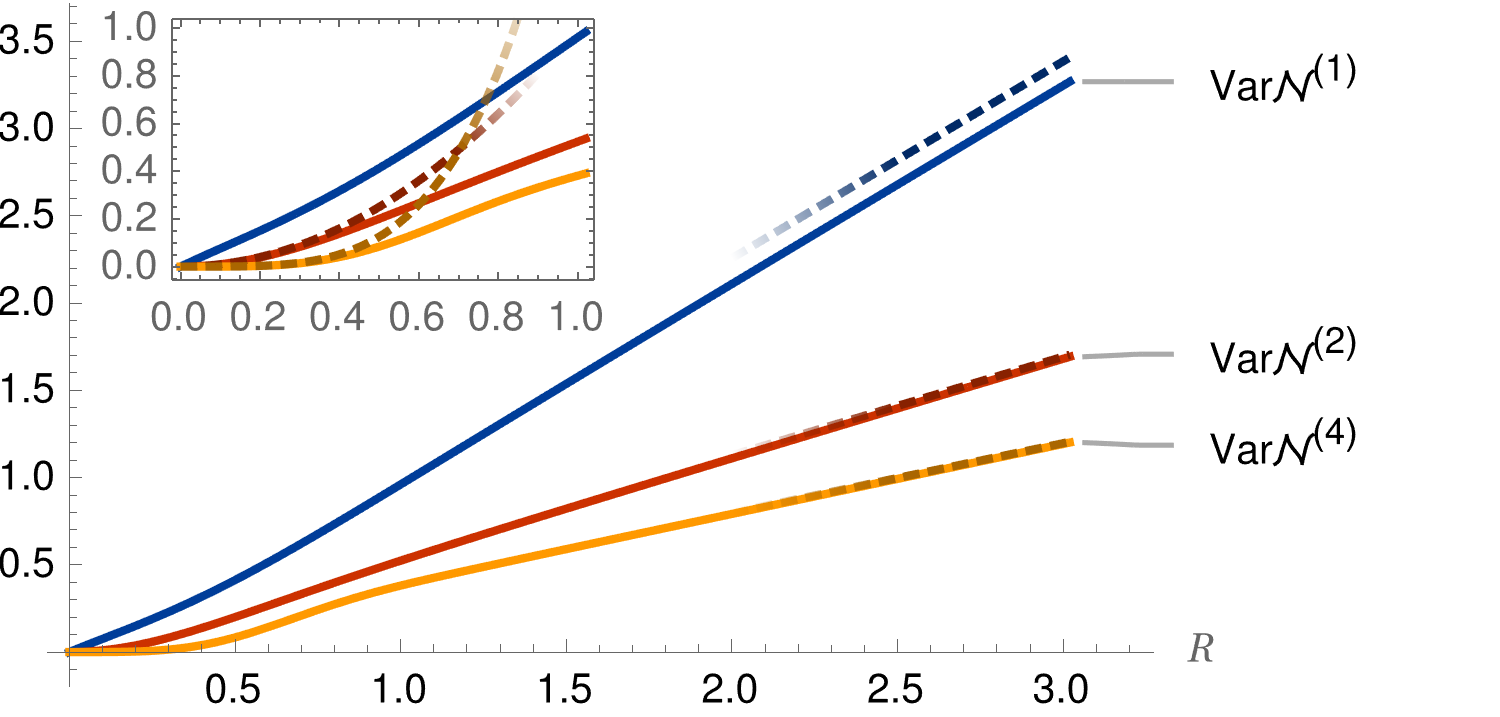}
		\end{center}
		\subcaption{Variance for GinOE, GinUE and GinSE}
	\end{subfigure}
	\begin{subfigure}{0.48\textwidth}
		\begin{center}
			\includegraphics[width=\textwidth]{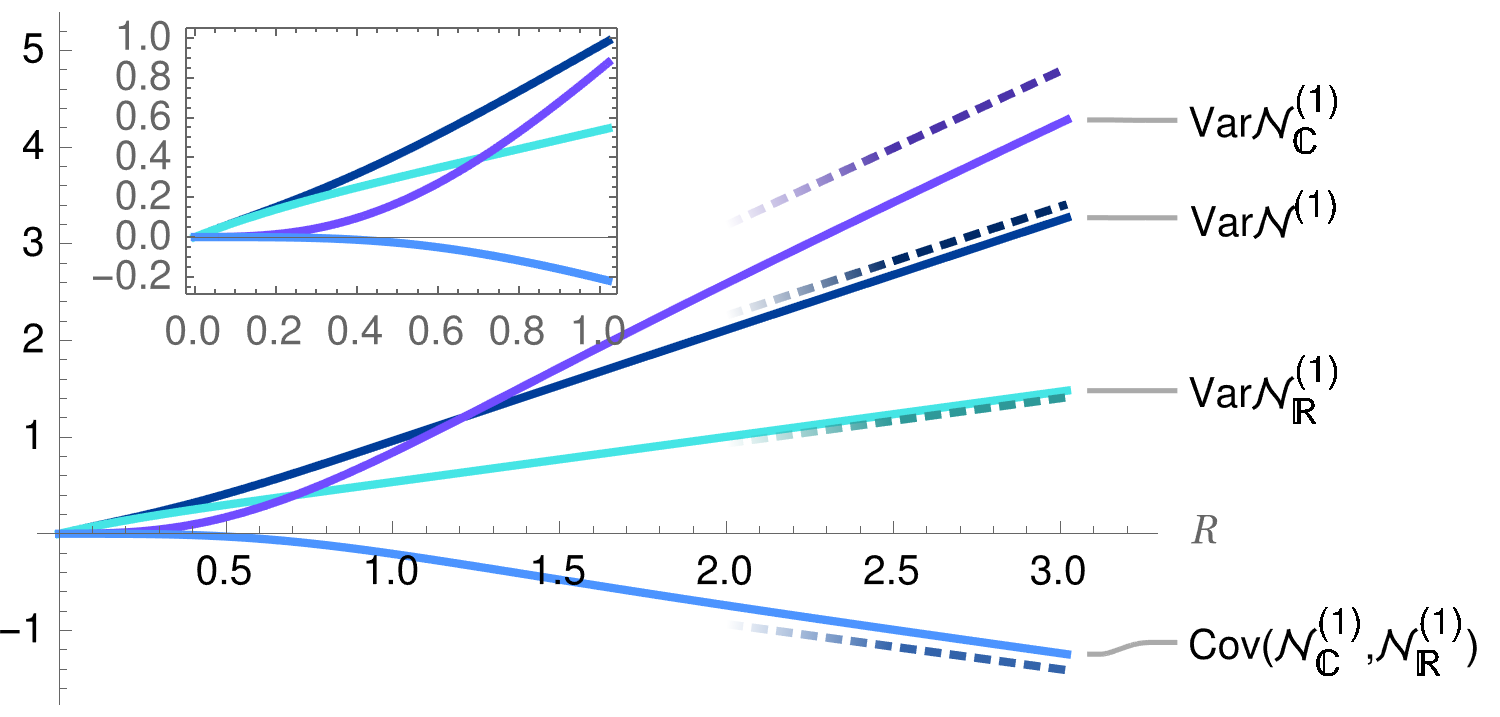}
		\end{center}
		\subcaption{Real and complex (co)variances for GinOE}
    \end{subfigure}
    \caption{Graphs of the number variance of the infinite Ginibre ensembles, cf.\ Proposition~\ref{Prop_Var inf csGinibre} and Proposition~\ref{Prop_Expression var inf}.
    Dashed lines indicate the $R \to 0$ and $R \to \infty$ asymptotics, cf.\ \eqref{Variance GinUSOE inf asy}, Proposition~\ref{Prop_Var inf csGinibre} and Proposition~\ref{Thm_GinOE inf}.
    }
    \label{Fig_InfGin_Var}
\end{figure}

\begin{rem}
By the asymptotic behaviours \eqref{linear behaviour Cplx}, \eqref{linear behaviour Sym} and \eqref{linear behaviour real}, we have for all three Ginibre ensembles the following universal behaviour in the bulk of the spectrum
\begin{equation} \label{Variance GinUSOE inf asy}
\Var \, \NN^{(\beta)}( D_R) \sim   \frac{R}{\sqrt{\pi}}  \frac{2}{\mathfrak{c}(\beta)}  \qquad R \to \infty, 
\end{equation}
where $\mathfrak{c}(\beta)$ is given by 
\begin{equation} \label{c(beta) 124}
\mathfrak{c}(\beta)= \begin{cases}
    1 & \beta=1,
    \\
    2 & \beta=2,
    \\
    2\sqrt{2} & \beta=4.
\end{cases}
\end{equation} 
By comparison with Remark~\ref{varbeta} below one can notice that this asymptotic formula is consistent with the bulk scaling limit \eqref{Var universal N}, with $a=R/\sqrt{N}$. 
\end{rem}

\subsubsection{Conjectured bulk and edge behaviour for the GinOE}

In this subsection we consider again the initial definition of the Ginibre ensembles at the beginning of this Section \ref{sec:main}, that is without rescaling the distances as in the origin limit. Hence, the limiting support of the spectrum is the centred disc $D_a$ with radius $a=1$.  
Let $\NN_a$ be the number of eigenvalues in $D_a$. 
For the finite GinOE it can be written as 
\begin{equation}
\NN_a= \NN_a^\C + \NN_a^\R, 
\end{equation}
where $\NN_a^\C$ (resp., $\NN_a^\R$) is the number of complex (resp.\ real) eigenvalues.

We mention that the counting statistics of real eigenvalues of the (elliptic) GinOE have been studied in several recent works \cite{fitzgerald2021fluctuations, MR3612267, byun2023large, forrester2023local, byun2021real,forrester2015diffusion}. Additionally, a very recent work \cite{goel2023central} has also investigated the counting statistics of  complex eigenvalues.
We also refer the reader to  \cite{cipolloni2021fluctuation,cipolloni2023central,rourke2016central,kopel2015linear} and references therein for the study of fluctuations of the linear statistics. 

For the GinUE and GinSE the limiting variance $\Var \NN_a^{(\beta=2,4)}$ is known in the bulk ($0<a<1$) and edge scaling limits. It was shown to be universal for non-Gaussian rotational invariant potentials in \cite{lacroix2019rotating} and \cite{akemann2022universality} respectively, see Remark~\ref{varbeta} below.  

For the GinOE, as in \eqref{VarDef1} we have
\begin{equation}
\Var \NN_a^{(1)} = \Var \NN_a^\C+ \Var \NN_a^\R + 2 
\Cov( \NN_a^\C, \NN_a^\R ).
\label{var-a-Def}
\end{equation}
Then we have the following conjecture for the scaling limits of each contribution to the different scaling regimes:

\begin{conj}[{\bf GinOE number variance scaling limits}]\label{Conj_GinOE}
We expect the following as $N \to \infty$.
\begin{itemize}
    \item \textup{\textbf{Bulk:}} For $a \in (0,1)$ fixed, we have 
    \begin{equation}
      \frac{ \Var \NN_a^\C }{ \sqrt{N/\pi} } \sim 2\sqrt{2} \, a , \qquad  \frac{ \Var \NN_a^\R }{ \sqrt{N/\pi} } \sim (2\sqrt{2}-2) \, a, \qquad   2\,\dfrac{  \Cov( \NN_a^\C, \NN_a^\R )  }{  \sqrt{N/\pi} } \sim -2(2\sqrt{2}-2) \, a. 
    \end{equation}
    In particular, for the total variance we obtain
    \begin{equation}
    \frac{ \Var \NN_a^{(1)} }{ \sqrt{N/\pi}  } \sim 2\, a .
    \end{equation}  
    \item \textup{\textbf{Edge:}} In the scaling limit $a=1-\SS/\sqrt{2N}$, we have
       \begin{equation}
    \frac{ \Var \NN_a^{(1)} }{ \sqrt{N/\pi}  } \sim 2 f(\SS) , \qquad f(\SS)  :=  \sqrt{ 2\pi }   \int_{-\infty}^{\SS} \frac{\erfc(t)\erfc(-t)}{4}  \,dt. 
    \end{equation} 
      \item \textup{\textbf{Outside:}} For $a>1$ fixed, we have
            \begin{equation}
\frac{\Var \NN_a^\C}{\sqrt{N/\pi}} \sim 2\sqrt{2}-2,\qquad   \frac{\Var \NN_a^\R}{\sqrt{N/\pi}} \sim 2 \sqrt{2}-2, \qquad 2\frac{ \Cov( \NN_a^\C, \NN_a^\R ) }{ \sqrt{N/\pi} } \sim -2(2\sqrt{2}-2),
\end{equation}
leading to the vanishing of the total variance
    \begin{equation}
    \frac{ \Var \NN_a^{(1)} }{ \sqrt{N/\pi}  } \sim 0.
    \end{equation}
\end{itemize}
\end{conj}

Let us mention that the bulk part of Conjecture~\ref{Conj_GinOE} is supported by Theorem~\ref{Thm_GinOE inf} and the outside part by \cite{forrester2007eigenvalue}, these are also summarised in Table~\ref{Table_GinOE variance}.
All parts of our conjecture are supported by numerical simulations, see Figure~\ref{Fig_GinOE var N}.

\begin{table}[h!]
    \centering    
\begin{center}
{\def\arraystretch{3}
\begin{tabular}{ |p{2.75cm}|p{2.75cm}|p{2.75cm}|p{2.75cm}|p{2.75cm}|  }
 \hline
& \centering $\dfrac{\Var \NN_a^\C}{\sqrt{N/\pi}} $ & \centering   $\dfrac{\Var \NN_a^\R}{\sqrt{N/\pi}} $  & \centering $2\,\dfrac{  \Cov( \NN_a^\C, \NN_a^\R )  }{  \sqrt{N/\pi} }$   & \hspace{1.5em} $\dfrac{\Var \NN_a^{(1)}}{\sqrt{ N/\pi}} $
 \\ 
\hline 
\centering $a\in (0,1)$ & \centering $2\sqrt{2} \, a$ & \centering $(2\sqrt{2}-2)\,a$ & \centering $-2(2\sqrt{2}-2)\, a$ & \hspace{2.75em}$2\, a$
 \\ 
   \hline 
\centering $a>1$ & \centering $2\sqrt{2}-2$ & \centering $2\sqrt{2}-2$ & \centering $-2(2\sqrt{2}-2)$ & \hspace{2.75em} $0$
 \\ 
 \hline
\end{tabular}
}
\end{center}   
    \caption{Leading order asymptotic behaviour of the individual contributions to the number variance in the GinOE in the bulk with fixed $0<a<1$ and outside the support for $a>1$.}
    \label{Table_GinOE variance}
\end{table}

\begin{figure}[t]
    \begin{subfigure}{0.5\textwidth}
		\begin{center}
			\includegraphics[width=\textwidth]{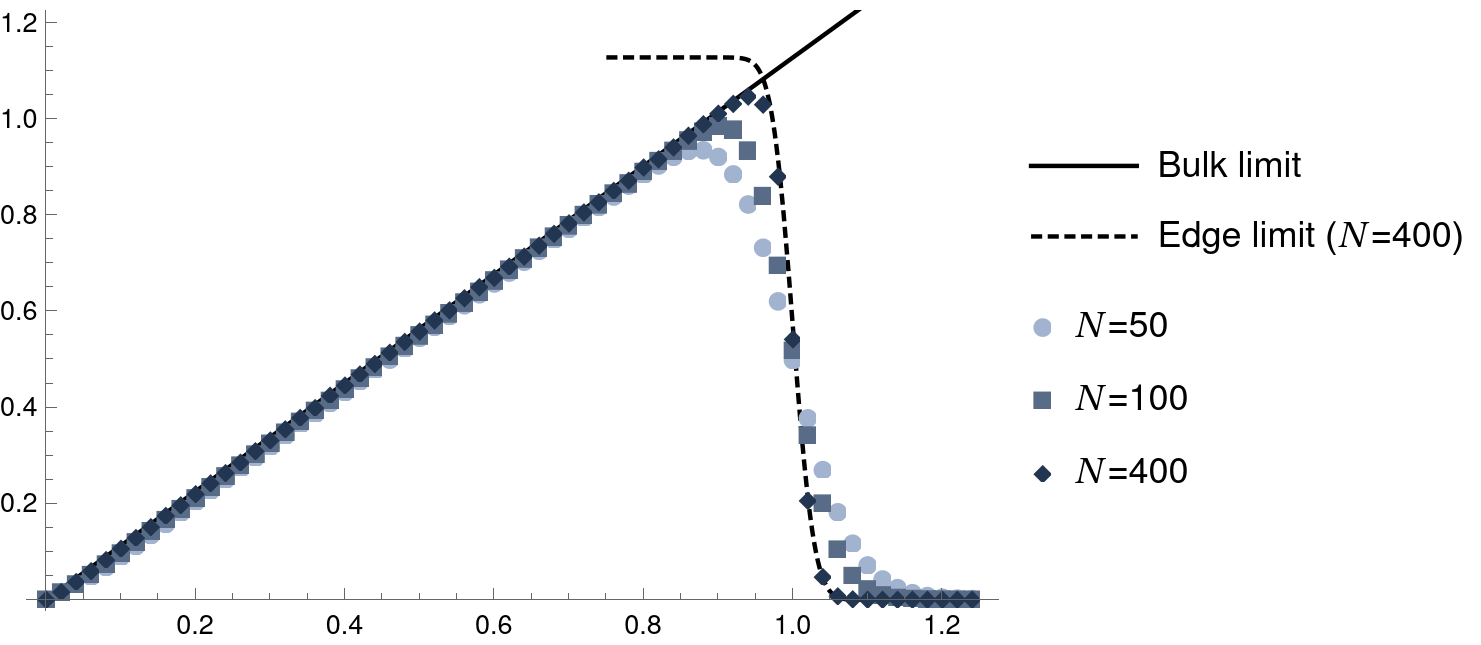}
		\end{center}
		\subcaption{Plot of $\Var \NN_a / \sqrt{N}$}
	\end{subfigure}
	
	\begin{subfigure}{0.32\textwidth}
		\begin{center}
			\includegraphics[width=\textwidth]{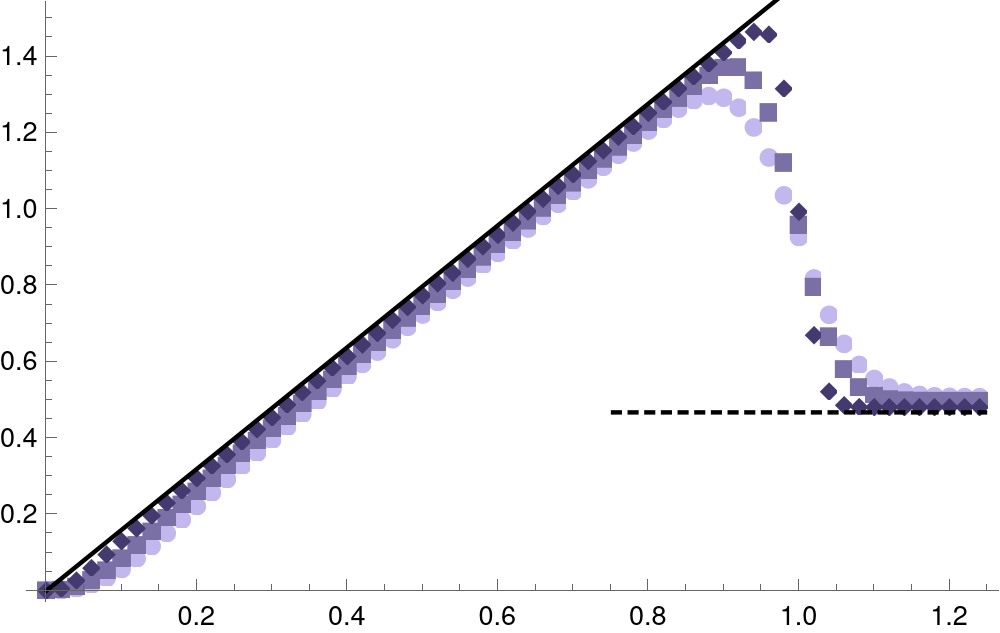}
		\end{center}
		\subcaption{Plot of $\Var \NN_a^\C / \sqrt{N}$}
	\end{subfigure}
	\begin{subfigure}{0.32\textwidth}
		\begin{center}
			\includegraphics[width=\textwidth]{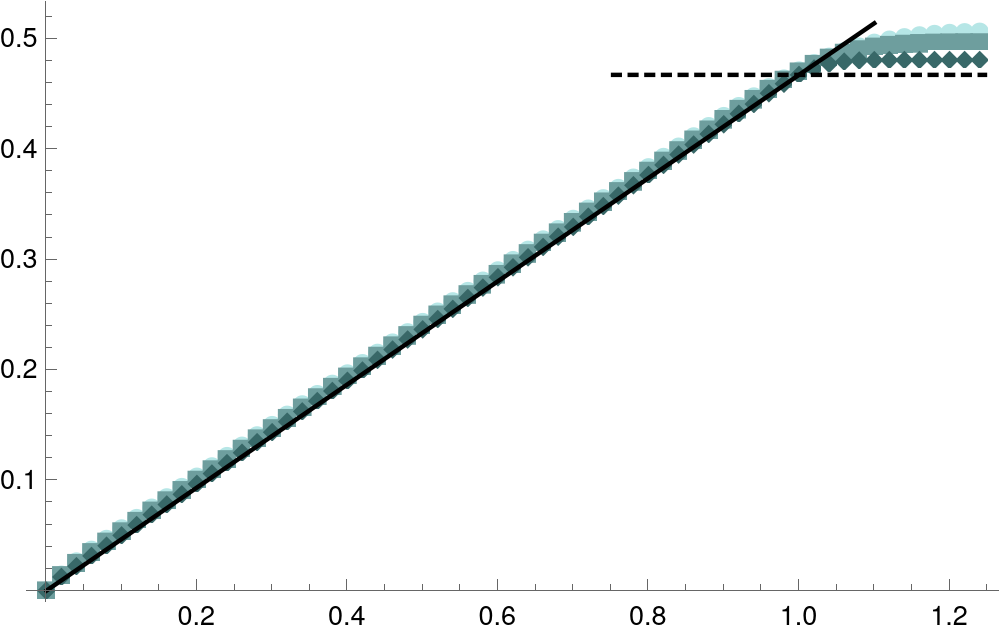}
		\end{center}
		\subcaption{Plot of $\Var \NN_a^\R / \sqrt{N}$}
	\end{subfigure}
	\begin{subfigure}{0.32\textwidth}
		\begin{center}
			\includegraphics[width=\textwidth]{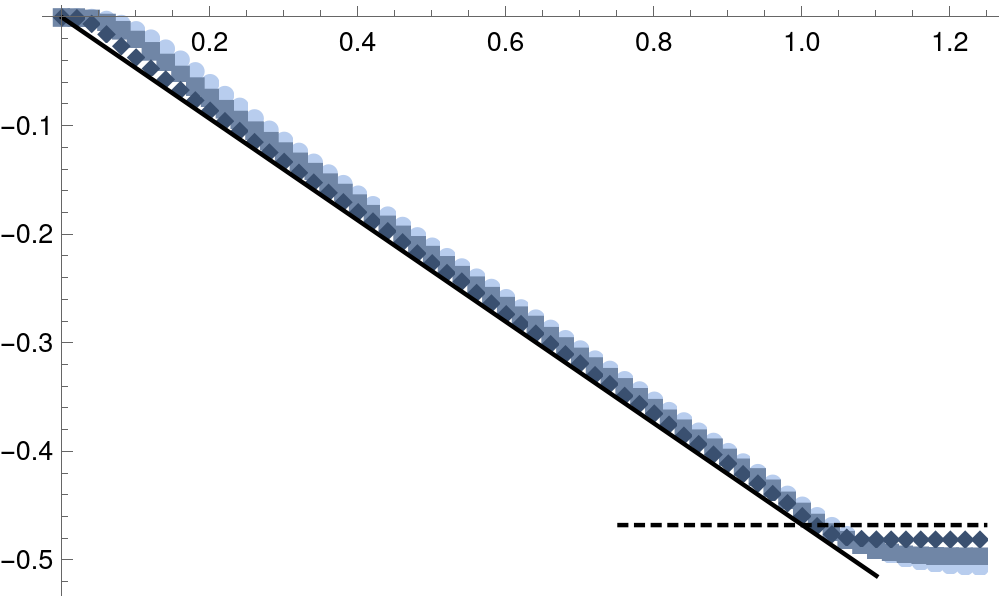}
		\end{center}
		\subcaption{Plot of $\Cov( \NN_a^\C, \NN_a^\R ) / \sqrt{N}$}
	\end{subfigure}

	\caption{Comparison of each contribution to the number variance of the GinOE with numerical simulations for ensembles of size $N$.} \label{Fig_GinOE var N}
\end{figure}


\begin{rem}[\textbf{Bulk and edge universality}]\label{varbeta}
Conjecture~\ref{Conj_GinOE} suggests that the number variance for all three  Ginibre ensembles is universal in the bulk and at the edge and takes the form
\begin{equation} \label{Var universal N}
 \mathfrak{c}(\beta) \frac{ \Var \NN_a^{(\beta)} }{ \sqrt{N/\pi}  } \sim  
\begin{cases}
   \displaystyle 
2 a  & a \in (0,1)\quad \mbox{(bulk)},
\smallskip 
\\
2 f(\SS) & a = 1-\SS/\sqrt{2N} \quad \mbox{(edge),}
\end{cases}
\end{equation}
where $\mathfrak{c}(\beta)$ is given in \eqref{c(beta) 124}.
The variance formula for the GinUE $(\beta=2)$ and GinSE $(\beta=4)$ can be found in \cite[Theorem 1.2]{akemann2022universality} for more general potentials $W(|z|)$, where the variance and scaling variable $\SS$ have to be rescaled by $1/\sqrt{\Delta W(a)}$, the inverse mean level spacing, to preserve the universal form.  
More detailed asymptotic behaviours for these cases can be found in \cite{charlier2022asymptotics,byun2022characteristic}. 
\end{rem}

\subsection{Full counting statistics of planar symplectic ensembles}\label{subsec:FCS-SE}
\hfill\\
For comparison we present both the FCS for the non-Gaussian extension of the GinUE and GinSE class, the so-called planar ensembles.
The determinantal Coulomb gas ensemble (random normal matrix ensemble) 
is defined by the joint density \eqref{jpdf2}, replacing the Gaussian potential $|z|^2$ in the exponent by a rotationally invariant external potential, $W \equiv W_N : \C \to \R$ which possibly depends on $N$. Its integrable structure \eqref{Rk beta2N} remains valid, with the kernel given in terms of planar orthogonal polynomials with respect to $e^{-NW}$, cf.~\cite{byun2022progress}.

Likewise, the same generalisation is made in the Pfaffian Coulomb gas ensemble based on the joint density \eqref{jpdf4}, called planar symplectic ensembles.
Also here, the integrable structure \eqref{Rk beta4N} remains, with the corresponding skew-kernel \eqref{kappaN} given in terms of planar skew-orthogonal polynomials, cf.~\cite{byun2023progress}.

Following \cite{akemann2022universality}, we shall use the following terminology. 

\begin{defi}\label{Qsuit}
A rotationally invariant potential $W_N(z) =g_N(|z|)$ will be called suitable, if it satisfies the following conditions:
\begin{enumerate}
\item $g_N(r) \gg \log r$ as $r \to \infty$;
\smallskip
\item $g_N \in C^2(0,1]$;
\smallskip
    \item $r g'_N(r)$ increases on $(0,\infty)$;
\smallskip
\item
$\lim\limits_{r \to 0} r g_N'(r)=0$, and $\lim_{N \to \infty}g'_N(1)=4$. 
\end{enumerate}
\end{defi}
Note that condition \textit{(3)} is equivalent to $\Delta W_N(z) > 0$ for $z \neq 0$ (here, $\Delta = \partial_z \overline{\partial_z}= \frac14 (\partial_x^2+\partial_y^2)$ for $z=x+iy$ is one-quarter of the usual Laplacian) and condition \textit{(4)} implies that the limiting droplet is the unit disk.

Let us write
\begin{equation} \label{ortho norm}
h_j:=\int_\C  e^{-N W_N(z)}\,|z|^{2j}\,dA(z) = 2 \int_0^\infty  e^{-N g_N(r)}\,r^{2j+1}\,dr,
\end{equation}
for the non-vanishing $j$-th moments, which are at the same time the squared norms of planar orthogonal polynomials (monomials).
We also write
\begin{align}
\begin{split} \label{ortho norm trunc}
h_{j,1}(a)&:= \int_{ |z|<a } e^{-N W_N(z)} |z|^{2j}\,dA(z)=  2 \int_0^a e^{-N g_N(r)} r^{2j+1}  \,dr,
\\
h_{j,2}(a)&:= \int_{ |z|>a } e^{-N W_N(z)} |z|^{2j}\,dA(z)=  2 \int_a^\infty e^{-N g_N(r)} r^{2j+1}  \,dr 
\end{split}
\end{align}
for the truncated moments or squared norms.
We also denote 
\begin{equation}
\mathcal{L}_{j}(a)=\frac{ h_{j,1}(a) }{ h_j }, \qquad \mathcal{M}_{j}(a)=\frac{ h_{j,2}(a) }{ h_j }.
\end{equation}

Let us recall that  $\NN_a := \# \{z_j: |z_j| \le a \}$ denotes the number of particles in the centred disc $D_a$ of radius $a.$
Here, 
\begin{equation} \label{polylog function}
\Li_s(x) = \sum_{k = 1}^{\infty} k^{-s} x^k, \quad |x|<1, 
\end{equation}
is the definition of the polylogarithm function, that also converges for $|x|=1$ and $\Re(s)>1$. It is given by analytic continuation otherwise.

\begin{prop}[\textbf{Finite-$N$ expressions of the cumulant generating functions}] \label{Prop_cumulant symplectic}
Let $W$ be a suitable potential, then we have the following for $a \geq 0$ and $N \in \N$. 
\begin{itemize}
    \item[\textup{(i)}] \textbf{\textup{(Random normal matrix ensembles)}} We have for the moment generating function
\begin{equation}
\mathbb{E}_{N,W}^{(\beta=2)}\Big[  e^{u \NN_a } \Big] = \prod_{j=0}^{N-1} \Big(e^{u} \mathcal{L}_{j}(a)+ \mathcal{M}_{j}(a) \Big) . 
\end{equation}
Furthermore, for any $p \ge 2$, the $p$-th cumulant of $\NN_a$ has the expression
\begin{equation} \label{cumulants_complex}
\kappa_p^{(\beta=2)}(a) = (-1)^{p+1} \sum_{j = 0}^{N-1} \Li_{1 - p}\Big( 1 - \frac{1}{\mathcal{L}_j(a)} \Big).
\end{equation}
We also have for the mean
\begin{equation}\label{mean value complex}
E_{N,W}^{(\beta=2)}(a):=\mathbb{E}_{N,W}^{(\beta=2)}[\NN_a]=\sum_{j = 0}^{N-1} \mathcal{L}_j(a).
\end{equation}
    \item[\textup{(ii)}] \textbf{\textup{(Planar symplectic ensembles)}} We have
\begin{equation} \label{eq:CGF symplectic}
\mathbb{E}_{N,W}^{(\beta=4)}\Big[  e^{u  \NN_a } \Big] = \prod_{j=0}^{N-1} \Big( e^{u}\mathcal{L}_{2j+1}(a)+ \mathcal{M}_{2j+1}(a) \Big).
\end{equation}
Furthermore, for any $p \ge 2$, the $p$-th cumulant of $\NN_a$ has the expression
\begin{equation} \label{cumulants symplectic}
\kappa_p^{(\beta=4)}(a) = (-1)^{p+1} \sum_{j = 0}^{N-1} \Li_{1 - p}\Big( 1 - \frac{1}{\mathcal{L}_{2j+1}(a)} \Big).
\end{equation}
For the mean, we have
\begin{equation} \label{mean value symplectic}
E_{N,W}^{(\beta=4)}(a):=\mathbb{E}_{N,W}^{(\beta=4)}[\NN_a] = \sum_{j = 0}^{N-1} \mathcal{L}_{2j+1}(a).
\end{equation}
\end{itemize}
\end{prop}

Note that Proposition~\ref{Prop_cumulant symplectic} (i) was shown in \cite{lacroix2019rotating}.

\begin{thm}[\textbf{FCS of the planar symplectic ensembles}]\label{Thm_higher cumulants}
Let $W$ be a suitable potential.
Then as $N \to \infty$, the following holds.
\begin{enumerate}[label=(\roman*)]
    \item \textup{\textbf{Bulk:}} For $a\in (0,1)$ fixed and $p \geq 2$, we have 
    \begin{equation}
  \lim_{N \to \infty}  \sqrt{ \frac{2}{N \Delta W(a)} } \, \kappa_p^{(\beta=4)}(a) = a\,\kappa_p^{\textup{bulk}},
    \end{equation}
    where
    \begin{equation}\label{kappa_p}
    \begin{split}
    \kappa_p^{\textup{bulk}} :=&\, (-1)^{p + 1} \int_{-\infty}^{\infty} \Li_{1 - p}\Big( -\frac{\erfc(-x)}{\erfc(x)} \Big) \, dx 
    \\
     =&     \begin{cases}
\displaystyle        - \int_{-\infty}^{\infty} \Li_{1 - p}\Big( -\frac{\erfc(-x)}{\erfc(x)} \Big) \, dx, & \textup{if }p \textup{ is even},
        \smallskip 
        \\
        0 , & \textup{if }p \textup{ is odd}.
        \end{cases}
    \end{split}
    \end{equation}
    \item \textup{\textbf{Edge:}} For $\SS \in \R$, let 
    \begin{equation} \label{a edge scaling}
    a=1-\frac{ \SS }{ \sqrt{ 2 \Delta W(1)  N} }.
    \end{equation}
    Then we have 
    \begin{equation}
\lim_{N \to \infty}    \sqrt{ \frac{2}{ N \Delta W(1) } } \, \kappa_p^{(\beta=4)}(a) = \kappa_p^{\textup{edge}}(\SS),
    \end{equation}
    where 
      \begin{equation}
        \kappa_p^{\textup{edge}}(\SS) := (-1)^{p + 1} \int_{-\infty}^{\SS} \Li_{1 - p}\Big( -\frac{\erfc(-x)}{\erfc(x)} \Big) \, dx.
    \end{equation}
\end{enumerate}
\end{thm}

\begin{rem}[FCS of the Mittag-Leffler ensembles] \label{Rem_ML FCS}
The Mittag-Leffler ensembles are two-parameter generalisations of the Ginibre ensembles associated with the potential of the form
\begin{equation}\label{ML potential}
W^{\rm ML}(z) = \alpha |z|^{2b} - \frac{2c}{N} \log|z|, \qquad (\alpha,b>0, \,\, c>-1). 
\end{equation}
For the complex Mittag-Leffler ensembles, a very precise FCS was obtained in a recent work by Charlier \cite[Corollary 1.6]{charlier2022asymptotics}.
We recall that the FCS is equivalent to the problem of deriving a large-$N$ expansion of the partition functions having jump-type singularities (see \cite{MR4610282} and references therein for various literature on the partition functions of random normal matrix and planar symplectic ensembles).
As a consequence, by \cite[Proposition 5.10]{byun2023progress}, the FCS of the complex Mittag-Leffler ensembles induces that of the symplectic Mittag-Leffler ensembles (with different parameters). 
For instance, the FCS of the symplectic Ginibre ensemble follows from the FCS of the complex Mittag-Leffler ensemble associated with the linear potential $W^{\rm ML}(z)=2|z|$.
\end{rem}

In Figure~\ref{Fig:higher order cumulants} we illustrate Proposition \ref{Prop_cumulant symplectic} and Theorem \ref{Thm_higher cumulants} for $p=3,4$ using the potentials
\begin{align}
W_N^{(m)}(z) &:= -\frac{1}{N} \log \MeijerG{m , 0}{0,  m}{ -  \\  \bfs{0}}{ (2 N)^m |z|^2 }, \qquad m \in \N, \label{Q products} \\
W^{\rm tr}(z) &:=
\begin{cases} -2 \tilde{c} \, \log \Big( 1-\dfrac{|z|^2}{1+{ \tilde{c} } } \Big) &\textup{if } |z| \le \sqrt{1+\tilde{c}},
\smallskip
\\
\infty &\textup{otherwise},
\end{cases} \qquad \tilde{c}>0. \label{Q truncated unitary strong}
\end{align}
Here, $G^{m,n}_{p,q}$ is the Meijer $G$-function defined by 
\begin{equation} \label{Meijer G}
\MeijerG {m,n} {p,q} { a_1,\dots,a_p\\
   b_1,\dots,b_q| } {z}:= \frac{1}{2\pi i} \oint_L  z^s \frac{     \prod_{j=1}^{m}\Gamma(b_j-s)   \prod_{j=1}^{n}\Gamma(1-a_j+s)      }{  \prod_{j=1+m}^{q}\Gamma(1-b_j+s) \prod_{j=1+n}^{p}\Gamma(a_j-s)       }\,ds,
\end{equation}
where the integration contour $L$ in $\C$ depends on $\{a_j\}, \{ b_j\}$, see e.g.\ \cite[Chapter 16]{olver2010nist}.
The potential \eqref{Q products} describes the product of $m$ symplectic Ginibre matrices ($m=1$: GinSE), see e.g.~\cite{MR3066113}.
On the one hand, the potential \eqref{Q truncated unitary strong} stems from the truncated symplectic unitary ensemble, see e.g.~\cite{khoruzhenko2021truncations}.

\begin{ex}
For $p = 2$ it follows from $\Li_{-1}(1-1/x)=x(x-1)$ that for the GinSE
\begin{equation} \label{variance symplectic}
\Var \NN_a^{(\beta=4)} = \sum_{j = 0}^{N-1} \mathcal{L}_{2j+1}(a) \mathcal{M}_{2j+1}(a). 
\end{equation}
This formula appears in \cite[Proposition 1.1]{akemann2022universality}. 
Furthermore we have
\begin{equation}
- \Li_{-1}\Big( -\frac{\erfc(-x)}{\erfc(x)} \Big) = \frac{1-\erf(x)^2}{4}, \qquad \int_\R \frac{1-\erf(x)^2}{4}  \,dx = \frac{1}{ \sqrt{2\pi} }.
\end{equation}
Therefore, one can observe that Theorem~\ref{Thm_higher cumulants} recovers \cite[Theorem 1.2]{akemann2022universality}.
\end{ex}

\begin{figure}[t]
	\begin{subfigure}{0.32\textwidth}
		\begin{center}
			\includegraphics[width=\textwidth]{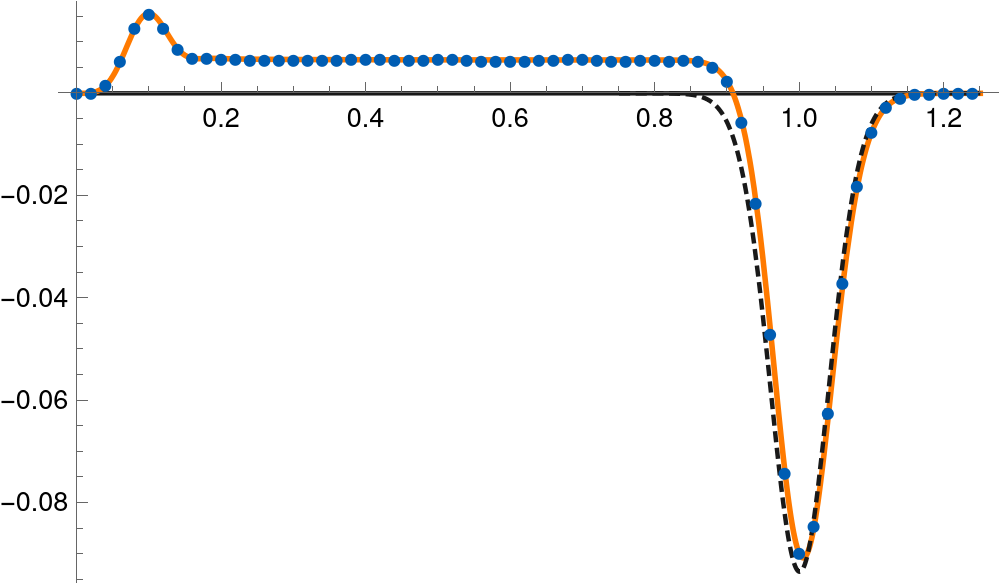}
		\end{center}
		\subcaption{$m=1$ ($\tilde{c}=\infty$), $p=3$}
	\end{subfigure} 
	\begin{subfigure}{0.32\textwidth}
		\begin{center}
			\includegraphics[width=\textwidth]{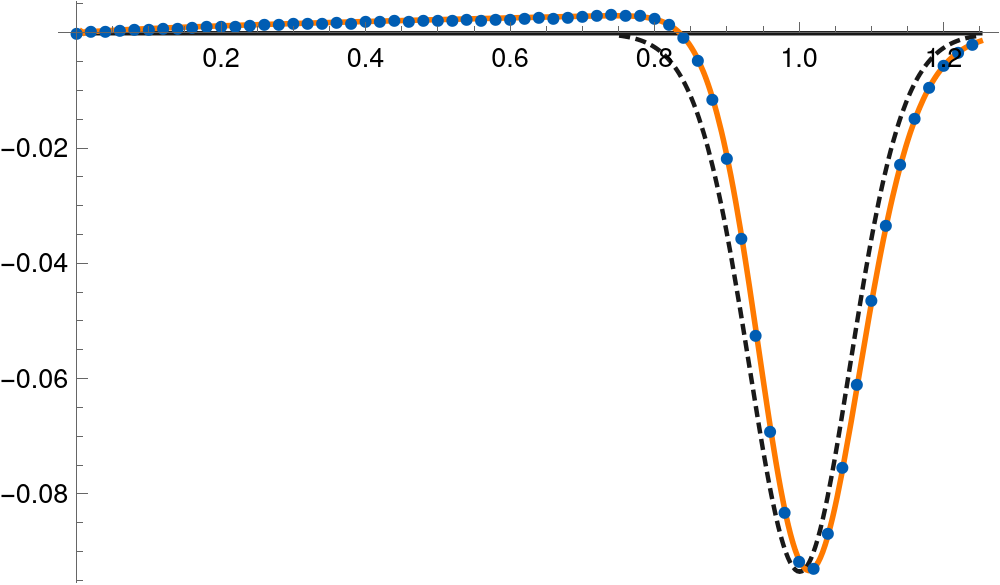}
		\end{center}
		\subcaption{$m=3$, $p=3$}
	\end{subfigure} 
\begin{subfigure}{0.32\textwidth}
		\begin{center}
			\includegraphics[width=\textwidth]{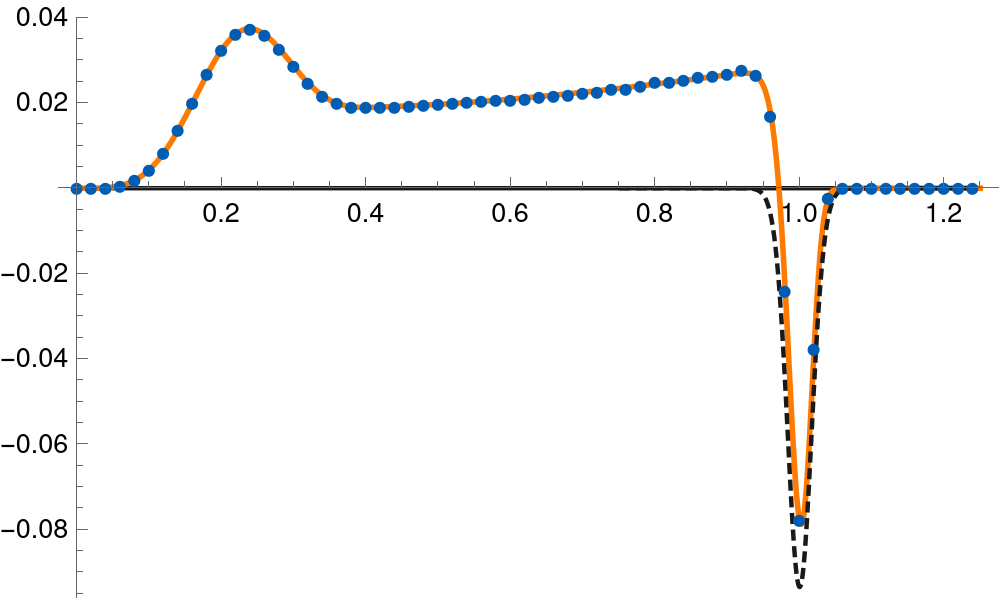}
		\end{center}
		\subcaption{$\tilde{c}=0.2$, $p=3$}
	\end{subfigure}  
 
	\begin{subfigure}{0.32\textwidth}
		\begin{center}
			\includegraphics[width=\textwidth]{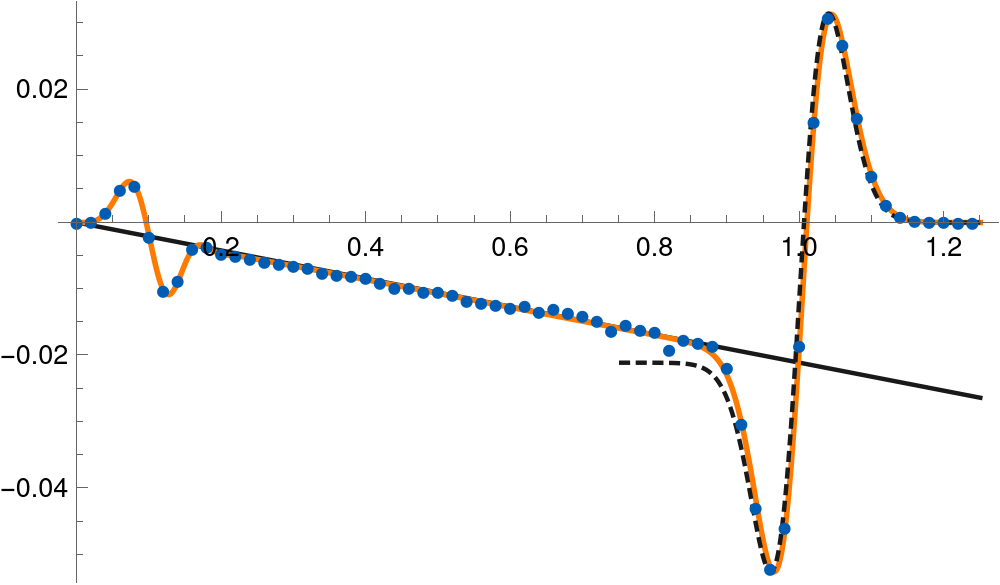}
		\end{center}
		\subcaption{$m = 1$ ($\tilde{c}=\infty$), $p=4$}
	\end{subfigure}
	\begin{subfigure}{0.32\textwidth}
		\begin{center}
			\includegraphics[width=\textwidth]{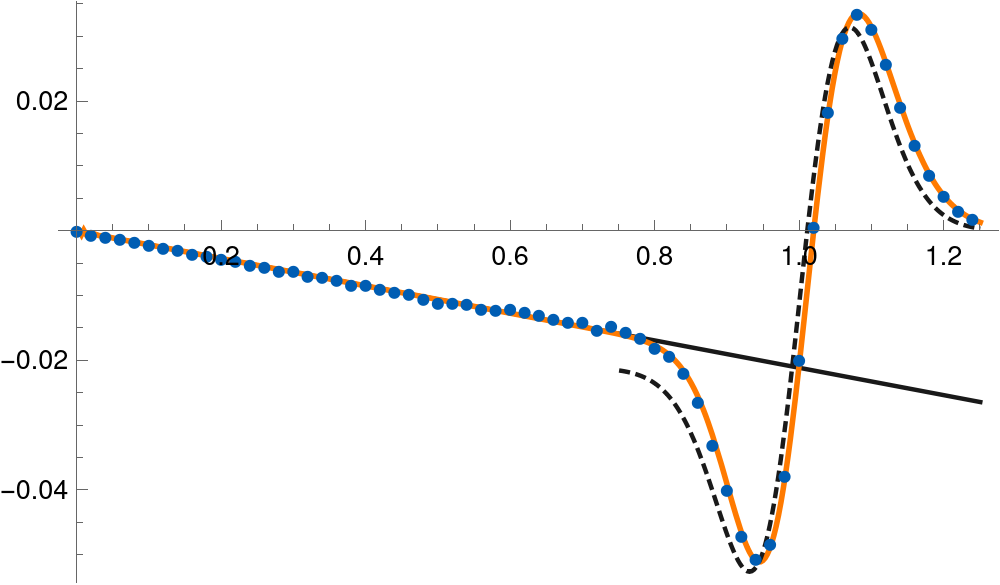}
		\end{center}
		\subcaption{$m=3$, $p=4$}
	\end{subfigure}
	\begin{subfigure}{0.32\textwidth}
		\begin{center}
			\includegraphics[width=\textwidth]{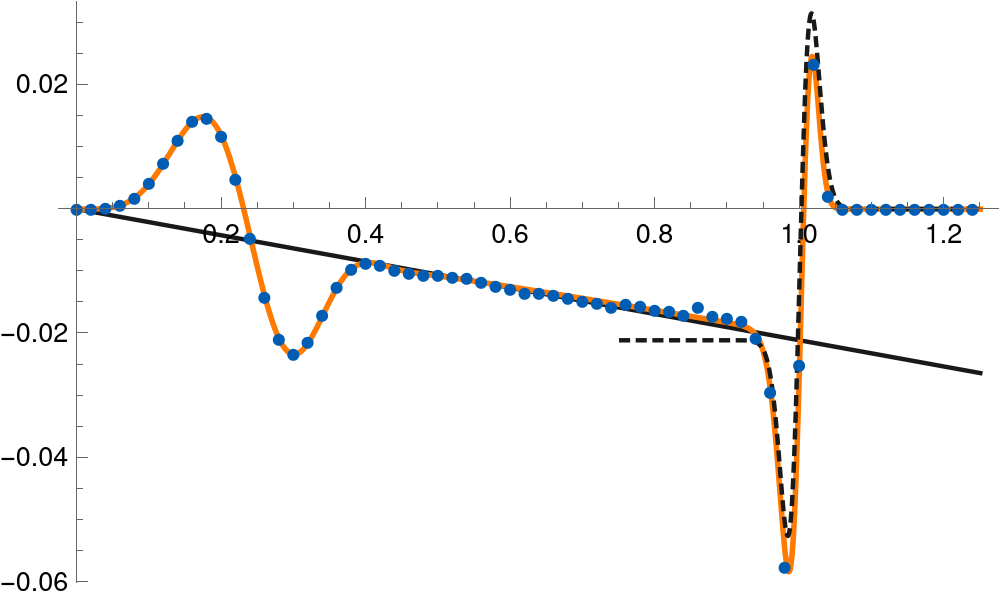}
		\end{center}
		\subcaption{$\tilde{c}=0.2$, $p=4$}
	\end{subfigure}
	\caption{Higher order cumulants for the product of $m$ symplectic Ginibre matrices, see \eqref{Q products}, and for the symplectic truncated unitary ensemble, see \eqref{Q truncated unitary strong}, with $N=50$. In blue: the $p$-th cumulant $\kappa_p^{(\beta=4)}(a)$ computed from $25{,}000{,}000$ random samples. In orange: right hand side of \eqref{cumulants symplectic}. In black: bulk limit (solid line) and edge limit (dashed line) from Theorem~\ref{Thm_higher cumulants}.}
    \label{Fig:higher order cumulants}
\end{figure}

\begin{rem}[Origin limit of the cumulants]\label{rem_Cumulant_origin}
In the origin limit where we rescale $R = \sqrt{N} \, a$, we obtain for cumulants of the GinSE
\begin{equation}
\lim_{N \to \infty} \kappa_p^{(\beta=4)}\Big(\frac{R}{\sqrt{N}}\Big) = (-1)^{p+1} \sum_{j = 0}^{\infty} \Li_{1 - p}\Big( -\frac{\Gamma(2j+2, 2R^2)}{\gamma(2j+2, 2R^2)} \Big).
\end{equation}
This follows from \eqref{cumulants symplectic}, where the convergence of the infinite sum is assured by the Gamma-function asymptotics \cite[Eq.~(8.11.4)]{olver2010nist}.
An analogous formula for the GinUE was already shown in the supplementary material of \cite{lacroix2019rotating}.
Note that the origin limit is not universal, e.g.\ for the potentials in \eqref{ML potential}, \eqref{Q products} and \eqref{Q truncated unitary strong}  
different results were obtained 
in \cite[Propositions~1.4, 1.5, and 1.7]{akemann2022universality}, respectively.
\end{rem}


\begin{table}[t]
\begin{center}
{\def\arraystretch{2}
\begin{tabular}{|cc||c|c||c|c||c|}
	\hline
	& & \multicolumn{2}{c||}{GinUE} & \multicolumn{2}{c||}{GinSE} & GinOE \\
	& & Ginibre & Potential $W(\abs{z})$ & Ginibre & Potential $W(\abs{z})$ & \\
	\hline\hline
	\multirow{4}{*}{\rotatebox[origin=c]{90}{Expectation}} & Finite $N$ & \makecell{\cite[Prop.~3.1]{akemann2022universality} \\ Prop.~\ref{Prop_EN GinOE}(i)} & \makecell{\cite{lacroix2019intermediate}, \cite[Prop.~1.1]{akemann2022universality} \\ Prop.~\ref{Prop_cumulant symplectic}(i)} & \makecell{\cite[Prop.~3.1]{akemann2022universality} \\ Prop.~\ref{Prop_EN GinOE}(i)} & \makecell{\cite[Prop.~1.1]{akemann2022universality} \\ Prop.~\ref{Prop_cumulant symplectic}(ii)} & \makecell{Prop.~\ref{Prop_EN GinOE}(iii) \\ Lem.~\ref{Prop_GinOE expected real}} \\
	\cline{2-7}
	& Origin & \makecell{\cite{lacroix2019rotating} \\ Prop.~\ref{Prop_EN inf}} & \makecell{Depends on $W$ \\ e.g.\ \cite[Sec.~1.3]{akemann2022universality}} & Prop.~\ref{Prop_EN inf} & \makecell{Depends on $W$ \\ e.g.\ \cite[Sec.~1.3]{akemann2022universality}} & Prop.~\ref{Prop_EN inf} \\
	\cline{2-7}
	& Bulk & \makecell{\cite[Prop.~3.1]{akemann2022universality},  \\ \cite[Cor.~1.6 (a)]{charlier2022asymptotics}}    & via \cite{HM13} & \cite[Prop.~3.1]{akemann2022universality} & \makecell{via \cite{MR2934715} \\ \cite[Cor.~1.3]{akemann2022universality}} & \makecell{\cite{MR1231689} \\ Eq.~\eqref{EN GinOE leading order}} 
 \\
	\cline{2-7}
	& Edge & \makecell{\cite[Cor.~1.6 (b)]{charlier2022asymptotics}, \\ \cite{lacroix2019rotating}, Cor.~\ref{Cor_outside the disc number}} & via \cite{MR3975882} &  \makecell{via Rem.~\ref{Rem_ML FCS} \\ Cor.~\ref{Cor_outside the disc number}}  & \makecell{ Unknown } &  \makecell{Unknown \\ Cor.~\ref{Cor_outside the disc number}}
 \\
	\hline\hline
	\multirow{4}{*}{\rotatebox[origin=c]{90}{Variance}} & Finite $N$ & \cite{lacroix2019intermediate,lacroix2019rotating} & \cite{lacroix2019intermediate}, \cite[Prop.~1.1]{akemann2022universality} & \cite[Eq.~(3.8)]{akemann2022universality} & \cite[Prop.~1.1]{akemann2022universality} & Rem.~\ref{Rem_GinOE var finiteN} \\
	\cline{2-7}
	& Origin & \makecell{\cite{lacroix2019rotating} \\ Prop.~\ref{Prop_Var inf csGinibre}(i)} & \makecell{Depends on $W$ \\ e.g.\ \cite[Sec.~1.3]{akemann2022universality}} & Prop.~\ref{Prop_Var inf csGinibre}(ii) & \makecell{Depends on $W$ \\ e.g.\ \cite[Sec.~1.3]{akemann2022universality}} & Prop.~\ref{Prop_Expression var inf} \\
	\cline{2-7}
	& Bulk & \multicolumn{2}{c||}{\cite{lacroix2019intermediate,lacroix2019rotating}, \cite[Thm.~1.2]{akemann2022universality}, Rem.~\ref{varbeta}} & \multicolumn{2}{c||}{\cite[Thm.~1.2]{akemann2022universality}, Rem.~\ref{varbeta}} & \makecell{Conj.~\ref{Conj_GinOE} \\ (Thm.~\ref{Thm_GinOE inf})} \\
	\cline{2-7}
	& Edge & \multicolumn{2}{c||}{\cite{lacroix2019rotating}, \cite[Thm.~1.2]{akemann2022universality}, Rem.~\ref{varbeta}} & \multicolumn{2}{c||}{\cite[Thm.~1.2]{akemann2022universality}, Rem.~\ref{varbeta}} & Conj.~\ref{Conj_GinOE} \\
	\hline\hline
	\multirow{4}{*}{\rotatebox[origin=c]{90}{Cumulants}} & Finite $N$ & \multicolumn{2}{c||}{\cite{lacroix2019intermediate,lacroix2019rotating}, Prop.~\ref{Prop_cumulant symplectic}(i)} & \multicolumn{2}{c||}{Prop.~\ref{Prop_cumulant symplectic}(ii)} & \makecell{Partly unknown,  \\ cf. \cite{forrester2015diffusion}  } \\
	\cline{2-7}
	& Origin & \cite{lacroix2019rotating} & Depends on $W$ & Rem.~\ref{rem_Cumulant_origin} & Depends on $W$ & Unknown \\
	\cline{2-7}
	& Bulk & \cite{lacroix2019intermediate,lacroix2019rotating} & \cite{lacroix2019intermediate} & \multicolumn{2}{c||}{Thm.~\ref{Thm_higher cumulants}(i)} & Unknown \\
	\cline{2-7}
	& Edge & \cite{lacroix2019rotating} & \cite{lacroix2019rotating} with \cite{Lacroix2018extremes} & \multicolumn{2}{c||}{Thm.~\ref{Thm_higher cumulants}(ii)} & Unknown \\
	\hline
\end{tabular}
}
\end{center}
\caption{Summary of our main results and (some) references to known results.}
\label{tab:Summary of Results}
\end{table}

\begin{rem}\label{rem_Poisson}
It is interesting to study the behaviour of the cumulants for the GinUE and the GinSE, given respectively in Eqs. (\ref{cumulants_complex}) and (\ref{cumulants symplectic}) in the limit $a \to 0$, with fixed $N$. Indeed, in this limit, one has
\begin{equation}\label{Lj_small_a}
{\mathcal L}_j(a) = \frac{e^{-Ng(0)}}{h_j} \frac{a^{2j+2}}{2j+2}(1 + o(1)), \quad a \to 0 \;.
\end{equation}
Note that for $n \in \mathbb{N},$ it follows from the inversion formula \cite[Eq. (1)]{jonquiere1889note}
\begin{equation}\label{Li-inversion}
\Li_{-n}(x) = (-1)^{n-1} \Li_{-n}(1/x)= (-1)^{n-1} \sum_{k=1}^\infty \frac{k^n}{x^k} ,
\end{equation}
that as $x \to 0,$
\begin{equation}
\label{asympt_polylog}
{\rm Li}_{1-p}(-x) = \frac{(-1)^{p+1}}{x} + O\Big(\frac{1}{x^2} \Big), \quad p \geq 2 \quad {\rm and} \quad p \in \mathbb{N}.
\end{equation}
Therefore one finds, as $a \to 0$, 
\begin{equation}\label{kappa_small_a}
\kappa_p^{(\beta=2)}(a) = \frac{e^{-Ng(0)}}{2h_0} a^2(1 + o(1)), \qquad \kappa_p^{(\beta=4)}(a) = \frac{e^{-Ng(0)}}{4h_1} a^4(1 + o(1)) \;.
\end{equation}
On the other hand, from Eqs. \eqref{mean value complex} and \eqref{mean value symplectic} one finds
\begin{align}
\begin{split}
\label{average_small_a}
& E_{N,W}^{(\beta=2)}(a) = {\mathcal L}_0(a)(1+o(1)) = \frac{e^{-Ng(0)}}{2h_0} a^2(1+o(1)),
\\
& E_{N,W}^{(\beta=4)}(a) = {\mathcal L}_1(a)(1+o(1)) = \frac{e^{-Ng(0)}}{4h_1} a^4(1+o(1)) \;.
\end{split}
\end{align}
Therefore, by comparing Eqs. \eqref{kappa_small_a} and \eqref{average_small_a} we see that as $a \to 0$, all the cumulants coincide with the mean value, i.e.,
\begin{equation}\label{Poisson}
\lim_{a \to 0} \frac{\kappa_p^{(\beta=2)}(a)}{E_{N,W}^{(\beta=2)}(a)} = \lim_{a \to 0} \frac{\kappa_p^{(\beta=4)}(a)}{E_{N,W}^{(\beta=4)}(a)} = 1, \qquad p \geq 2 \;,
\end{equation}
which shows, as anticipated in the Remark \ref{remark_Poisson}, that ${\mathcal N}_a$ indeed becomes a Poissonian random variable in this limit $a\to 0$, keeping $N$ fixed. We emphasize that the behaviour of the cumulants found in the limit $N \to \infty$, keeping $a \in (0,1)$ fixed in Theorem \ref{Thm_higher cumulants}, is quite different. Indeed, although the functional dependence on the variable $a$ of all the cumulants is similar, the coefficients $\kappa^{\rm bulk}_p$ in (\ref{kappa_p}) have a non-trivial dependence on $p$. This shows that in this limit the fluctuations of ${\mathcal N}_a$ are clearly non-Poissonian (see Ref. \cite{lacroix2019intermediate} for a similar discussion in the context of the GinUE).    
\end{rem}

Our main results are summarized in Table~\ref{tab:Summary of Results}.

\section{Expected number of the Ginibre ensembles}\label{sec:proofmean}

\subsection{Finite Ginibre ensembles}
\hfill\\
The Lemmas~\ref{Prop_GinOE expected real} and \ref{Prop_GinOE expected complex} below imply Proposition~\ref{Prop_EN GinOE}. 

\begin{lem} \label{Prop_GinOE expected real}
For $N$ odd, 
\begin{equation}
\begin{split} \label{expected real eigenvalues N odd}
E_{N,\R}^{(1)}(a)& = \frac{1}{ \sqrt{2} } \sum_{k=0}^{N-2} \frac{ (2k-1)!! }{ (2k)!! }  P\Big(k+\frac12, Na^2\Big) +P\Big(\frac{N}{2}, \frac{Na^2}{2}\Big) 
\\
&\quad -   \frac{1}{ 2^{N/2} \, (N-2)!! } \sum_{k=0}^{ (N-3)/2 } \frac{ (N+2k-2)!! }{2^k \, (2k)!!} P\Big( \frac{N}{2}+k, Na^2\Big).
\end{split}
\end{equation}
For $N$ even,
\begin{equation}
\begin{split} \label{expected real eigenvalues N even}
E_{N,\R}^{(1)}(a)& = \frac{1}{ \sqrt{2} } \sum_{k=0}^{N-2} \frac{ (2k-1)!! }{ (2k)!! }  P\Big(k+\frac12, Na^2\Big)  +\frac{ N^{N/2}   }{ (N-2)!!  }     \int_{0}^a  e^{ -\frac{N}{2} x^2 } x^{N-1} \erf\Big(\sqrt{\frac{N}{2}}x\Big)\,dx
\\
&\quad -\frac{1}{ 2^{N/2} \, (N-2)!! } \sum_{k=1}^{ N/2-1 } \frac{  (N+2k-3)!! }{ 2^{k-1/2} \, (2k-1)!! }   P\Big( \frac{N-1}{2}+k, Na^2\Big).  
\end{split}
\end{equation}
\end{lem}

Note that for any $N \in \N$
\begin{equation}\label{P2gamma}
 \frac{1}{ \sqrt{2} } \sum_{k=0}^{N-2} \frac{ (2k-1)!! }{ (2k)!! }  P\Big(k+\frac12, Na^2\Big)= \sqrt{ \frac{2}{\pi} } \frac{ a\sqrt{N} \,  \Gamma(N-1,Na^2)+\gamma(N-\tfrac12, Na^2) }{ (N-2)! }. 
\end{equation}

\begin{rem}
If $a=\infty$, $E_{N,\R}^{(1)}=E_{N,\R}^{(1)}(\infty)$ corresponds to the expected number of real eigenvalues. 
For $a=\infty$, we have that for $N$ odd, 
\begin{align}
E_{N,\R}^{(1)} &=1+\frac{1}{ \sqrt{2} } \sum_{k=0}^{N-2} \frac{ (2k-1)!! }{ (2k)!! }- \frac{ 1   }{ 2^{N/2 } (N-2)!!  }  \sum_{k=0}^{ (N-3)/2 } \frac{ (N+2k-2)!!  }{ 2^k \,(2 k)!!} 
\end{align}
whereas  for $N$ even,  using the series expansion of the error function inside the integral we obtain
\begin{align}
\begin{split}
E_{N,\R}^{(1)} &= \frac{1}{ \sqrt{2} } \sum_{k=0}^{N-2} \frac{ (2k-1)!! }{ (2k)!! } + \frac{2}{ \sqrt{\pi} } \frac{ \Gamma( \frac{N+1}{2} ) }{ \Gamma( \frac{N}{2} ) } {}_2 F_1\Big( \frac12, \frac{N+1}{2} ; \frac32; -1 \Big)
\\
&\quad  -        \frac{1}{ 2^{N/2} \, (N-2)!! } \sum_{k=1}^{ N/2-1 } \frac{  (N+2k-3)!! }{ 2^{k-1/2} \, (2k-1)!! }. 
\end{split}
\end{align}
One can check that these correspond to the expressions \cite{MR1231689}
\begin{equation}
E_{N,\R}^{(1)} = 
\begin{cases}
\displaystyle   1+ \sqrt{2} \sum_{k=1}^{(N-1)/2} \frac{(4k-3)!!}{ (4k-2)!! } &N \textup{ odd},
 \smallskip
   \\
\displaystyle     \sqrt{2} \sum_{k=0}^{N/2-1} \frac{(4k-1)!!}{ (4k)!! } &N \textup{ even}.
\end{cases}
\end{equation}
\end{rem}

\begin{proof}[Proof of Lemma~\ref{Prop_GinOE expected real}]
Inserting the $1$-point function $\textbf{R}_{N,1,\R}^{(1)}(x)$ of the real eigenvalues from \eqref{RN1real1} into the definition of $E_{N,\R}^{(1)}(a)$ from \eqref{EN RC 1 a Def} gives
\begin{equation}
E_{N,\R}^{(1)}(a)= \int_{-a}^{a} \Big(\sqrt{\frac{N}{2\pi}} Q(N-1,Nx^2) +  \frac{ N^{N/2}   }{ 2^{N/2 } \Gamma(\frac{N}{2})  } e^{ -\frac{N}{2} x^2 } |x|^{N-1}  P\Big( \frac{N-1}{2}, \frac{Nx^2}{2} \Big) \Big)\,dx, 
\end{equation}
for any $N \in \N.$
For the first term we use
\[
Q(N-1,Nx^2) = e^{-Nx^2} \sum_{j=0}^{N-2} \frac{N^k x^{2k}}{k!}, \qquad
\Gamma\Big(k+\frac12\Big)= \sqrt{\pi} \frac{(2k-1)!!}{2^k},
\]
and obtain
\begin{align}
\begin{split}
&\quad \sqrt{\frac{N}{2\pi}}\int_{-a}^a  Q(N-1,Nx^2)\,dx
=\sqrt{\frac{N}{2\pi}} \int_{-a}^{a}  e^{ -N x^2 } \sum_{k=0}^{N-2} \frac{ N^k x^{2k} }{k!}\,dx 
\\
&= \frac{1}{ \sqrt{2\pi} } \sum_{k=0}^{N-2} \frac{ \gamma(k+\frac12, Na^2) }{k!} =  \frac{1}{ \sqrt{2\pi} } \sum_{k=0}^{N-2} \frac{ \Gamma(k+\frac12) }{k!}  P\Big(k+\frac12, Na^2\Big)  
\\
&=  \frac{1}{ \sqrt{2} } \sum_{k=0}^{N-2} \frac{ (2k-1)!! }{ 2^{k} k! }  P\Big(k+\frac12, Na^2\Big) . 
\end{split}
\end{align}
On the other hand, using integration by parts we obtain 
\begin{align}
\begin{split}
\sqrt{\frac{N}{2\pi}}\int_{-a}^a  Q(N-1,Nx^2)\,dx  & 
= \frac{1}{ \sqrt{2\pi} } \int_{0}^{Na^2} \frac{ Q(N-1,x) }{ \sqrt{x} }\,dx
\\
&= \sqrt{ \frac{2}{\pi} } \frac{ a\sqrt{N} \,  \Gamma(N-1,Na^2)+\gamma(N-\tfrac12, Na^2) }{ (N-2)! }.  
\end{split}
\end{align}
This implies the identity \eqref{P2gamma}. 

For the second term we need to distinguish between even and odd $N$ because
\begin{align*}
P\Big( \frac{N-1}{2}, \frac{Nx^2}{2} \Big) &= 
\begin{cases}
\displaystyle   1- e^{ -\frac{N}{2}x^2 } \sum_{k=0}^{ (N-3)/2 } \frac{(N x^2/2)^k }{k!} &\textup{for odd }N,
\smallskip 
\\
\displaystyle  \erf\Big(\sqrt{\frac{N}{2}}|x|\Big)-\frac{ e^{ -\frac{N}{2}x^2 } }{ \sqrt{\pi} } \sum_{k=1}^{N/2-1} \frac{ (N x^2/2)^{k-1/2} }{ (1/2)_k }   &\textup{for even }N,
\end{cases}
\end{align*}
see \cite[Section 8.4]{olver2010nist}.
Therefore if $N$ is odd, 
\begin{align*}
&\quad  \int_{-a}^a e^{ -\frac{N}{2} x^2 } |x|^{N-1}  P\Big( \frac{N-1}{2}, \frac{Nx^2}{2} \Big) \,dx 
\\
&= 2\int_{0}^a  e^{ -\frac{N}{2} x^2 } x^{N-1}\,dx -2\sum_{k=0}^{ (N-3)/2 } \frac{(N /2)^k }{k!} \int_{0}^a  e^{ -N x^2 } x^{N-1+2k}      \,dx  
\\
&= 2^{N/2} N^{-N/2} \gamma\Big(\frac{N}{2}, \frac{Na^2}{2}\Big) - N^{-N/2}\sum_{k=0}^{ (N-3)/2 } \frac{ 1 }{2^k \,k!} \gamma\Big( \frac{N}{2}+k, Na^2\Big)
\end{align*}
and consequently
\begin{align*}
&\quad \frac{ N^{N/2}   }{ 2^{N/2 } \Gamma(\frac{N}{2})  } \int_{-a}^a    e^{ -\frac{N}{2} x^2 } |x|^{N-1}  P\Big( \frac{N-1}{2}, \frac{Nx^2}{2} \Big) \,dx 
\\
&= P\Big(\frac{N}{2}, \frac{Na^2}{2}\Big) -   \sum_{k=0}^{ (N-3)/2 } \frac{ \Gamma(\frac{N}{2}+k) }{2^{k+N/2}  \Gamma(\frac{N}{2})\,k!} P\Big( \frac{N}{2}+k, Na^2\Big)
\\
&= P\Big(\frac{N}{2}, \frac{Na^2}{2}\Big) -  \frac{1}{ 2^{N/2} \, (N-2)!! } \sum_{k=0}^{ (N-3)/2 } \frac{ (N+2k-2)!! }{2^{2k} \,k!} P\Big( \frac{N}{2}+k, Na^2\Big)
\end{align*}
where we used \cite[Eq.~(5.4.2)]{olver2010nist} to write the Gamma-functions in terms of double factorials.
This implies  \eqref{expected real eigenvalues N odd}, since $2^k k! = (2k)!!$.

On the other hand, if $N$ is even, 
\begin{align*}
&\quad \int_{-a}^a e^{ -\frac{N}{2} x^2 } |x|^{N-1}  P\Big( \frac{N-1}{2}, \frac{Nx^2}{2} \Big) \,dx 
\\
&= 2\int_{0}^a  e^{ -\frac{N}{2} x^2 } x^{N-1} \erf\Big(\sqrt{\frac{N}{2}}x\Big)\,dx -\frac{2}{ \sqrt{\pi} }\sum_{k=1}^{ N/2-1 } \frac{(N /2)^{k-1/2} }{ (1/2)_k } \int_{0}^a  e^{ -N x^2 } x^{N-2+2k}      \,dx  
\\
&=  2\int_{0}^a  e^{ -\frac{N}{2} x^2 } x^{N-1} \erf\Big(\sqrt{\frac{N}{2}}x\Big)\,dx -\frac{  \sqrt{2} }{ \sqrt{\pi} } N^{-N/2} \sum_{k=1}^{ N/2-1 } \frac{1 }{ (2k-1)!! }   \gamma\Big( \frac{N-1}{2}+k, Na^2\Big),  
\end{align*}
and including the normalisation we have 
\begin{align*}
&\quad \frac{ N^{N/2}   }{ 2^{N/2 } \Gamma(\frac{N}{2})  } \int_{-a}^a    e^{ -\frac{N}{2} x^2 } |x|^{N-1}  P\Big( \frac{N-1}{2}, \frac{Nx^2}{2} \Big) \,dx 
\\
&= \frac{ N^{N/2}   }{ 2^{N/2-1 } \Gamma(\frac{N}{2})  }     \int_{0}^a  e^{ -\frac{N}{2} x^2 } x^{N-1} \erf\Big(\sqrt{\frac{N}{2}}x\Big)\,dx 
\\
&\quad -        \frac{ 1   }{ 2^{(N-1)/2 } \Gamma(\frac{N}{2})  }   \frac{1}{ \sqrt{\pi} }  \sum_{k=1}^{ N/2-1 } \frac{  \Gamma(\frac{N-1}{2}+k) }{ (2k-1)!! }   P\Big( \frac{N-1}{2}+k, Na^2\Big).  
\end{align*}
This yields \eqref{expected real eigenvalues N even}, after rewriting the Gamma-functions using \cite[Eq.~(5.4.2)]{olver2010nist} in the second term.
\end{proof}


\begin{lem} \label{Prop_GinOE expected complex}
For $N \in \N$, we have for the radial density at finite-$N$
\begin{equation}\label{RNC1radial}
\frac{1}{N}\frac{1}{2\pi}\int_{ 0 }^{2\pi} \textbf{R}_{N,1,\C}^{(1)}\Big(\frac{r e^{i\theta}}{\sqrt{N}} \Big)\,d\theta = 
\Big( 1-e^{2r^2} \erfc(\sqrt{2}r) \Big) Q(N-1,r^2) . 
\end{equation}
This leads to 
\begin{align}\label{ENC1mean}
\begin{split}
E_{N,\C}^{(1)} (a) 
&= \sum_{k=0}^{N-2} P(k+1,Na^2)  - 2 \int_0^{ \sqrt{N}a }  e^{2r^2} \erfc(\sqrt{2}r)   Q(N-1,r^2)\,dr
\end{split}
\end{align}
\end{lem}
\begin{proof}
The $1$-point function $\textbf{R}_{N,1,\C}^{(1)}(x)$ of the complex eigenvalues is given by  
\begin{equation}
\frac{1}{N}\textbf{R}_{N,1,\C}^{(1)}\Big(\frac{z}{\sqrt{N}}\Big)= \sqrt{ {2}{\pi}  } |y| \erfc(\sqrt{2}|y|) e^{2y^2} Q(N-1,|z|^2).  
\end{equation}
Using
\begin{equation}
\erfc(\sqrt{2}y)= 2 \sqrt{ \frac{2}{\pi} } \int_{-\infty}^0 e^{-2(t-y)^2}\,dt,
\end{equation}
with $y=r\sin{\theta}$, we obtain 
for the integral \eqref{RNC1radial}
\begin{align}
\begin{split}
\frac{1}{N2\pi}\int_{ 0 }^{2\pi} \textbf{R}_{N,1,\C}^{(1)}\Big(\frac{r e^{i\theta}}{\sqrt{N}} \Big)\,d\theta &= \frac{2}{\pi} Q(N-1,r^2)\int_{-\infty}^0 \int_{ 0 }^{2\pi} |y| e^{-2(t-|y|)^2+2y^2  } \,d\theta \,dt
\\
&=  \frac{4r}{\pi}  Q(N-1,r^2)\int_{-\infty}^0  e^{-2t^2}  \int_{0}^{\pi} (\sin \theta) e^{ 4 t r\sin \theta  } \,d\theta \,dt
\\
&=  \frac{4r}{\pi}  Q(N-1,r^2)\int_{-\infty}^0  e^{-2t^2}  
\sum_{l=0}^\infty \frac{(4tr)^l}{l!}\frac{2^{l+1}\Gamma\big(\frac{l}{2}+1\big)^2}{\Gamma(l+2)} \,dt
\\
&=  \frac{4r}{\pi}  Q(N-1,r^2)\sum_{l=0}^\infty \frac{(8r)^l}{l!}\frac{\Gamma\big(\frac{l}{2}+1\big)^2}{\Gamma(l+2)}\frac{(-1)^l\Gamma\big(\frac{l+1}{2}\big)}{2^{\frac{l+1}{2}}},
\end{split}
\end{align}
using the elementary integrals
\begin{equation}
\label{sin-m-int}
\int_0^{\pi} (\sin \theta)^{l+1} \,d\theta = \frac{2^{l+1} \Gamma(\frac{l}{2}+1)^2}{\Gamma(l+2)}, \qquad
\int_{-\infty}^0 e^{-2t^2} t^l \,dt = \frac{(-1)^l}{2^{1+\frac{l+1}{2}}} \Gamma\Big(\frac{l+1}{2}\Big).
\end{equation}
The duplication formula
\begin{equation}\label{Gamma-duplo}
\Gamma(2z)= \frac{1}{\sqrt{\pi}} 2^{2z-1} \Gamma(z)\Gamma(z+\tfrac12),
\end{equation}
can be used to further simplify the sum, and we obtain
\begin{align}
\begin{split}
    \sqrt{\pi}\sum_{l=0}^\infty \frac{2^{\frac{3}{2}l-\frac{1}{2}}(-r)^l\Gamma\big(\frac{l}{2}+1\big)}{\Gamma(l+2)}&=
    \sqrt{\pi} \sum_{k=0}^\infty \frac{2^{3k-\frac{1}{2}r^{2k}\Gamma(k+1)}}{\Gamma(2k+2)}
    -\sqrt{\pi} \sum_{k=0}^\infty \frac{2^{2k+1}r^{2k+1}\Gamma\big(k+ \frac{3}{2}\big)}{\Gamma(2k+3)}
    \\
    &= \sqrt{\pi} \sum_{k=0}^\infty \frac{2^{2k-\frac{1}{2}}r^{2k}}{(2k+1)!!} -\pi \sum_{k=0}^\infty \frac{2^{k-1}r^{2k+1}}{(k+1)!}  , 
    \label{sum1}
\end{split}
    \end{align}
after splitting the sum into even and odd indices $l$ and applying once more the duplication formula \eqref{Gamma-duplo}.
Using the expansion of the error function \cite[Eq.~(7.6.2)]{olver2010nist},
\begin{equation}
    \erf(z)e^{z^2}=\frac{2}{\pi}\sum_{n=0}^\infty \frac{2^nz^{2n+1}}{(2n+1)!!}
\end{equation}
and $\erfc(z)=1-\erf(z)$ we obtain the following identity
\begin{align}
\frac{\pi}{4r}\Big(1-e^{2r^2}\erfc(\sqrt{2}r) \Big)    &=
\frac{\pi}{4r}\Big( 1-\sum_{n=0}^\infty \frac{(2r)^n}{n!}+\frac{2}{\sqrt{\pi}}\sum_{n=0}^\infty\frac{(2r)^{2n}\sqrt{2}r}{(2n+1)!!}\Big),
\end{align}
which agrees with the sum in \eqref{sum1}. This yields \eqref{RNC1radial}.

After having determined the radial density we can compute the mean in \eqref{ENC1mean} as
\begin{align}
\begin{split}
E_{N,\C}^{(1)} (a) 
&=  2 \int_0^{ \sqrt{N}a }  r \Big( 1-e^{2r^2} \erfc(\sqrt{2}r)  \Big) Q(N-1,r^2)\,dr
\\
&= 2 \sum_{k=0}^{N-2} \int_0^{ \sqrt{N}a }   e^{-r^2}\frac{r^{2k+1}}{k!} \,dr
- 2 \int_0^{ \sqrt{N}a }  e^{2r^2} \erfc(\sqrt{2}r)   Q(N-1,r^2)\,dr
\\
&= \sum_{k=0}^{N-2} P(k+1,Na^2)  - 2 \int_0^{ \sqrt{N}a }  e^{2r^2} \erfc(\sqrt{2}r)   Q(N-1,r^2)\,dr.
\end{split}
\end{align}
\end{proof}

\begin{proof}[Proof of Corollary~\ref{Cor_outside the disc number}]
For $\beta=2,4$ these are immediate consequences of Proposition~\ref{Prop_EN GinOE} (i) and (ii).
To be more precise, for $\beta=2$, it immediately follows from the Stirling's formula.
On the other hand, for $\beta=4$, the behaviour \eqref{EN disc GinSE} follows from $\beta=2$ case, when estimating the second term in \eqref{EN symplectic Ginibre} given by an exponential times a truncated sine-hyperbolic to be of order unity.  

It remains to show the case $\beta=1$. For this, using the well-known asymptotic behaviour 
$$\erfc(x) \sim \frac{ e^{-x^2} }{ \sqrt{\pi} x }, \qquad x \to \infty,$$ 
we have
\begin{align*}
&\quad 2 \int_0^{ \sqrt{N} } r\, e^{2r^2} \erfc(\sqrt{2}r)   Q(N-1,r^2)\,dr = 2 N \int_0^{ 1 } t\, e^{2t^2N} \erfc(\sqrt{2N}\,t )   Q(N-1,t^2N)\,dt
\\
& \sim \sqrt{ \frac{2}{\pi}  N} \int_0^{ 1 }    Q(N-1,t^2N)\,dt = \sqrt{ \frac{2}{\pi}  N} \Big(  \frac{ \gamma(N-\frac12, N) }{  \sqrt{N} \, \Gamma(N-1) } +Q(N-1,N) \Big) \sim \sqrt{ \frac{2}{\pi}  N}. 
\end{align*}
In the last line, we have used integration by parts. Using the asymptotic of the incomplete Gamma-functions \cite[Eq.~(8.10.11)]{olver2010nist}, each summand inside the bracket gives $1/2$, which yields the final result.
Combining Proposition~\ref{Prop_EN GinOE}, \eqref{EN GinOE leading order}, and \eqref{EN disc GinUE} for $\beta=2$, we obtain the desired asymptotic behavior \eqref{EN disc GinUE} for $\beta=1.$
\end{proof}

\subsection{Infinite Ginibre ensembles at the origin}

\begin{proof}[Proof of Proposition~\ref{Prop_EN inf}]
Note that for $\beta=2$, it is obvious since $R_1^{(2)}(z)=1$ \cite{ginibre1965statistical}. For $\beta=2,4$ we use  
\begin{equation}
E^{(\beta)}(R) = 2 \int_0^R r \widehat{R}_1^{(\beta)}(r)\,dr.
\end{equation}
For $\beta=4,$ note that \cite{MR1928853,Mehta}
\begin{equation}
R_1^{(4)}(z)= 4y F(2y), \qquad F(z)= e^{-z^2} \int_0^z e^{t^2}\,dt,
\end{equation}
where $F$ is Dawson's integral function. 
Then the desired identity \eqref{E DR S} follows from 
\begin{align*}
E^{(4)}(R)&=\int_{D_R} 4y F(2y) \,dA(z) =  \frac{4}{\pi} \int_{0}^R r^2 \,dr  \int_0^{2\pi} d\theta \,  \sin(\theta) F(2r \sin(\theta)).
\end{align*}
Using the series expansion of Dawson's function \cite[Eqs.~(7.5.1) and (7.6.3)]{olver2010nist} 
\begin{equation}
F(z)=\sum_{k=0}^\infty \frac{(-1)^k2^k}{(2k+1)!!}z^{2k+1},    
\end{equation}
we obtain for the angular integral
\begin{align*}
 \int_0^{2\pi} d\theta \,  \sum_{k=0}^\infty\frac{(-1)^k2^k}{(2k+1)!!}(2r)^{k+1}\sin(\theta)^{2k+2} &=
 \sum_{k=0}^\infty\frac{(-1)^k2^{5k+4}\Gamma(k+\frac{3}{2})^2}{(2k+1)!!\,\Gamma(2k+3)}r^{2k+1}
 \\
 &=\pi \sum_{k=0}^\infty \frac{(-1)^k(2r)^{2k+1}}{(k+1)!}, 
\end{align*}
using again \eqref{sin-m-int} and \eqref{Gamma-duplo}. The sum can be easily done and we are lead to 
\begin{align}
    E^{(4)}(R)= \frac{4}{\pi} \int_0^R r^2 \frac{\pi}{2r}(1-e^{-4r^2}) \,dr =   R^2-\frac14+\frac14 e^{-4R^2}. 
\end{align}


For $\beta=1$, the radial density in the large-$N$ limit follows from Lemma~\ref{Prop_GinOE expected complex}, with $Q(N-1,r^2)$ replaced by $1$. This gives $E_\C^{(1)}(R)$ from integration by parts: 
\begin{align*}
    E_{\C}^{(1)} (R)  &=  2 \int_0^{R}  r \Big( 1-e^{2r^2}(1 -\erf(\sqrt{2}r))  \Big) \,dr
\\
&=R^2-\frac{1}{2}e^{2R^2}+\frac{1}{2}+\frac{1}{2}e^{2R^2}\erf(\sqrt{2}R)-\int_0^R \frac{1}{2}e^{2r^2}\frac{2\sqrt{2}}{\sqrt{\pi}}e^{-2r^2}dr.
\end{align*}
On the other hand, $E_\R^{(1)}(R)$ follows from $R_{1,\R}^{(1)}(x)=1/\sqrt{2\pi}$ \cite{edelman1997probability}.
\end{proof}

It is also possible to obtain direct proofs for $\beta=1$ based on the limiting form of $R_{1,\C}(z)$ and the integral representation of the complementary error function. We omit these for brevity here.

\section{Number variance of the infinite Ginibre ensembles at the origin}\label{sec:proofvar}

\subsection{Number variance of the GinUE at the origin}
\hfill\\
The number variance of the GinUE is written as 
\begin{equation}
\Var \, \NN^{(2)}( D_R) = \int_{D_R} \int_{D_R^c} |K(z_1,z_2)|^2\,dA(z_1)\,dA(z_2) , 
\end{equation}
see e.g. \cite{akemann2022universality}.

\begin{proof}[Proof of Proposition~\ref{Prop_Var inf csGinibre} (i)]
Note that 
\begin{align}
\Var \, \NN^{(2)}( D_R) 
&= \int_{D_R} \,dA(z_1) \int_{D_R^c}  \,dA(z_2) e^{-|z_1|^2-|z_2|^2+z_1 \bar{z}_2+ \bar{z}_1 z_2 } .
\end{align}
It follows from the orthogonality with respect to the angular integration that 
\begin{align}
\begin{split}
\Var \, \NN^{(2)}( D_R)
&=  \int_{D_R} \,dA(z_1) \int_{D_R^c}  \,dA(z_2) e^{-|z_1|^2-|z_2|^2 } \sum_{j=0}^\infty \frac{  (z_1 \bar{z}_2+ \bar{z}_1 z_2)^j  }{ j! } 
\\
&=  \int_{D_R} \,dA(z_1) \int_{D_R^c}  \,dA(z_2) e^{-|z_1|^2-|z_2|^2 } \sum_{j=0}^\infty \frac{  |z_1|^{2j} |z_2|^{2j}  }{ (2j)! } \binom{2j}{j}  \label{Var cGinibre deep bulk v2}.
\end{split}
\end{align}
This gives 
\begin{align}
\Var \, \NN^{(2)}( D_R) & = \sum_{j=0}^\infty \int_{D_R} \frac{ |z_1|^{2j} e^{-|z_1|^2} }{ j! } \,dA(z_1) \int_{D_R^c}  \frac{ |z_2|^{2j} e^{-|z_2|^2} }{ j! } \,dA(z_2) =  \sum_{j=1}^\infty P(j,R^2) Q(j,R^2),
\end{align}
which is the first expression in \eqref{Var inf cGinibre v1}. 

On the other hand, note that by \eqref{Var cGinibre deep bulk v2}, upon inserting the series \eqref{I nu definition} we have 
\begin{align}
\Var \, \NN^{(2)}( D_R) 
&= \int_{D_R} \,dA(z_1) \int_{D_R^c}  \,dA(z_2)  e^{-|z_1|^2-|z_2|^2 }I_0(2|z_1z_2|)
\\
&= 
\int_{0}^R dr_1 f(R,r_1) ,
\label{Int-f}
\end{align}
where we have defined
\begin{equation}
f(R,r_1):=4 \int_R^\infty   e^{-r_1^2-r_2^2} r_1 r_2\,I_0(2r_1r_2)\,dr_2. 
\end{equation}
Using Leibniz' rule, we obtain
\begin{align}\label{diff-f}
\frac{d}{dR} \int_0^R f(R,r_1)\,dr_1 & = f(R,R) +   \int_0^R \frac{d}{dR}f(R,r_1)\,dr_1 
\\
&= 4\int_R^\infty   e^{-R^2-r_2^2}  Rr_2\,I_0(2Rr_2)\,dr_2  -4 \int_0^R e^{-r_1^2-R^2} r_1 R  I_0(2r_1R)\,dr_1 
\\
&= 4R e^{-R^2} \Big( \int_R^\infty-\int_0^R\Big) e^{-x^2} x I_0(2Rx)\,dx
= 4R e^{-R^2} g(R),
\label{2int-f}
\end{align}
defining 
\begin{equation}\label{gR def}
    g(R):=\Big( \int_R^\infty-\int_0^R\Big) e^{-x^2} x I_0(2Rx)\,dx.
\end{equation}
It is determined by the following two known integrals \cite[Eqs. (6.631.4), (6.631.8)]{gradshteyn2014table}, respectively
\begin{align}
\label{Grad Inu}
    \int_0^\infty xe^{-x^2}I_0(2xR)&=\frac{1}{2}e^{R^2}\\
    \int_0^R x e^{-x^2} I_0(2xR)&= \frac{1}{4}\left( e^{R^2}-e^{-R^2}I_0(2R^2)\right).  \label{Grad Inu 2}
\end{align}
This yields for \eqref{2int-f} 
\begin{equation}
\label{gR result}
   4R e^{-R^2} g(R)= 4R e^{-R^2} \Big( \int_0^\infty-2\int_0^R\Big) e^{-x^2} x I_0(2Rx)\,dx =2Re^{-2R^2}I_0(2R^2).
\end{equation}
On the one hand, the integral of \eqref{diff-f} over $[0,R]$ trivially gives \eqref{Int-f}. On the other hand, using the previous result \eqref{gR result}, we obtain from the same integral 
\begin{equation}
  \Var \, \NN^{(2)}( D_R) =  \int_0^R 2xe^{-2x^2}I_0(2x^2)dx = R^2e^{-2R^2}(I_0(2R^2)+I_1(2R^2)),
    \end{equation}
thanks to the anti-derivative given in Remark \ref{rem:Iprime}. 
The asymptotic behaviours \eqref{linear behaviour Cplx} and \eqref{linear behaviour Cplx 0} follow from the well-known asymptotic behaviours of the modified Bessel-function 
\begin{equation} \label{I asymptotic}
I_\nu(x) \sim  \begin{cases}
    \dfrac{ e^x }{ \sqrt{2\pi x} }, & x \to \infty,
    \smallskip 
    \\
    \dfrac{x^\nu}{ 2^\nu \nu! }, & x \to 0. 
\end{cases} 
\end{equation}
This completes the proof. 
\end{proof}
\begin{rem}\label{rem:gR prep}
In order to prepare for the proof in the next subsection, where we will not have all integrals as in \eqref{Grad Inu 2} at hand, we present a different determination of $g(R)$ starting from \eqref{gR def}. It contains a nice expression closely related to the generating function for 
$Q(k+1,R^2)$. 

Inserting the series \eqref{I nu definition} for the Bessel function into the definition \eqref{gR def} we obtain
\begin{align*}
g(R)&
= \frac12\sum_{k=0}^\infty \frac{R^{2k}}{k!} \Big( Q(k+1,R^2)-P(k+1,R^2) \Big)
= \sum_{k=0}^\infty \frac{R^{2k}}{k!} Q(k+1,R^2)- \frac{ e^{R^2} }{2}. 
\end{align*}
Furthermore, using the expression
\begin{equation}
Q(k+1,x) = e^{-x} \sum_{l=0}^k \frac{x^l}{l!},
\end{equation}
we have
\begin{align}\label{Qgen}
\sum_{k=0}^\infty \frac{R^{2k}}{k!} Q(k+1,R^2) &= e^{-R^2} \sum_{k=0}^\infty \frac{R^{2k}}{k!} \sum_{l=0}^k \frac{ R^{2l} }{l!} 
=  e^{-R^2 } \sum_{ m=0 }^\infty \frac{1}{m!}   
\sum_{k = \lceil \frac{m}{2}\rceil  }^{m} \binom{m}{k}        R^{2m}. 
\end{align}
Here, we have introduced the variable $m=k+l$, with $k\ge l=m-k\ge0$ and $m\ge k$, resulting into the summation boundaries indicated, and $\lceil \frac{m}{2}\rceil$ is the ceiling function.  The following two identities follow from the Binomial Theorem, by relabelling one of the sums in the first identity, and using it together with Pascal's identity to derive the second: 
\begin{equation}
 2\sum_{k=m}^{2m} \binom{2m}{k} = 2^{2m} + \binom{2m}{m}, \qquad  
 \sum_{k=m+1}^{2m+1} \binom{2m+1}{k} = 2^{2m} .
\end{equation}
Splitting the sum in \eqref{Qgen} into even and odd summands, we obtain
\begin{align*}
\sum_{k=0}^\infty \frac{R^{2k}}{k!} Q(k+1,R^2)&= \sum_{ m=0 }^\infty \frac{1}{(2m)!}  \Big(  2^{2m-1} +\frac{1}{2}\binom{2m}{m}\Big)    (R^2)^{2m} + \sum_{ m=0 }^\infty \frac{1}{(2m+1)!}   2^{2m}    (R^{2})^{2m+1} 
\\
&= \frac12 \sum_{m=0}^\infty \frac{(2R^2)^m}{m!} +\frac12 \sum_{m=0}^\infty \frac{ (R^2)^{2m} }{ (m!)^2 } = \frac12 e^{2R^2} + \frac12 I_0(2R^2). 
\end{align*}
This leads again to \eqref{gR result}.
\end{rem}

\subsection{Number variance of the GinSE at the origin} \label{Subsec_variance GinSE inf}
\hfill\\
For the GinSE, we have 
\begin{equation} 
\Var \NN^{(4)}(D_R)= -\int_{ D_R }\,dA(z) \int_{ D_R^c }\,dA(w)\, R_2^{(4),c}(z,w),
\end{equation}
where $R_2^{(4),c}(z,w)=R_2^{(4)}(z,w)-R_1^{(4)}(z)R_1^{(4)}(w)$  is the connected $2$-point function.

\begin{proof}[Proof of Proposition~\ref{Prop_Var inf csGinibre} (ii)]
For completeness we present $R_2^{(4),c}(z,w)$ in terms of the limiting kernel at the origin \cite{MR1928853,Mehta}
\begin{equation}
    \kappa(z,w)=e^{z^2+w^2}\erf (z-w).
\end{equation}
Inserted this into the limit of \eqref{Rk beta4N} defined in \eqref{lim Rkb24}, we obtain
\begin{equation}
R_1^{(4)}(z)=(\bar{z}-z)e^{ -2|z|^2 } \kappa(z,\bar{z})
\end{equation}
and
\begin{equation}
	R_2^{(4)}(z,w)=(\bar{z}-z)(\bar{w}-w)  e^{ -2|z|^2-2|w|^2 } 
\Big( \kappa(z,\bar{z})  \kappa(w,\bar{w})-|\kappa(z,w)|^2+|\kappa(z,\bar{w})|^2  \Big).
\end{equation}
This leads to 
\begin{align}
\begin{split}
R_2^{(4),c}(z,w)&=-(\bar{z}-z)(\bar{w}-w)  e^{ -2|z|^2-2|w|^2 } 
\Big( |\kappa(z,w)|^2-|\kappa(z,\bar{w})|^2  \Big)
\\
&= -\pi (\bar{z}-z)(\bar{w}-w)  e^{ -2|z|^2-2|w|^2 } \Big( \Big|   e^{z^2+w^2} \erf(z-w) \Big|^2- \Big|  e^{z^2+\bar{w}^2} \erf(z- \Bar{w} ) \Big|^2 \Big).  
\end{split}
\end{align}
In the following, we will use a different expression, directly following from 
the skew-orthogonal polynomial formalism \cite{MR1928853,Mehta}, where the angular integrals are more easily performed:
\begin{equation}
\label{kappa origin}
\kappa(z,w)= \sum_{ 0 \le l \le k \le \infty } \frac{ 2^{k+l+1} }{ (2k+1)!!\, (2l)!! } ( z^{2k+1} w^{2l} - w^{2k+1 } z^{2l} )=: G(z,w)-G(w,z),
\end{equation}
see e.g. \cite{akemann2021scaling}. 
Following the proof in \cite[Section 2.2]{akemann2022universality}, using the complex conjugation symmetry and the orthogonality, we have
\begin{align*}
\Var \NN^{(4)}(D_R)  &= 2\int_{ D }\,dA(z) \int_{ D^c }\,dA(w)\,(\bar{z}-z)(\bar{w}-w)  e^{ -2|z|^2-2|w|^2 } |\kappa(z,w)|^2
\\
&= -4\int_{ D }\,dA(z) \int_{ D^c }\,dA(w)\,(zw+\bar{z}\bar{w}-z \bar{w}-\bar{z}w)  e^{ -2|z|^2-2|w|^2 } G(z,w) G( \bar{w} , \bar{z} ).
\\
&= 4\int_{ D }\,dA(z) \int_{ D^c }\,dA(w)\, \bar{z}w  e^{ -2|z|^2-2|w|^2 } G(z,w) G( \bar{w} , \bar{z} ),
\end{align*}
where only the third term in the second line can be shown to contribute \cite{akemann2022universality}.
The angular integrals project onto a single sum,
\begin{align}\label{Var4R int}
\Var \NN^{(4)}(D_R)  &=  \sum_{k=0}^\infty  \frac{2^{4k+4}}{(2k+1)!^2} \int_{ D }\,dA(z) \int_{ D^c }\,dA(w)\,  e^{ -2|z|^2-2|w|^2 } |z|^{4k+2} |w|^{4k+2}
\\
&= \sum_{k=0}^\infty P(2k+2,2R^2) Q(2k+2,2R^2), 
\end{align}
which leads to \eqref{Var inf sGinibre v1}.

It remains to show \eqref{Var inf sGinibre v2}. 
For this, note that the series expansion for the Bessel function of first kind, 
\begin{equation}
\label{J0 def}
    J_0(z)=\sum_{k=0}^\infty\frac{(-1)^k(z/2)^{2k}}{(k!)^2},
\end{equation}
together with that for the modified Bessel function $I_0(z)$ in \eqref{I nu definition}, can be used to rewrite the sum in \eqref{Var4R int} as the following difference:
\begin{align}
\begin{split}
\Var \NN^{(4)}(D_R) 
&= 2\int_{ D }\,dA(z) \int_{ D^c }\,dA(w)\,  e^{ -2|z|^2-2|w|^2 } \Big( I_0(4|zw|)- J_0(4|zw|) \Big) \\
&=  8 \int_0^R \,dr_1 \int_R^\infty \,dr_2\,  e^{ -2r_1^2-2r_2^2 } r_1r_2\Big( I_0(4r_1r_2)- J_0(4r_1r_2) \Big).
\end{split}
\end{align}
The first integral follows from Proposition~\ref{Prop_Var inf csGinibre} (i), by substituting $R^2\to2R^2$:
\begin{align*}
  8 \int_0^R \,dr_1 \int_R^\infty \,dr_2\,  e^{ -2r_1^2-2r_2^2 } r_1r_2 I_0(4r_1r_2)
= R^2 e^{ -4R^2 } \Big ( I_0(4 R^2) + I_1(4 R^2) \Big ).
\end{align*}

Next, we evaluate the second integral 
\begin{equation}
 8 \int_0^R \,dr_1 \int_R^\infty \,dr_2\,  e^{ -2r_1^2-2r_2^2 } r_1r_2  J_0(4r_1r_2)=
 \int_0^R \,dr_1 f(R,r_1), 
\end{equation}
where we defined
\begin{equation}
f(R,r_1):=8 \int_R^\infty   e^{-2r_1^2-2r_2^2} r_1 r_2\,J_0(4r_1r_2)\,dr_2. 
\end{equation}
As in the previous subsection we obtain
\begin{align*}
\frac{d}{dR} \int_0^R f(R,r_1)\,dr_1 & = f(R,R) +   \int_0^R \frac{d}{dR}f(R,r_1)\,dr_1 
\\
&= 8R \int_R^\infty   e^{-2R^2-2r_2^2}  r_2\,J_0(4Rr_2)\,dr_2  -8R \int_0^R e^{-2r_1^2-2R^2} r_1  J_0(4Rr_1)\,dr_1 
\\
&= 8R e^{-2R^2} g(R)
\end{align*}
with  
\begin{equation}
g(R):=  \Big( \int_R^\infty-\int_0^R\Big) e^{-2x^2} x J_0(4Rx)\,dx,
\end{equation}
now containing the function $J_0$ instead of $I_0$. Although the first integral in \eqref{Grad Inu} can be analytically continued, we did not find the second integral with $J_0$ in the literature. Thus, we follow the same strategy as in Remark \ref{rem:gR prep}, with \eqref{J0 def} instead:
\begin{align}
\begin{split}
g(R) &=  \Big( \int_R^\infty-\int_0^R\Big) e^{-2x^2} x  \sum_{k=0}^\infty \frac{ (-1)^k (2Rx)^{2k} }{ k!^2 } \,dx
 \\
&=\frac{1}{4}\sum_{k=0}^\infty \frac{(-1)^k(2R^2)^k}{k!} \Big( Q(k+1,2R^2)- P(k+1,2R^2) \Big)
\\
&= \frac{1}{2}\sum_{k=0}^\infty    \frac{(-1)^k(2R^2)^k}{k!} Q(k+1,2R^2)-\frac14 e^{-2R^2} .
\label{gR beta4}
\end{split}
\end{align}
Similarly to \eqref{Qgen} we can thus write 
\begin{align*}
\frac{1}{2}\sum_{k=0}^\infty \frac{(-1)^k(2R^2)^k}{k!} Q(k+1,2R^2) 
 &=\frac{1}{2} e^{-2R^2} +  \frac{1}{2} e^{-2R^2} \sum_{m=1}^\infty  \frac{1}{m!}
 \sum_{k = \lceil \frac{m}{2}\rceil  }^{m}
 (-1)^k \binom{m}{k}   (2R^2)^{m}. 
\end{align*} 
Similarly to the Remark \ref{rem:gR prep} we can use two combinatorial identities, that follow from the Binomial theorem, this time with alternating coefficients. For that reason the terms $2^m$ are absent, and we obtain
\begin{align}
& 2\sum_{ k=m }^{2m} (-1)^k \binom{2m}{k}  = (-1)^m \binom{2m}{m}, \qquad (m \ge 1),
\\
&\sum_{ k=m+1 }^{2m+1} (-1)^k \binom{2m+1}{k}  = - (-1)^m \binom{2m}{m}, \qquad (m \ge 0). 
\end{align}
Using this, we have
\begin{align}
\begin{split}
    &\quad \frac{1}{2}\sum_{k=0}^\infty \frac{(-1)^k(2R^2)^k}{k!} Q(k+1,2R^2) 
    \\
    &= \frac{e^{-2R^2}}{2}\bigg(1+ \sum_{ m=1 }^\infty \frac{1}{(2m)!}  (-1)^m\frac{1}{2}\binom{2m}{m} (2R^2)^{2m} + \sum_{ m=0 }^\infty \frac{1}{(2m+1)!}   (-1)^{m+1}\binom{2m}{m} (2R^{2})^{2m+1} \bigg)
    \\
    &= \frac{e^{-2R^2}}{4} \Big(1+ J_0(4R^2) \Big) - e^{-2R^2} \,R^2\,{}_1 F_2  \Big( \frac{1}{2}; 1,\frac{3}{2}; -4 R^4 \Big) .
  \label{sum Q 1F2}  
\end{split}
\end{align}
The last step follows from \eqref{J0 def} and from spelling out the hypergeometric function defined in \eqref{def hypergeometric}, in this case as
\begin{equation} \label{J0 1F2 expansion}
{}_1 F_2  \Big( \frac{1}{2}; 1,\frac{3}{2}; -4 x^4 \Big) = \sum_{k=0}^\infty  \frac{ (-1)^k (2x^2)^{2k} }{ (k!)^2(2k+1) }.
\end{equation}
Inserted into \eqref{gR beta4} we thus obtain 
\begin{equation}
g(R)= \frac14 e^{-2R^2} J_0(4R^2) - e^{-2R^2} \,R^2\,{}_1 F_2  \Big( \frac{1}{2}; 1,\frac{3}{2}; -4 R^4 \Big). 
\end{equation}
Therefore, we have shown that 
\begin{align}
\begin{split}
\Var \NN^{(4)}(D_R) &=    R^2 e^{ -4R^2 } \Big ( I_0(4 R^2) + I_1(4 R^2) \Big ) 
 \\
 &\quad -2  \int_0^R  x e^{-4x^2} \left( J_0(4x^2) -4  x^2\,{}_1 F_2  \Big( \frac{1}{2}; 1,\frac{3}{2}; -4 x^4 \Big) \right)\,dx.  \label{int F J0}
\end{split}
\end{align}
It follows from 
\begin{equation}\label{diff 1F2}
    \frac{\partial}{\partial x} \Big(x^2 {}_1 F_2  \Big( \frac{1}{2}; 1,\frac{3}{2}; -4 x^4 \Big)\Big) = 2x J_0(4x^2)
\end{equation}
and $ \frac{\partial}{\partial x}(e^{-4x^2})=-8xe^{-4x^2}$, that the integrand in the integral \eqref{int F J0} is a total derivative. This leads to the final answer 
\begin{align*}
\begin{split}
\Var \NN^{(4)}(D_R) &=    R^2 e^{ -4R^2 } \Big ( I_0(4 R^2) + I_1(4 R^2) \Big )
 -R^2e^{-4R^2}\,{}_1 F_2  \Big( \frac{1}{2}; 1,\frac{3}{2}; -4 R^4 \Big). 
\end{split}
\end{align*}

The asymptotic behaviours \eqref{linear behaviour Sym} and \eqref{linear behaviour Sym 0} follow from \eqref{I asymptotic} and that of the hypergeometric function \cite[Section 16.11(ii)]{olver2010nist}
and
\begin{equation}
    {}_1 F_2  \Big( \frac{1}{2}; 1,\frac{3}{2}; -x \Big) =\frac{1}{2\sqrt{x}}+O(x^{-\frac{3}{2}}), \qquad x \to \infty,
\end{equation}
and 
\begin{equation}
      {}_1 F_2  \Big( \frac{1}{2}; 1,\frac{3}{2}; -x \Big) = 1 +O(x)
      , \qquad x \to 0.  
\end{equation}
\end{proof}

To derive \eqref{Var inf sGinibre v3}, observe from \eqref{diff 1F2} that 
\begin{align} 
\begin{split}
R^2 {}_1 F_2  ( 1/2; 1,3/2; -4 R^4 ) &= \int_0^R 2t J_0(4t^2)\,dt  
=\frac{1}{4}\int_0^{4R^2}  J_0(s)\,ds
\\
&= R^2 J_0(4R^2)+\frac{\pi}{2} R^2 \Big( J_1(4R^2) \textbf{H}_{0}(4R^2)-J_{0}(4R^2) \textbf{H}_1(4R^2) \Big),
\end{split}
\end{align}
where we used \cite[6.511.6]{gradshteyn2014table}.
Then the expression \eqref{Var inf sGinibre v3} follows from the recurrence relation \cite[Eq.~(11.4.23)]{olver2010nist}
\begin{equation}
\textbf{H}_{\nu-1}(z)+\textbf{H}_{\nu+1}(z) = \frac{2\nu}{z} \textbf{H}_\nu(z) + \frac{(z/2)^\nu}{ \sqrt{\pi} \, \Gamma(\nu+\frac32) },
\end{equation}
with $\nu=0$. 

\begin{rem}
Combining Propositions~\ref{Prop_Var inf csGinibre} (i) and (ii), we also obtain
\begin{align}
\sum_{k=0}^\infty P(2k+1,2R^2) Q(2k+1,2R^2) =    R^2 e^{ -4R^2 } \Big ( I_0(4 R^2) + I_1(4 R^2) + {}_1 F_2  ( 1/2; 1,3/2; -4 R^4 ) \Big).
\end{align}
The left-hand side of this identity appears in the number variance of the induced GinSE, see \cite[Proposition 1.4]{akemann2022universality}. 
\end{rem}

\subsection{Number variance of the GinOE at the origin}
\hfill\\
In this subsection, we show Theorem~\ref{Thm_GinOE inf}.

\begin{prop}[\textbf{Expression of the number variance of the infinite GinOE}] \label{Prop_Expression var inf}
For any radius $R \in [0, \infty)$ we have the following.
\begin{itemize}
    \item[\textup{(i)}] 
We have for the real eigenvalues
\begin{equation} \label{Var NR DR inf closed}
\Var \NN_\R^{(1)}(D_R)= 2 \sqrt{ \frac{2}{\pi} } R - \frac{2}{\sqrt{\pi}} R \erf(2R) -\frac12 \erf( \sqrt{2}R ) + \frac14 \erf( \sqrt{2}R )^2 + \frac{1}{\pi} \Big( 1 -e^{-4R^2} \Big) .
\end{equation}
\item[\textup{(ii)}] We have for the mixing between real and complex eigenvalues 
\begin{align}\label{Cov prop4.3}
&\quad \Cov \Big( \NN_\C^{(1)}(D_R),   \NN_\R^{(1)}(D_R) \Big)
\\
&= \frac{2}{\pi} \int_{0}^R  \erfc(\sqrt{2}v)  v \,e^{v^2}\Big( e^{-(\hat{a} - R)^2}-e^{-(\hat{a} + R)^2}+ \sqrt{\pi}(\hat{a} - R) \erf(\hat{a} - R) -\sqrt{\pi} (\hat{a} + R) \erf(\hat{a} + R) \Big)\,dv, 
\nonumber
\end{align}
where $\hat{a} \equiv \hat{a}(R)=\sqrt{R^2-v^2}.$
\item[\textup{(iii)}] We have for the complex eigenvalues
\begin{equation}
\Var \NN_\C^{(1)} (D_R)  = 2 E_\C^{(1)}(R)  + 4 \Big( I_-(R) - I_+(R) \Big)   
\end{equation}
where 
\begin{equation}
I_\pm(R):= I_{\pm,1}(R) + I_{\pm,2}(R) + I_{\pm,3}(R) + I_{\pm,4}(R). 
\end{equation}
Here, 
\begin{align*}
I_{\pm,1}(R) & = \frac{1}{2\pi} \int_0^R \int_0^R \erfc(\sqrt{2}y_1) \erfc(\sqrt{2}y_2) e^{(y_1 \pm y_2)^2} e^{-(a+b)^2} (1+(y_1 \pm y_2)^2) \,dy_2 \,dy_1, 
\\
I_{\pm,2}(R) & = -\frac{1}{2\pi} \int_0^R \int_0^R \erfc(\sqrt{2}y_1) \erfc(\sqrt{2}y_2) e^{(y_1 \pm y_2)^2} e^{-(a-b)^2} (1+(y_1 \pm y_2)^2) \,dy_2 \,dy_1, 
\\
I_{\pm,3}(R) & = \frac{1}{4\sqrt{\pi}} \int_0^R \int_0^R \erfc(\sqrt{2}y_1) \erfc(\sqrt{2}y_2) e^{(y_1 \pm y_2)^2} (a+b)\erf(a+b) (1+2(y_1 \pm y_2)^2) \,dy_2 \,dy_1, 
\\
I_{\pm,4}(R) & = -\frac{1}{4\sqrt{\pi}} \int_0^R \int_0^R \erfc(\sqrt{2}y_1) \erfc(\sqrt{2}y_2) e^{(y_1 \pm y_2)^2} (a-b)\erf(a-b) (1+2(y_1\pm y_2)^2) \,dy_2 \,dy_1, 
\end{align*}
where $a \equiv a(R)= \sqrt{R^2-y_1^2}$ and $b \equiv b(R)= \sqrt{R^2-y_2^2}$. 
\end{itemize}
\end{prop}

\begin{rem}[\textbf{Number variance of the finite GinOE}] \label{Rem_GinOE var finiteN}
The counterpart of Proposition~\ref{Prop_Expression var inf} for the finite GinOE can be computed analogously. This involves leveraging the $1$- and $2$-point correlation functions of the real and complex eigenvalues, along with the mixed real and complex $(1,1)$-point function available in the references \cite{forrester2007eigenvalue, sommers2007symplectic, MR2530159}.
The resulting formula requires straightforward yet lengthy computations. 
Nonetheless, we do not include it here because it exceeds the scope of this paper.
\end{rem}

\begin{proof}[Proof of Proposition~\ref{Prop_Expression var inf}]
We first show (i). Note that
\begin{equation}
\Var \NN_\R^{(1)}(D_R)= \int_{-R}^R R_{1,\R}^{(1)}(x)\,dx 
+\int_{(-R,R)^2} \Big(R_{2,\R}^{(1)}(x,y)-R_{1,\R}^{(1)}(x) R_{1,\R}^{(1)}(y)\Big)\, dx\,dy, 
\end{equation}
see e.g. \cite{forrester2007eigenvalue} or \cite[Chapter 16]{Mehta}. Since 
\begin{equation}
R_{2,\R}^{(1)}(x,y)=S_\R(x,x)S_\R(y,y)-D_\R(x,y)\tilde{I}_\R(x,y)-S_\R(x,y)S_\R(y,x),
\end{equation}
as given in Appendix \ref{App:A}, we have 
\begin{equation}
\begin{split}
	\Var \NN_\R^{(1)}(D_R)&=E_\R^{(1)}(R)-\int_{(-R,R)^2}\Big( S_\R(x,y)^2 + \tilde{I}_\R(x,y) D_\R(x,y)\Big)\, dx\,dy.
\end{split}
\end{equation}
With the limiting matrix elements given in \eqref{SR def}, \eqref{DR def} and \eqref{IR def}, note that 
\begin{align}
\begin{split}
-\int_{(-R,R)^2} \sgn(x-y) D_\R(x,y) \,dx\,dy &= \int_{(-R,R)^2} \sgn(x-y) \frac{\pa}{\pa y} S_\R(x,y) \,dx\,dy
\\
&= \int_{-R}^R \Big[ \int_{-R}^R \frac{\pa}{\pa y} S_\R(x,y) \,dy -2 \int_{x}^R \frac{\pa}{\pa y} S_\R(x,y) \,dy \Big] \,dx 
\\
&= \int_{-R}^R 2S_\R(x,x)-S_\R(x,R)-S_\R(x,-R) \,dx. 
\end{split}
\end{align}
Furthermore, since 
\begin{align*}
&\quad \int_{-R}^R \int_x^y S_\R(t,x)\,dt \frac{\pa}{\pa y} S_\R(x,y)\,dy
\\
&= S_\R(x,y) \int_x^y S_\R(t,x)\,dt \Big|_{-R}^R- \int_{-R}^R S_\R(x,y)^2 \,dy
\\
&= S_\R(x,R) \int_x^R S_\R(t,x)\,dt- S_\R(x,-R) \int_{x}^{-R} S_\R(t,x)\,dt - \int_{-R}^R S_\R(x,y)^2 \,dy, 
\end{align*}
we have 
\begin{align}
\begin{split}
\label{Var prop4.3i}
\Var \NN_\R^{(1)}(D_R) &= 2 E_\R^{(1)}(R) -2 \int_{(-R,R)^2} S(x,y)^2\,dx\,dy
-\frac12 \int_{-R}^R \,dx\, \Big( S(x,R)+S(x,-R)\Big) 
\\
& \quad + \int_{-R}^R\,dx \Big[ S(x,R) \int_x^R S(t,x)\,dt + S(x,-R) \int_{-R}^{x} S(t,x)\,dt \Big].
\end{split}
\end{align}
Since 
\begin{align*}
\int_{ (-R,R)^2 } S(x,y)^2 \,dx\,dy &= \frac{1}{2\pi}  \int_{ (-R,R)^2 } e^{-(x-y)^2}\,dx\,dy = \frac{1}{2\pi} \Big( 2\sqrt{\pi} R \erf(2R) -1 +e^{-4R^2} \Big),
\end{align*}
we have together with $E_\R^{(1)}(R)$ from Proposition~\ref{Prop_EN inf}
\begin{equation}
\begin{split}
 2 E_\R^{(1)}(R) -2 \int_{(-R,R)^2} S(x,y)^2\,dx\,dy 
 &= 2 \sqrt{ \frac{2}{\pi} } R - \frac{2}{\sqrt{\pi}} R \erf(2R) + \frac{1}{\pi} \Big( 1 -e^{-4R^2} \Big).
\end{split}
\end{equation}
Similarly, by direct computation, we have
\begin{equation}
 -\frac12 \int_{-R}^R \,dx\, \Big( S(x,R)+S(x,-R)\Big) = -\frac12 \erf( \sqrt{2}R )
\end{equation}
and
\begin{equation}
 \int_{-R}^R\,dx \Big[ S_\R(x,R) \int_x^R S_\R(t,x)\,dt + S_\R(x,-R) \int_{-R}^{x} S(t,x)\,dt \Big]=  \frac14 \erf( \sqrt{2}R )^2.  
\end{equation}
By inserting the last three equations into \eqref{Var prop4.3i}, we obtain Proposition~\ref{Prop_Expression var inf} (i).
\bigskip 

Next, we show (ii). By definition, we have 
\begin{align}
\begin{split}
&\quad \Cov \Big( \NN_\C^{(1)}(D_R),   \NN_\R^{(1)}(D_R) \Big)  = 2 \, \Cov \Big( \NN_\C^{(1)}(H_R),   \NN_\R^{(1)}(D_R) \Big) 
\\
&= 2\int_{H_R} \frac{du\,dv}{\pi} \int_{-R}^R \,dx \,   \Big( R_{1,\R,1,\C}^{(1)}(x,u+iv)- R_{1,\R}^{(1)}(x) R_{1,\C}^{(1)}(u+iv) \Big) , 
\end{split}
\end{align}
where $H_R = \C_+ \cap D_R$ is the semi-disc of radius $R$ in the upper-half plane. 
On the one hand, it follows from \eqref{R k1k2 rc} that  
\begin{equation}
R_{1,\R,1,\C}^{(1)}(x,u+iv)- R_{1,\R}^{(1)}(x) R_{1,\C}^{(1)}(u+iv) =  -\erfc(\sqrt{2}v) v\,  e^{ -(x-u)^2+v^2 } . 
\end{equation}

Note that 
\begin{align*}
&\quad \int_{-\hat{a}}^{\hat{a}} \,du \int_{-R}^R \,dx\,  e^{ -(x-u)^2 } 
\\
&= \sqrt{\pi}\Big(\frac{ e^{-(\hat{a}+ R)^2}-e^{-( \hat{a}-R)^2} }{\sqrt{\pi}}+ (\hat{a} + R) \erf(\hat{a} + R) - (\hat{a} - R) \erf(\hat{a} - R) \Big),
\end{align*}
where we use $\hat{a} \equiv \hat{a}(R)=\sqrt{R^2-v^2}$ for the bounds of the $u$ integral.
Therefore we obtain
\begin{align*}
&\quad \Cov \Big( \NN_\C^{(1)}(H_R),   \NN_\R^{(1)}(D_R) \Big) 
\\
&= -\frac{1}{\pi} \int_{0}^R  \erfc(\sqrt{2}v)  v \,e^{v^2}\Big( e^{-(\hat{a} + R)^2}-e^{-(\hat{a} - R)^2}+ \sqrt{\pi}(\hat{a} + R) \erf(\hat{a} + R) -\sqrt{\pi} (\hat{a} - R) \erf(\hat{a} - R) \Big)\,dv.
\end{align*}

\bigskip 

Finally, we show (iii). Note that 
\begin{align}
\Var \NN_\C^{(1)} (D_R) & = \int_{ D_R } R_{1,\C}^{(1)}(z)\,dA(z)  +\int_{D_R^2} \Big( R_{2,\C}^{(1)} (z,w)-R_{1,\C}^{(1)}(z) R_{1,\C}^{(1)}(w) \Big) \,dA(z)\,dA(w)
\\
&= E_\C^{(1)}(R)  +\int_{D_R^2} \Big( R_{2,\C}^{(1)} (z,w)-R_{1,\C}^{(1)}(z) R_{1,\C}^{(1)}(w) \Big) \,dA(z)\,dA(w). 
\end{align}
Since 
\begin{equation}
 \NN_\C^{(1)} (D_R) = 2  \NN_\C^{(1)} (H_R),
\end{equation}
we have 
\begin{equation} \label{Var multi 4}
E_\C^{(1)}(R) = 2 
\int_{ H_R } R_{1,\C}^{(1)}(z)\,dA(z)
, \quad \Var \NN_\C^{(1)} (D_R) = 4 \Var \NN_\C^{(1)} (H_R).
\end{equation}
Here, 
\begin{align} \label{Var N HR decom}
\Var \NN_\C^{(1)} (H_R) 
&= 
\int_{ H_R } R_{1,\C}^{(1)}(z)\,dA(z)+ \int_{H_R^2 } \Big( R_{2,\C}^{(1)} (z,w)-R_{1,\C}^{(1)}(z) R_{1,\C}^{(1)}(w) \Big) \,dA(z)\,dA(w). 
\end{align}

By \eqref{R k GinOE c}, we have
\begin{align}
\begin{split}
&\quad R_{2,\C} (z_1,z_2)-R_{1,\C}(z_1) R_{1,\C}(z_2) 
\\
&= \frac{\pi}{2}  \erfc(\sqrt{2}y_1)  \erfc(\sqrt{2}y_2)  e^{-(x_1-x_2)^2}
\\
&\quad \times \Big[   - \Big( (x_1-x_2)^2 +(y_1+y_2)^2 \Big)  e^{ (y_1+y_2)^2  }+ \Big( (x_1-x_2)^2 +(y_1-y_2)^2 \Big)  e^{ (y_1-y_2)^2  }  \Big], 
\end{split}
\end{align}
where $z_k=x_k+iy_k$ ($k=1,2$) with $y_k>0.$
Using the Gaussian integrals
\begin{align}
\begin{split}
&\quad
\int_{-a}^{a} \int_{-b}^{b} e^{-(x_1-x_2)^2} \Big( (x_1-x_2)^2 + c \Big) \,dx_2 \,dx_1
\\
&= \Big( e^{-(a+b)^2} - e^{-(a-b)^2} \Big) (1+c) + \frac{\sqrt{\pi}}{2} \Big( (a+b)\erf(a+b) - (a-b)\erf(a-b) \Big) (1+2c), 
\end{split}
\end{align}
we obtain 
\begin{align}
\int_{H_R ^2 } \Big( R_{2,\C}^{(1)} (z,w)-R_{1,\C}^{(1)}(z) R_{1,\C}^{(1)}(w) \Big) \,dA(z)\,dA(w)
= I_-(R) - I_+(R).
\end{align}
Now the proposition follows.
\end{proof}

To prove Theorem~\ref{Thm_GinOE inf} (iii), we first show the following lemma.

\begin{lem}\label{Lem_integrals GinOE c}
As $R \to \infty$, we have
\begin{align}  \label{I -12 asymp}
I_{-,1}(R) & = O(R^4 e^{-2R^2}), &
I_{-,2}(R) & = O(1), 
\\  \label{I -34 asymp}
I_{-,3}(R) & = \frac{\pi (2 \sqrt{2} - 1) - 4}{4 \pi^{3/2}} R + O\Big( \frac{1}{R} \Big), &
I_{-,4}(R) & = O\Big( \frac{1}{R^2} \Big), 
\end{align}
and
\begin{align}  \label{I +12 asymp}
I_{+,1}(R) & = O\Big( \frac{1}{R^2} \Big), &
I_{+,2}(R) & = -\frac{R}{4 \sqrt{\pi}} + o(R), 
\\     \label{I +34 asymp}
I_{+,3}(R) & = \frac{1}{2} R^2 - \frac{4 + \pi (2 \sqrt{2} - 1)}{4 \pi^{3/2}} R + o(R), 
&
I_{+,4}(R) & = -\frac{R}{4 \sqrt{\pi}} + o(R). 
\end{align}
\end{lem}
A comparison with numerics shows that the $o(R)$ error terms for $I_{+,2}(R)$, $I_{+,3}(R)$ and $I_{+,4}(R)$ are in fact of order $O(1)$ for $I_{+,2}(R)$, and $O(\sqrt{R})$ for $I_{+,3}(R)$ and $I_{+,4}(R)$.
While we do not need a more precise asymptotic analysis for Theorem~\ref{Thm_GinOE inf}, we note that these subleading behaviours can be derived analytically in a similar way as our analysis of $I_{+,3}(R)$.

\begin{proof}
To analyse $I_{-,1}(R)$ and $I_{-,2}(R)$, we first note that 
$$
(y_1 - y_2)^2 - (a \pm b)^2 = 2 (y_1^2 + y_2^2) - 2 y_1 y_2 - 2 R^2 \Big( 1 \pm \sqrt{1 - \frac{y_1^2 + y_2^2}{R^2} + \frac{y_1^2 y_2^2}{R^4}} \Big)
$$
and that for any $y_1,y_2 \in [0,R],$ 
$$
0 \leq \sqrt{1 - \frac{y_1^2 + y_2^2}{R^2} + \frac{y_1^2 y_2^2}{R^4}} \leq 1 - \frac{y_1^2 + y_2^2}{2 R^2}. 
$$
Combining these with elementary estimates
$$
\erfc(\sqrt{2} y) \leq e^{-2 y^2}, \qquad 1 + (y_1 - y_2)^2 \leq 1 + R^2, \qquad e^{-2 y_1 y_2} \leq 1, 
$$
we obtain
\begin{align}
|I_{-,1}(R)| \leq \frac{1}{2 \pi} \int_0^R \int_0^R e^{-2 (y_1^2 + y_2^2)} \, e^{2 (y_1^2 + y_2^2) - 2 y_1 y_2 - 2 R^2} (1 + R^2) \,dy_1 \,dy_2
\leq \frac{1}{2 \pi} R^4 e^{-2 R^2}.
\end{align}
The integrand of $I_{-,2}(R)$ is bounded by the integrable function $e^{-(y_1+y_2)^2} (1 + (y_1 - y_2)^2)$ and thus
\begin{align}
\abs{I_{-,2}(R)} \leq -\frac{1}{2\pi} \int_0^\infty \int_0^\infty e^{-(y_1+y_2)^2} (1 + (y_1 - y_2)^2) \,dy_2 \,dy_1
= \frac{1}{3 \pi}.
\end{align}
These give \eqref{I -12 asymp}. 

For $I_{-,3}(R)$ and $I_{-,4}(R)$, we observe that for $y_1,y_2 \geq 0,$
$$
\erfc(\sqrt{2}y_1) \erfc(\sqrt{2}y_2) e^{(y_1-y_2)^2} \leq e^{-(y_1^2+y_2^2)}.
$$ 
After expanding the other factors in the integrand using a Taylor series, it can be observed that the error of the approximation is polynomial in $y_1$ and $y_2$. As a result, the integral over this error will remain bounded as $R \to \infty.$
Note that 
$$
a + b = \sqrt{R^2-y_1^2} + \sqrt{R^2-y_2^2} = 2 R + O\Big(\frac{1}{R}\Big), \qquad \erf(a+b) = 1 + O(e^{-4R^2}). 
$$
Using these, we have 
\begin{align}
I_{-,3}(R)
& = \frac{1}{4\sqrt{\pi}} \int_0^\infty \int_0^\infty \erfc(\sqrt{2}y_1) \erfc(\sqrt{2}y_2) e^{(y_1 - y_2)^2} \Big( 2 R + O\Big(\frac{1}{R}\Big) \Big) (1+2(y_1 - y_2)^2) \,dy_2 \,dy_1
\nonumber\\
& = \frac{R}{2\sqrt{\pi}} \frac{\pi (2 \sqrt{2} - 1) - 4}{2 \pi} + O\Big(\frac{1}{R}\Big).
\label{I-3final}
\end{align}
Here, the double integral has been evaluated as follows. For the inner integral over $y_2$, we obtain from an integration by parts
\begin{align*}
\begin{split}
&\int_0^\infty \erfc(\sqrt{2}y_2) e^{(y_1 - y_2)^2}(1+2(y_1 - y_2)^2) \,dy_2\\
&=\left[ \erfc(\sqrt{2}y_2)(y_2-y_1)e^{(y_2-y_1)^2}\right]_0^\infty+
\frac{2\sqrt{2}}{\sqrt{\pi}}\int_0^\infty e^{-2y_2^2}(y_2-y_1)e^{(y_2-y_1)^2} \,dy_2
\\
&=e^{y_1^2}\Big(y_1+\frac{2\sqrt{2
}}{\sqrt{\pi}}e^{y_1^2}\Big( \frac{1}{2}e^{-y_1^2}-y_1\sqrt{\pi}\erfc(y_1)\Big)\Big) \;.
\end{split}
\end{align*}
The remaining integral over $y_1$ contains three parts:
\begin{align}\label{3terms}
    \int_0^\infty\erfc(\sqrt{2}y_1) e^{y_1^2}\Big(y_1+\sqrt{\frac{2}{\pi}}-2\sqrt{2}y_1\erfc(y_1)e^{y_1^2}\Big)\,dy_1.
\end{align}
For the first term we obtain
\cite[6.238.1]{gradshteyn2014table}
\begin{equation}\label{firstterm}
    \int_0^\infty\erfc(\sqrt{2}y_1) e^{y_1^2}y_1\,dy_1=-\frac{1}{2}+\frac{\sqrt{2}}{2}.
\end{equation}
For the third term we have 
\begin{align*}
&-2\sqrt{2}\int_0^\infty\erfc(\sqrt{2}y_1)\erfc(y_1) e^{2y_1^2}y_1\,dy_1\\
&=
-2\sqrt{2}\left[ \erfc(\sqrt{2}y_1)\erfc(y_1)\frac{1}{4}e^{2y_1^2}\right]_0^\infty+2\sqrt{2}\int_0^\infty \Big( 
\erfc(\sqrt{2}y_1)\frac{-2}{\sqrt{\pi}}e^{-y_1^2}-\frac{2\sqrt{2}}{\sqrt{\pi}}e^{2y_1^2}\erfc(y_1)
\Big)\frac{1}{4}e^{2y_1^2}\,dy_1
\\
&=\frac{\sqrt{2}}{2}-\frac{\sqrt{2}}{\sqrt{\pi}}\int_0^\infty 
\erfc(\sqrt{2}y_1)e^{y_1^2}\,dy_1-\frac{2}{\sqrt{\pi}}\int_0^\infty 
\erfc(y_1)\,dy_1.
\end{align*}
Here, the second integral cancels the second term in \eqref{3terms}, and the last integral can be evaluated using \cite[6.281.1]{gradshteyn2014table}, to give $1/\sqrt{\pi}$. Putting all together yields the value given in \eqref{I-3final}.

Similarly, it follows for $I_{-,4}(R)$ from 
$$
a - b = O\Big(\frac{1}{R}\Big), \qquad \erf(a-b) = O\Big(\frac{1}{R}\Big),
$$
that we can estimate $I_{-,4}(R) = O(1/R^2)$ without further computation. We have shown \eqref{I -34 asymp}. 

It remains to show \eqref{I +12 asymp} and \eqref{I +34 asymp}. 
For $I_{+,1}(R)$, using
$$
\exp((y_1+y_2)^2 - (a+b)^2) \leq \exp(3(y_1^2+y_2^2) - 2R^2), \qquad 1+(y_1+y_2)^2 \leq 1+4R^2, 
$$
one can decouple the integrals
\begin{align}
\begin{split}\label{I+1R}
|I_{+,1}(R)| &\leq \frac{1}{2\pi} \int_0^R \int_0^R \erfc(\sqrt{2}y_1) \erfc(\sqrt{2}y_2) e^{3(y_1^2+y_2^2)-2R^2} (1+4R^2) \,dy_2 \,dy_1 \\
& = \frac{1+4R^2}{2\pi} \Big( e^{-R^2} \int_0^R \erfc(\sqrt{2} y) e^{3y^2} \,dy \Big)^2.
\end{split}
\end{align}
When $R$ is large, the function $y \mapsto \erfc(\sqrt{2} y) e^{3y^2}$ reaches its maximum value within the interval $[0,R]$ at $y=R$. 
A saddle-point approximation for $f(y)=\log[\erfc(\sqrt{2}y)]+3y^2$, with $f^{\prime}(y)\sim 2y$ for large $y$ thus gives
\begin{align*}
\int_0^R \erfc(\sqrt{2} y) e^{3y^2} \,dy 
&\sim \frac{e^{R^2}}{\sqrt{2\pi}R}\int_0^R e^{(y-R)2R}\,dy
=\frac{e^{R^2}}{\sqrt{2\pi}R}\frac{(1-e^{-2R^2})}{2R}.
\end{align*}
Inserted into \eqref{I+1R} this demonstrates that $I_{+,1}(R)=O(1/R^2)$.

For $I_{+,2}(R)$, $I_{+,3}(R)$, and $I_{+,4}(R)$, it is convenient to use light-cone coordinates 
$$
\begin{cases}
 u  = y_1 - y_2,
\\
 v  = y_1 + y_2, 
\end{cases} \qquad  \textup{i.e.} \qquad 
\begin{cases}
y_1= \frac{1}{2} (v + u) ,
\\
y_2  = \frac{1}{2} (v - u). 
\end{cases}
$$
From these transformations, we have $dy_2 \,dy_1 = \frac{R}{2} du \,ds$, where the integration domains are given by  
\begin{equation}
\label{lightcone set}
    s \in [0, 2], \qquad u \in \Big[R \abs{s - 1} - R, R - R \abs{s - 1}\Big]. 
\end{equation}
Note that as $R \to \infty$, we have
\begin{align*}
\erfc\Big( \frac{Rs+u}{\sqrt{2}} \Big) \erfc\Big( \frac{Rs-u}{\sqrt{2}} \Big) e^{R^2 s^2} (1 + c R^2 s^2) = \frac{2 c }{\pi} e^{-u^2} +O\Big(\frac{1}{R^2}\Big)
\end{align*}
and 
\begin{align*}
 a + b &= \sqrt{R^2-\Big( \frac{Rs+u}{2} \Big)^2} + \sqrt{R^2-\Big( \frac{Rs-u}{2} \Big)^2} = \sqrt{4-s^2} \, R+O\Big(\frac{1}{R}\Big),
 \\
a - b 
&= -\frac{us}{\sqrt{4-s^2}}+O\Big(\frac{1}{R^2}\Big).
\end{align*}
Using this, the $s$-integral in $I_{+,2}(R)$ is dominated by contributions from $0 < s < 2$. Therefore, $|s-1|<1$ and thus we can set the bounds of the $u$-integral in \eqref{lightcone set} to $\pm\infty$ and arrive at the following integral which is solvable:
\begin{align}
I_{+,2}(R)  = -\frac{R}{2 \pi^2} \int_0^2 \int_{-\infty}^{\infty} e^{-\frac{4}{4-s^2} u^2} \,du \,ds + o(R)
= -\frac{R}{4 \sqrt{\pi}} + o(R).
\end{align}

For $I_{+,3}(R)$, following the same strategy as for $I_{+,2}(R)$, and together with $\erf(a + b) = 1 + O(e^{-R^2})$, we obtain the leading order asymptotic behaviour 
\begin{equation}
    I_{+,3}(R) \sim \frac{R^2}{2\pi^{3/2}} \int_0^2 \int_{-\infty}^{\infty} \sqrt{4-s^2} \, e^{-u^2} \,du \,ds = \frac{R^2}{2}.
\end{equation}
To derive the subleading terms we write 
$$
I_{+,3}(R) = \frac{R^2}{2} - I_{+,3}^{A}(R) - I_{+,3}^{B}(R) + O(1),
$$ 
where we add and subtract terms, such that the limits of the following two integrals can be obtained more easily:
\begin{align*}
I_{+,3}^{A}(R) & = \frac{R^2}{2} - \frac{R}{8\sqrt{\pi}} \int_0^{2R} \int_{\abs{v-R}-R}^{R-\abs{v-R}} \frac{4}{\pi} e^{-u^2} \sqrt{4-\frac{v^2}{R^2}} \,du \,dv,
\\
I_{+,3}^{B}(R) & = \frac{R}{8\sqrt{\pi}} \int_0^{2R} \int_{\abs{v-R}-R}^{R-\abs{v-R}} \Big[ \frac{4}{\pi} e^{-u^2} - \erfc\Big( \frac{v+u}{\sqrt{2}} \Big) \erfc\Big( \frac{v-u}{\sqrt{2}} \Big) e^{v^2} (1+2 v^2) \Big] \sqrt{4-\frac{v^2}{R^2}} \,du \,dv.
\end{align*}
In the last integral we have switched back to the unscaled $v$-coordinate.
Note that the $v$-integral in $I_{+,3}^{A}(R)$ is dominated by contributions from small $v$, and we obtain from the $u$-integral 
\begin{align}
\begin{split}
I_{+,3}^{A}(R) & = 
\frac{R^2}{2} - \frac{R}{2\pi} \int_0^{2R} \erf(R-|v-R|)
\sqrt{4-\frac{v^2}{R^2}} \,dv\\
&=
\frac{R}{2\pi} \left(\int_0^{R}+\int_R^{2R}\right) \erfc(R-\abs{v-R}) \sqrt{4-\frac{v^2}{R^2}} \,dv
\\
& = \frac{R}{2\pi} \int_0^{\infty} \erfc(v) \Big( 2 + O\Big(\frac{1}{R^2}\Big) \Big) \,dv + o(R)  = \frac{1}{\pi^{3/2}} R + o(R).
\end{split}
\end{align}
In the second step we can neglect the second integral on $[R,2R]$, as the error function then leads to an exponential suppression. This yields the last line, giving a contribution to the linear order in $R$.
Similarly, the $v$-integral in $I_{+,3}^{B}(R)$ is dominated by small $v$.
The limiting integral is easier to compute in the original $(y_1,y_2)$-coordinates where we use $\sqrt{4-(y_1+y_2)^2/R^2} = 2 + O(\frac{1}{R^2})$ to derive
$$
I_{+,3}^{B}(R)  = \frac{R}{2\sqrt{\pi}} \int_0^\infty \int_0^\infty
\Big[ \frac{4}{\pi} e^{-(y_1-y_2)^2} - \erfc(\sqrt{2}y_1) \erfc(\sqrt{2}y_2) e^{(y_1 + y_2)^2} (1+2(y_1 + y_2)^2) \Big] \,dy_2 \,dy_1 + o(R) .
$$
We evaluate the double integral analogously to $I_{-,3}(R)$ above.
The integral over the first term gives
\begin{equation*}
\frac{4}{\pi} \int_0^\infty e^{-(y_1-y_2)^2} \,dy_2 = \frac{2}{\sqrt{\pi}} \erfc(-y_1).
\end{equation*}
Via integration by parts we obtain for the second term
\begin{align*}
\int_0^\infty \erfc(\sqrt{2}y_2) e^{(y_1 + y_2)^2} (1+2(y_1 + y_2)^2) \,dy_2 
&
=-y_1e^{y_1^2}
+\int_0^\infty\frac{2\sqrt{2}}{\sqrt{\pi}}e^{-2y_2^2}(y_1+y_2)e^{(y_1 + y_2)^2} \,dy_2\\
&=-y_1e^{y_1^2} +2\sqrt{2}y_1e^{2y_1^2}\erfc(-y_1)+\sqrt{\frac{2}{\pi}}e^{y_1^2}.
\end{align*}
The integral over $y_1$ now contains four parts
\begin{equation}\label{I+3-4parts}
\begin{split}
I_{+,3}^{B}(R) = \frac{R}{2\sqrt{\pi}} \int_0^\infty &\Big[
 \frac{2}{\sqrt{\pi}}\erfc(-y_1)+   
    \erfc(\sqrt{2}y_1) y_1e^{y_1^2}
- \erfc(\sqrt{2}y_1)2\sqrt{2}y_1e^{2y_1^2}\erfc(-y_1)\\
&\quad
-\erfc(\sqrt{2}y_1)\sqrt{\frac{2}{\pi}}e^{y_1^2}\Big]\,dy_1.
\end{split}
\end{equation}
For the third part of the integral in the first line we have, again after integration by parts
\begin{align*}
&-\int_0^\infty
2\sqrt{2}y_1e^{2y_1^2}\erfc(\sqrt{2}y_1)\erfc(-y_1)\,dy_1
\\
&=\frac{\sqrt{2}}{2} 
+\int_0^\infty \frac{\sqrt{2}}{2}e^{2y_1^2}\Big(
-\frac{2\sqrt{2}}{\sqrt{\pi}}e^{-2y_1^2}\erfc(-y_1)+\erfc(\sqrt{2}y_1)\frac{2}{\sqrt{\pi}}e^{-y_1^2}
\Big)\, dy_1.
\end{align*}
The first integral cancels the first term in line one of \eqref{I+3-4parts} and the second integral cancels the fourth term in line two of \eqref{I+3-4parts}. We are left with the second term in line one of \eqref{I+3-4parts}, which follows from \eqref{firstterm}.
We thus arrive at
\begin{equation}
I_{+,3}^{B}(R) = \frac{R}{2\sqrt{\pi}} \Big( -\frac{1}{2}+\sqrt{2}\Big) 
+ o(R).
\end{equation}
Combining the leading and subleading terms, we obtain the asymptotic formula for $I_{+,3}(R)$ in \eqref{I +34 asymp}.

Finally, for $I_{+,4}(R)$ we use that as $R \to \infty,$
\begin{equation*}
a - b = -\frac{us}{\sqrt{4-s^2}} +O\Big(\frac{1}{R^2}\Big), \qquad
\erf(a - b) = \erf\Big( -\frac{us}{\sqrt{4-s^2}} \Big) +O\Big(\frac{1}{R^2}\Big).
\end{equation*}
Since the $I_{+,4}(R)$ integral is dominated by $0 < s < 2$, as before we set the bounds of the $u$-integral to $\pm\infty$, do an integration by parts and conclude
\begin{align}\begin{split}
I_{+,4}(R) & = -\frac{R}{2 \pi^{3/2}} \int_0^2 \frac{s}{\sqrt{4-s^2}} \int_{-\infty}^{\infty} u e^{-u^2} \erf\Big( \frac{us}{\sqrt{4-s^2}} \Big) \,du \,ds + o(R)
\\
& = -\frac{R}{\pi^2} \int_0^2 \frac{s^2}{4-s^2} \int_{-\infty}^{\infty}  e^{-\frac{u^2}{4-s^2}}\,du\,ds + o(R)
\\
&=-\frac{R}{4\pi^{\frac{3}{2}}} \int_0^2 \frac{s^2}{\sqrt{4-s^2}} \,ds + o(R)
= -\frac{R}{4 \sqrt{\pi}} + o(R).
\end{split}
\end{align}
This completes the proof. 
\end{proof}

We are now ready to show Theorem~\ref{Thm_GinOE inf}.

\begin{proof}[Proof of Theorem~\ref{Thm_GinOE inf}]
Eq.~\eqref{Var NR DR inf} follows from Proposition~\ref{Prop_Expression var inf} (i) by series expansion. 

For \eqref{Cov prop4.3}, recalling $\hat{a}=\sqrt{R^2-v^2}$, 
we have as $R \to \infty,$
\begin{align*}
&\quad \int_{0}^R  \erfc(\sqrt{2}v)   v \,e^{v^2} \Big( e^{-(\hat{a} - R)^2}-e^{-(\hat{a} + R)^2}+ \sqrt{\pi}(\hat{a} - R) \erf(\hat{a} - R) -\sqrt{\pi} (\hat{a} + R) \erf(\hat{a} + R) \Big)\,dv
\\
& \sim -\sqrt{\pi} \int_{0}^R  \erfc(\sqrt{2}v)   v \,e^{v^2}  (\hat{a} + R) \erf(\hat{a} + R) \,dv  \sim  -2\sqrt{\pi} R \int_{0}^R  \erfc(\sqrt{2}v)   v \,e^{v^2} \,dv.
\end{align*}
Then by using Proposition~\ref{Prop_Expression var inf} (ii) and 
\eqref{firstterm} 
we obtain Eq.~\eqref{GinOE inf Cov asy}, 
\begin{align*}
 \Cov \Big( \NN_\C^{(1)}(D_R),   \NN_\R^{(1)}(D_R) \Big) \sim -2 \frac{ \sqrt{2}-1 }{\sqrt{\pi}}  R. 
\end{align*}

Finally, we show Eq.~\eqref{Var NC DR inf}.
It follows from Lemma~\ref{Lem_integrals GinOE c} that 
\begin{equation}
4 \Big( I_-(R) - I_+(R) \Big) = -2 R^2 + 4\sqrt{\frac{2}{\pi}} R + o(R).
\end{equation}
On the one hand, by Proposition~\ref{Prop_EN inf}, 
\begin{equation}
2 E_\C^{(1)}(R) = 2 R^2-2 \sqrt{ \frac{2}{\pi} } R +O(1) , \qquad R \to \infty.
\end{equation}
Combining this with Proposition~\ref{Prop_Expression var inf} (iii), we obtain 
\begin{equation}
\Var \NN_\C^{(1)} (D_R)  \sim 2 \sqrt{ \frac{2}{\pi} } R. 
\end{equation}
Eq.~\eqref{linear behaviour real} follows simply from \eqref{VarDef1}. 
\end{proof}

\section{Full counting statistics of planar symplectic ensembles}\label{sec:proofFCS}


Recall that $W \equiv W_N$ is a (possibly $N$-dependent) rotationally invariant potential that satisfies the conditions in Definition~\ref{Qsuit} and that $g: [0, \infty) \to \R$ with $W(z) = g(\abs{z})$.

\begin{proof}[Proof of Proposition~\ref{Prop_cumulant symplectic} (ii)]
We start by writing out the Laplace transform of $\NN_a $ as
\begin{align*}
\mathbb{E}_{N,W}^{(\beta=4)}\Big[ e^{u \NN_a} \Big]
&= \int_{\C^N} e^{u \sum_{k = 1}^{N} \chi_a(z_k)} \,d\mathcal{P}_{N,W}^{(4)} (\bfs{z}) 
\\
&=  \frac{1}{Z_{N,W}^{(\beta=4)}} \int_{\C^N} 
\det\begin{bmatrix} z_l^{k - 1} \\ \overline{z}_l^{k - 1} \end{bmatrix}_{k \leq 2 N, l \leq N} \, \prod_{j=1}^N (z_j-\overline{z}_j) e^{u \chi_a(z_j) -N W(z_j)} \,dA(z_j),
\end{align*}
where $\chi_a(z) = 1$ if $|z| \le a$ and $\chi_a(z) = 0$ otherwise.
Next, we simplify the $N$-fold integration via de Bruijn's formula, 
\begin{equation} \label{cumulant Pf}
\mathbb{E}_{N,W}^{(\beta=4)}\Big[ e^{ u \NN_a } \Big]
=  \frac{ N! }{Z_{N,W}^{(\beta=4)}}  \, \Pf\Big[ \int_{\C} (z^{k - 1} \overline{z}^{l - 1} - z^{l - 1} \overline{z}^{k - 1}) (z-\overline{z}) e^{u \chi_a(z) -N W(z)} \,dA(z) \Big]_{k, l =1}^{2N},
\end{equation}
see e.g. \cite[Remark 2.5]{akemann2021skew}. 
Due to the radial symmetry of $W$, we can further simplify the integrals in the Pfaffian, see also \cite{Mehta,rider2004order}.
Namely, by multiplying out $(z^{k - 1} \overline{z}^{l - 1} - z^{l - 1} \overline{z}^{k - 1}) (z-\overline{z})$ and using that $e^{u \chi_a(z) -N W(z)}$ only depends on the radius $|z|$, we obtain for the entries
\begin{align}
\begin{split}
&\quad \int_{\C} (z^{k - 1} \overline{z}^{l - 1} - z^{l - 1} \overline{z}^{k - 1}) (z-\overline{z}) e^{u \chi_a(z) -N W(z)} \,dA(z) 
\\
&= 2 \int_0^\infty (r^k r^{l - 1} \delta_{k, l - 1} - r^{k - 1} r^l \delta_{k - 1, l} - r^l r^{k - 1} \delta_{k - 1, l} + r^{l - 1} r^k \delta_{k, l - 1}) e^{u \chi_a(r) -N g(r)} r \,dr 
\\
&= 4 (\delta_{k, l - 1} - \delta_{k, l + 1}) \int_0^\infty r^{k + l} e^{u \chi_a(r) -N g(r)} \,dr. \label{entries in Pf}
\end{split}
\end{align}
Since $e^{u \chi_a(r)} = 1 - (1 - e^{u}) \chi_a(r)$, we obtain for the integral
\begin{align}
\begin{split} \label{int in Pf}
\int_0^\infty r^{k + l} e^{u \chi_a(r) -N g(r)} \,dr
& = \int_0^\infty r^{k + l} e^{-N g(r)} \,dr - (1 - e^{u}) \int_0^a r^{k + l} e^{-N g(r)} \,dr \\
& = \frac{1}{2} h_{\frac{k + l - 1}{2}} - (1 - e^{u}) \frac{1}{2} h_{\frac{k + l - 1}{2},1}(a).
\end{split}
\end{align}
Note that the indices are integer-valued for $k = l - 1$ and $k = l + 1$.
Combining \eqref{cumulant Pf}, \eqref{entries in Pf} and \eqref{int in Pf}, we arrive at 
\begin{equation}
\mathbb{E}_{N,W}^{(\beta=4)}\Big[ e^{ u \NN_a } \Big]
= \frac{N!}{Z_{N,W}^{(\beta=4)}}  \Pf\Big[ 2 (\delta_{k, l - 1} - \delta_{k, l + 1}) \big( h_{\frac{k + l - 1}{2}} - (1 - e^{u}) h_{\frac{k + l - 1}{2},1}(a) \big) \Big]_{k, l =1}^{2N}.
\end{equation}
The Kronecker deltas inside the Pfaffian lead to an antisymmetric matrix whose only non-zero entries appear above and below the main diagonal.
The Pfaffian of such a matrix is the product of every second entry above the main diagonal (i.e.\ $l = k - 1$ and then only take the even $k$):
\begin{align*}
&\quad \Pf\Big[ 2 (\delta_{k, l - 1} - \delta_{k, l + 1}) \big( h_{\frac{k + l - 1}{2}} - (1 - e^{u}) h_{\frac{k + l - 1}{2},1}(a) \big) \Big]_{k, l =1}^{2N}
\\
&= \prod_{ \substack{k=2 \\ k:\textup{ even}} }^{2 N} 2 \big( h_{k-1} - (1 - e^{u}) h_{k-1,1}(a) \big) = 2^N \prod_{j = 0}^{N-1} \big( h_{2j+1} - (1 - e^{u}) h_{2j+1,1}(a) \big).
\end{align*}
On the other hand, it follows from \cite[Remark 2.5 and Corollary 3.3]{akemann2021skew} that 
\begin{equation}
Z_{N,W}^{(\beta=4)} = N! \, 2^N \prod_{j = 0}^{N-2} h_{2 j+ 1}. 
\end{equation}
This leads to 
\begin{equation}
\mathbb{E}_{N,W}^{(\beta=4)}\Big[ e^{ u \NN_a } \Big]= \prod_{j = 0}^{N-1} \Big( 1 - (1 - e^{u}) \frac{h_{2j+1,1}(a)}{h_{2j+1}} \Big).
\end{equation}
Using $h_{j} - h_{j,1}(a) = h_{j,2}(a)$, we obtain the desired identity \eqref{eq:CGF symplectic}.

The expression \eqref{cumulants symplectic} follows along the lines of the supplementary material in \cite{lacroix2019rotating} using the series expansion 
\begin{equation}
\log \mathbb{E}_{N,W}^{(\beta=4)}\Big[ e^{ u \NN_a } \Big] =u \, E_{N,W}^{(\beta=4)}(a)+ \sum_{p = 2}^{\infty} \kappa_p^{(\beta=4)}(a) \,  \frac{u^p}{p!}. 
\end{equation}
Note that the expression \eqref{mean value symplectic} was obtained in \cite[Proposition 1.1]{akemann2022universality}.
\end{proof}


\begin{proof}[Proof of Theorem~\ref{Thm_higher cumulants}]
In \cite[Subsection 4.1]{akemann2022universality}, it was shown that if $g$ satisfies the conditions in Definition~\ref{Qsuit}, then in the bulk $a\in (0,1)$ the function $\mathcal{L}_k(a)$ has the large-$N$ asymptotic
\begin{equation}
\mathcal{L}_k(a) \sim \frac{1}{2} \erfc\Big( \frac{s}{\sqrt{2 \Delta W(a)} \, a} \Big), \qquad k = \frac{1}{2}(N a g'(a) - 1) + \sqrt{N} s,
\end{equation}
and is exponentially suppressed for $|k-\frac{1}{2}(N a g'(a) - 1)|>M\sqrt{N}$ for large $M$.
Hence for $k = 2j+1$ we get
\begin{equation}
\mathcal{L}_{2j+1}(a) \sim \frac{1}{2} \erfc\Big( \frac{2 j - \frac{1}{2} N g'(a) a + \frac{3}{2}}{\sqrt{2 \Delta W(a) N} \, a} \Big).
\end{equation}
Now we insert the asymptotic behavior of $\mathcal{L}_{2j+1}(a)$ into the finite-$N$ cumulant formula \eqref{cumulants symplectic}. 
Note that the argument of the polylogarithm is then of the form 
\begin{equation}
1 - \frac{1}{\frac{1}{2} \erfc(x)} = \frac{\erfc(x) - 2}{\erfc(x)} = -\frac{\erfc(-x)}{\erfc(x)}, \quad
\text{ with } x = \frac{2 j - \frac{1}{2} N g'(a) a + \frac{3}{2}}{\sqrt{2 \Delta W(a) N} \, a}.
\end{equation}
Next we replace the sum over $j$ in the range 
$|2j+1-\frac{1}{2}(N a g'(a) - 1)|\leq\sqrt{N} s$ that contributes to the sum 
by a Riemann integral over $x$, and we obtain
\begin{equation}
    \kappa_p^{(\beta=4)}(a) \sim (-1)^{p + 1} \frac{\sqrt{2 \Delta W(a) N} \, a}{2} \int_{-\infty}^{\infty} \Li_{1 - p}\Big( -\frac{\erfc(-x)}{\erfc(x)} \Big) \, dx
\end{equation}
as $N \to \infty$.
Note that the prefactor comes from the change of integration variable $j \to x$. 
The parity in \eqref{kappa_p} comes from that of the integrand under $x\to-x$, using the inversion formula \eqref{Li-inversion}.

Similarly, for the edge case when $a$ is given by \eqref{a edge scaling}, it follows that
\begin{equation}
    \kappa_p^{(\beta=4)}(a) \sim (-1)^{p + 1} \frac{\sqrt{2 \Delta W(a) N} \, a}{2} \int_{-\infty}^{\SS} \Li_{1 - p}\Big( -\frac{\erfc(-x)}{\erfc(x)} \Big) \, dx
\end{equation}
as $N \to \infty$.
Since we sum up to $j = N - 1$, the integration will stop at some value $x < \infty$ now, which can be derived as in \cite[Subsection 4.1]{akemann2022universality}.
Note that we use $4 + g''(1) = 4 \Delta W(1)$ after we expand the $g'(a)$-term.
\end{proof}

\section{Conclusions and open questions}\label{sec:conc}

In this paper we have analysed the counting statistics of the number of eigenvalues in a centred disc of radius $R$ in the real and symplectic Ginibre ensembles. For the symplectic ensemble we determined the full counting statistics by computing the generating function for all cumulants at finite matrix dimension $N$, for a large class of rotationally invariant non-Gaussian potentials. In the large-$N$ limit we could prove that the symplectic ensemble is in the same universality class as its complex counterpart. This holds both in the bulk regime, when the radius is smaller but of the same order as the edge of the limiting support of the eigenvalues, and in the edge regime when choosing the radius in the vicinity of the edge.  

For the real Ginibre ensemble, the counting statistics is much more involved, as the correlations of its real and complex eigenvalues have to be considered separately, including their covariance. Therefore, we restricted ourselves to the Gaussian potential and computed the mean number of points at finite and large-$N$, as well as the variance in the large-$N$ limit, in both cases in the vicinity of the origin. In the large radius limit, leading into the bulk regime, we found a matching for both quantities with the universality class in the two other ensembles. In a small radius expansion, however, all three ensembles differ, reflecting their different repulsion of eigenvalues from the real axis. 
For the bulk and edge regime, based on numerical simulations we conjectured the same universality found for the symplectic ensemble to hold for the real ensemble. It remains an open problem to prove these conjectures, and to extend our computations to the full counting statistics for all cumulants within the real ensemble. {This would in principle give access to the full distribution of the number of eigenvalues in a disk of radius $R$. The large deviation regime of this distribution was studied using Coulomb gas techniques in Ref. \cite{allez2014index}. However, one expects that this large deviation regime does not carry any information about the cumulants, which as in the GinUE case \cite{lacroix2019intermediate}, are probably determined by some ``intermediate deviation regime'' of the distribution. Because of the intricate Pfaffian structure of the GinOE with two kinds of eigenvalues (complex and real) this remains a challenging task.} 

Our analysis was driven by the question of universality of the counting statistics, within the three different ensembles and for non-Gaussian deformations of the weight function. While all three Ginibre ensembles enjoy an interpretation as a static Coulomb gas in two dimensions, only for the complex Ginibre ensemble we know of a map to a quantum Hamiltonian for fermions in a rotating harmonic trap, or equivalently in the presence of a magnetic field. It remains an open question if such a map exists also for the real and symplectic Ginibre ensembles studied here, which both lead to Pfaffian point processes.

\subsection*{Acknowledgements}
The work of Gernot Akemann was partly funded by the Deutsche Forschungsgemeinschaft (DFG) grant SFB 1283/2 2021 – 317210226.
Sung-Soo Byun was partially supported by Samsung Science and Technology Foundation (SSTF-BA1401-51) and by the POSCO TJ Park Foundation (POSCO Science Fellowship).
This work was initiated during the Fourth ZiF Summer School in August 2022, and we wish to express our gratitude for the hospitality of the Centre for Interdisciplinary Research ZiF.

\appendix

\section{Integrable structure of the GinOE: jpdf and origin limit}\label{App:A}

In this appendix we recall some aspects of the integrable structure of the GinOE in terms of Pfaffian determinants. 
First, we give the joint density of real and complex eigenvalues of the GinOE. This allows us to define the correlation functions of real, complex or both types of eigenvalues as in \eqref{RkDef}. Second, we provide the origin scaling of all three types of correlation functions as they will be needed in the computation of the number variance of the GinOE in this limit. 
For more details we refer to 
 \cite{byun2023progress} as well as to the original literature.

A real $N\times N$ matrix from the GinOE can have $k$ real eigenvalues $x_1,\dots,x_k\in\R$ and $2l$ complex (non-real) eigenvalues that come in complex conjugate pairs, $z_1,\dots,z_l,\bar{z}_1,\dots,\bar{z}_l\in\C\setminus \R$. Because of the sum $N=k+2l$,  $k$ has to be of the same parity as $N$, i.e. both are either even or odd, from here on. Consequently, $N-k$ is always even. 
For the probability $p_{N,k}$ that exactly $k$ out of $N$ eigenvalues are real, closed form expressions were given for instance in \cite{forrester2007eigenvalue,MR3563192}, see also \cite[Section 2.3]{byun2023progress}. 
Furthermore, since the eigenvalues in the lower half  plane are completely fixed by complex conjugation of those in the upper half plane, and in particular are not independent, we can restrict ourselves to $l$ complex eigenvalues in the upper half plane $\C_+$, with $\Im(z_1),\dots,\Im(z_l)>0$.
The partial joint density for such eigenvalues with given $N$ with fixed $k$ real and $l$ complex 
 eigenvalues was derived independently by Lehmann and Sommers \cite{lehmann1991eigenvalue} and Edelman \cite{edelman1997probability}:
\begin{equation} \label{jpdf-kl-1}
\begin{split}
\mathbf{P}_{N,k,l}^{(1)}(x_1,\dots,x_k,z_1,\dots,z_l)
&= C_{N,k,l}\prod_{i>j}^k|x_i-x_j| \prod_{i=1}^k\prod_{j=1}^l(x_i-z_j)(x_i-\bar{z}_j)\prod_{i>j}^l|z_i-z_j|^2|z_i-\bar{z}_j|^2 \\
&\quad \times \prod_{j=1}^li^{-1}(z_j-\bar{z}_j)\prod_{j=1}^k e^{-x_j^2/2}\prod_{j=1}^l\erfc\left(\frac{z_j-\bar{z}_j}{i\sqrt{2}}\right)e^{-\frac12 (z_j^2+\bar{z}_j^2)},
\end{split}
\end{equation}
where the normalisation constant reads
\begin{equation}
\label{CNkDef}
C_{N,k,l}:= \frac{2^{l-N(N+1)/4}}{k!\,l!\prod_{j=1}^{N}\Gamma(j/2)}.
\end{equation}
Note that in the second line of \eqref{jpdf-kl-1} we do not scale the exponent with $N$ (to be consistent with the references mentioned above), hence the limiting support is the unit disk with radius $\sqrt{N}$.
However for \eqref{RN1complex1} and \eqref{RN1real1} in the main part we use a scaling with $N$.
Note also that the first line in \eqref{jpdf-kl-1} together with the first factor in the second line is proportional to $|\Delta_{k+2l}(x_1,\dots,x_k,z_1,\dots,z_l,\bar{z}_1,\dots,\bar{z}_l)|$. That part of the joint density that only depends on the $z_1,\dots,z_l$ is reminiscent of the jpdf of the GinSE \eqref{jpdf4}. A Coulomb gas interpretation of the jpdf \eqref{jpdf-kl-1} at inverse temperature $2=1/(k_B T)$ was given in \cite{forrester2016analogies}.

The full joint density of $N$ eigenvalues $w_1,\dots,w_N\in\C_+\cup\R$, where it is not specified if they are real or complex,  is then given by summing over all sectors,
\begin{equation}
\mathcal{P}_{N}^{(1)}(w_1,\dots,w_N)=
\sum_{\substack{k=0\\N-k \text{ even}}}^N
\mathbf{P}_{N,k,(N-k)/2}^{(1)}(x_1,\dots,x_k,z_1,\dots,z_{(N-k)/2}).
\label{jpdf1}
\end{equation}
There are two equivalent ways of defining correlation functions of real, complex or mixed eigenvalues. The first introduced in \cite{sommers2007symplectic} and \cite{forrester2007eigenvalue} is to add extra source terms to the joint densities \eqref{jpdf-kl-1} to define generalised partition functions
\begin{equation}
\mathbf{Z}_{N,k,l}[u,v]:= \prod_{i=1}^k \int _{\R}dx_i u(x_i) \prod_{j=1}^l \int _{\C_+}dA(z_j) v(z_j)  \mathbf{P}_{N,k,l}^{(1)}(x_1,\dots,x_k,z_1,\dots,z_l),
\end{equation}
and then to take functional derivatives of the full generalised partition function 
\begin{equation}
\mathcal{Z}_N[u,v]:=
\sum_{\substack{k=0\\N-k \text{ even}}}^N
\mathbf{Z}_{N,k,(N-k)/2}[u,v]
\end{equation}
as follows:
\begin{align}
\mathbf{R}_{N,k,\R}^{(1)}(x_1,\dots,x_k)
&:= \frac{1}{\mathcal{Z}_N[1,1]}\frac{\delta^k}{\delta u(x_1) \cdots\delta u(x_k)} \left. \mathcal{Z}_N[u,1]\right|_{u=1}, \label{RkR1def}\\
\mathbf{R}_{N,l,\C}^{(1)}(z_1,\dots,z_l)
&:= \frac{1}{\mathcal{Z}_N[1,1]}\frac{\delta^k}{\delta v(z_1) \cdots\delta v(z_l)} \left. \mathcal{Z}_N[1,v]\right|_{v=1}, \label{RlC1def}\\
\mathbf{R}_{N,k,\R,l,\C}^{(1)}(x_1,\dots,x_{k},z_1,\dots,z_l)
&:= \frac{1}{\mathcal{Z}_N[1,1]}\frac{\delta^{k+l}}{\delta u(x_1) \cdots\delta u(x_k)\delta v(z_1) \cdots\delta v(z_l)} \left. \mathcal{Z}_N[u,v]\right|_{u=1,v=1}, \label{RkRlC1Def}
\end{align}
for the real-real, complex-complex, and real-complex correlation functions, respectively. Obviously, the last line contains the first, respectively second when setting $k=0$ or $l=0$, respectively.  
Alternatively, the following sum over $L$ complex eigenvalues with $l\leq L\leq (N-K)/2$ and $k\leq K=N-2L$ real eigenvalues, for a given $N$ with fixed $k$  and $l$, was used in \cite{MR2530159}
\begin{equation}
\begin{split}
\mathbf{R}_{N,k,\R,l,\C}^{(1)}(x_1,\dots,x_{k},z_1,\dots,z_l)
&=\frac{2^l}{\mathbf{Z}_{N,k,l}[1,1]}
\sum_{\substack{K=k\\N-K \text{ even}}}^N
\sum_{L=l}^{(N-K)/2}\frac{1}{(K-k)!\,(L-l)!}\\
&\times \prod_{i=1}^{K-k} \int _{\R}dx_i\prod_{j=1}^{L-l} \int _{\C_+}dA(z_j)  \mathbf{P}_{N,K,L}^{(1)}(x_1,\dots,x_K,z_1,\dots,z_L),
\end{split}
\end{equation}
where we only give the most general correlation function. 

All three types of correlation functions can be expressed as a Pfaffian determinant of a matrix valued skew-kernel. The latter is given in terms of polynomials that are skew-orthogonal with respect to a combination of an antisymmetric product on $\R^2$ and a similar product on $\C_+^2$, and we refer to \cite{forrester2007eigenvalue} for details. 
We will not reproduce them here as for finite-$N$ we only need the real and complex 1-point functions given in \eqref{RN1complex1} and \eqref{RN1real1}, in order to compute the mean number of eigenvalues in a (one- respectively two-dimensional) disc of radius $R$, see 
Proposition~\ref{Prop_EN GinOE} and Lemma~\ref{Prop_GinOE expected real} for $N$ even and odd.

Instead, we will directly present the resulting correlation functions  \cite{forrester2007eigenvalue,sommers2007symplectic,MR2530159,MR2439268} in the large-$N$ limit and in the origin scaling limit defined in \eqref{originRkC1}, \eqref{originRkR1} and \eqref{originR11}, that will be used in the computation of the limiting number variance: 

\begin{itemize}

\item \textbf{GinOE} complex $k$-point function in the origin limit
\begin{equation} \label{R k GinOE c}
R_{k,\C}^{(1)}(z_1,\cdots, z_k) =\prod_{j=1}^{k} (2\pi i \erfc(\sqrt{2}y_j)) \, \Pf \Big[
\begin{pmatrix} 
S_\C(\bar{z}_j,\bar{z}_l)   &  S_\C(\bar{z}_j,z_l) 
\\
S_\C(z_j,\bar{z}_l) & S_\C(z_j,z_l)
\end{pmatrix}  \Big]_{ j,l=1,\cdots k }, 
\end{equation}
where 
    \begin{equation} \label{S bulk}
   S_\C(z,w) := \frac{z-w}{2\sqrt{2\pi}} e^{-(z-w)^2/2},
    \end{equation}
    see e.g. \cite[Eq.~(6.21)]{MR2430570}. 
   \smallskip  
  \item \textbf{GinOE}  real $k$-point function  in the origin limit
\begin{equation} \label{R k GinOE r}
R_{k,\R}^{(1)}(x_1,\cdots, x_k) =\Pf \Big[
\begin{pmatrix} 
D_\R(x_j,x_l)   &  S_\R(x_j,x_l) 
\\
-S_\R(x_l,x_j) & \tilde{I}_\R(x_j,x_l)
\end{pmatrix}  \Big]_{ j,l=1,\cdots k },
\end{equation}
where 
\begin{equation}\label{SR def} 
S_\R(x,y):= \frac{1}{ \sqrt{2\pi} } e^{-(x-y)^2/2}
\end{equation}
and 
\begin{align}\label{DR def}
D_\R(x,y)&:= -\frac{\pa}{\pa y} S_\R(x,y) = - \frac{x-y}{\sqrt{2\pi}} e^{-(x-y)^2/2}, 
\\
\begin{split}
\label{IR def}
\tilde{I}_\R(x,y)&:= \int_{x}^{y} S_\R(t,y)\,dt +\frac12 \operatorname{sgn}(x-y)
= -\frac12 \erf \Big( \frac{x-y}{\sqrt{2}} \Big) +\frac12 \operatorname{sgn}(x-y), 
\end{split}
\end{align}
see e.g. \cite[Eq.~(2.36)]{byun2023progress}.
  \smallskip  
\item \textbf{GinOE} mixed real and complex $(k,l)$-point function in the origin limit \cite{MR2530159}
\begin{equation}
 \label{R k1k2 rc}
 R_{k,\R,l,\C}^{(1)}(x_1,\dots,x_{k},z_1,\dots,z_{l}) 
= \Pf \begin{bmatrix}
[\mathcal{K}_\R^{(1)}(x_j,x_l)]_{j,l=1}^{k} &
[\mathcal{K}_{\R,\C}^{(1)}(x_j,z_l)]_{\substack{j=1,\dots,k\\ l=1,\dots,l}} \\
\Big ( - [\mathcal{K}_{\R,\C}^{(1)}(x_j,z_l)]_{\substack{j=1,\dots,k\\ l=1,\dots,l}} \Big )^T &
[\mathcal{K}_{\C}^{(1)}(z_j,z_l)]_{j,l=1}^{l}
\end{bmatrix}
\end{equation}
where 
\begin{align}
&\mathcal{K}_{\C}^{(1)}(z,w):= 2\pi i (\erfc(\sqrt{2}\im z) \erfc(\sqrt{2} \im w) )^{1/2}  \begin{pmatrix} 
S_\C(\bar{z},\bar{w})   &  S_\C(\bar{z},w) 
\\
S_\C(z,\bar{w}) & S_\C(z,w)
\end{pmatrix},
\\
&\mathcal{K}_\R^{(1)}(x,y):= \begin{pmatrix} 
D_\R(x,y)   &  S_\R(x,y) 
\\
-S_\R(x,y) & \tilde{I}_\R(x,y)
\end{pmatrix}, 
\\
&\mathcal{K}_{\R,\C}^{(1)}(x,z) := \frac{1}{\sqrt{2}} (\erfc(\sqrt{2} \im z) )^{1/2}
\begin{pmatrix} (z-x) e^{-(x-z)^2/2} & i (\bar{z} - x) e^{-(x-\bar{z})^2/2} \\
- e^{-(x-z)^2/2} &
-i e^{-(x-\bar{z})^2/2} \end{pmatrix}.
\end{align}
\end{itemize}

Notice that in the large-$N$ limit in the bulk of the spectrum the GinOE becomes a DPP that agrees with the GinUE \cite{MR2530159}. The same holds true for the GinSE \cite{akemann2019universal}. However, as the integration domain for the computation of the number variance in the bulk limit includes the real axis, this does not automatically imply that the variances of all three ensembles agree as well.

\bibliographystyle{abbrv}
\bibliography{RMTbib}
\end{document}